%% file: main.tex
\documentclass{article}
\input{macros}
\usepackage{amssymb}
\usepackage[capitalize]{cleveref}
\title{XOR Lemmas for Communication via Marginal Information}
\author{Siddharth Iyer\thanks{Supported by NSF award 2131899.} \\ siyer@cs.washington.edu \and Anup Rao\footnotemark[1] \\ anuprao@cs.washington.edu}
\begin{document}
\maketitle
\abstract{
We define the  \emph{marginal information} of a communication protocol, and use it to prove XOR lemmas for  communication complexity.  We show that if every $C$-bit protocol has bounded advantage for computing a Boolean function $f$, then every $\tilde \Omega(C \sqrt{n})$-bit protocol has advantage $\exp(-\Omega(n))$ for computing the $n$-fold xor $f^{\oplus n}$. We prove exponentially small bounds in the average case setting,  and near optimal bounds for product distributions and for bounded-round protocols.
}
{
  \hypersetup{linkcolor=black}
  \tableofcontents
}
\section{Introduction}
If a function is hard to compute, is it even harder to compute it many times? This old question is often challenging, and  new answers are usually accompanied by foundational ideas. We give new answers in the framework of communication complexity,    accompanied by a new measure of  complexity  called  \emph{marginal information}. This definition provides a new tool for proving lower bounds in theoretical computer science.

A wide variety of important lower bounds in computer science ultimately rely on information theoretic lower bounds in communication complexity, including  lower bounds on the depth of monotone circuits  \cite{KRW},  lower bounds on data structures  \cite{P} and   lower bounds on the extension complexity of polytopes   \cite{BP16,Rothvoss, Sinha,JSY}, to name a few nice examples. We refer the reader to the textbook \cite{RaoYehudayoff-text} for an introduction to the basic definitions and concepts in communication complexity, the role played by the questions we address here, and the connections to other areas.

Given a Boolean function $f:\mathcal{X}\times \mathcal{Y} \rightarrow \{0,1\}$, define the functions $f^n: \mathcal{X}^n \times \mathcal{Y}^n \rightarrow \{0,1\}^n$ and  $f^{\oplus n}:\mathcal{X}^n \times \mathcal{Y}^n \rightarrow \{0,1\}$ as follows\footnote{Throughout, we drop the delimiters between variables. $f(xy)$ is to be read as $f(x,y)$.}:
\begin{align*}
    f^n(xy) &= (f(x_1y_1),f(x_2y_2),\ldots,f(x_ny_n)),\\
    f^{\oplus n}(xy) & = f(x_1y_1) \oplus f(x_2y_2) \oplus \dotsb \oplus f(x_ny_n).
\end{align*}
So, $f^n$ computes $f$ on $n$ different pairs of inputs, and $f^{\oplus n}$ computes the parity of the outputs of $f^n$. If $f$ is hard to compute, are $f^n$ and $f^{\oplus n}$ even harder to compute? For deterministic communication complexity, Feder, Kushilevitz, Naor and Nisan  \cite{FKNN}  proved that if $|\mathcal{X}|,|\mathcal{Y}| \leq 2^\ell$ and $f$ requires $C$ bits of communication, then $f^n$ requires at least $n (\sqrt{C} - \log_2 \ell -1)$ bits of communication. In this work, we study randomized communication complexity. Let $\|\pi \|$  denote the communication complexity of a randomized communication protocol $\pi$ and  define the advantage:
\begin{align*}
    \adv(C,f) &= \sup_{\|\pi \| \leq C} \inf_{xy} \E[(-1)^{\pi(xy)+f(xy)}].
\end{align*}
This quantity measures the best worst-case advantage achievable by a $C$-bit protocol  over random guessing. We can now state our main result:
\begin{theorem}\label{maintheorem}     There is a universal constant $\kappa >0$ such that 
    if $C >1/\kappa$ and  $\adv(C,f) < 1/2$, then  $$\adv\Big(\frac{\kappa C  \sqrt{n}}{ \log(Cn)},f^{\oplus n}\Big) < \exp(-\kappa n).$$
\end{theorem}
The constant $1/2$ is not important, it can be replaced by any constant less than $1$. Some assumption of the type $C> 1/\kappa$ is necessary, because if $x,y \in \{0,1\}$ and $f(xy) = x \oplus y$, then   $\adv(1,f) =0$, yet $\adv(2,f^{\oplus n})=1$. Prior to our work, the best known upper bound was proved by the second author with Barak, Braverman and Chen  \cite{BBCR}, who showed that the advantage is bounded by $1/2$ for a similar choice of the other parameters. Our work builds on the work of Yu \cite{huacheng}, who proved exponentially small bounds on the advantage in the setting of bounded-round communication protocols. 

Our ideas lead to many  results similar to \Cref{maintheorem}.  Next, we review the history that led us to the notion of marginal information, explain the intuitions behind the choices made in the definition, and then describe all of our results in \Cref{usingmarginal}.

\subsection{The evolution of information complexity}
Marginal information is the most recent advance in an evolution of definitions about information. We relate bounds on the communication and advantage for computing $f$ to the corresponding parameters for $f^{\oplus n}$ via a scheme that has been applied many times before. We prove:
\begin{description}
     \item [Step 1] Every protocol computing $f^{\oplus n}$ with significant advantage and small communication has small marginal information; see \Cref{thm:construct-rect-dist}.
    \item [Step 2] Marginal information is subadditive, so the marginal information for computing $f$ is smaller by a factor of $n$; see \Cref{lem:single-shot-marg-bound}. 
    \item [Step 3] Small marginal information can be compressed to give protocols with small communication; see \Cref{simulationthoerem,simulationtheorem2,thm:bounded-round-sim,thm:braverman-sim}.
\end{description}
Definitions of information are famously subtle. In order to make this strategy work, the marginal information needs to permit all 3 steps, and even minor changes to the definition can make one of the steps infeasible. 
 
 Our current definition builds on important insights and intuitions developed in theoretical computer science over a period of decades. An early precursor to the use of  information theory in computer science is the work of  Kalyanasundaram and Schnitger, who used Kolmogorov complexity to prove lower bounds on the randomized communication complexity of the disjointness function \cite{KalyanasundaramSchnitger87}. The proof was subsequently simplified by Razborov \cite{Razborov92}, who gave a beautiful short argument that used Shannon's notion of entropy \cite{Shannon} and implicitly followed  the outline of the steps 1,2,3 described above. This is related to the  questions we study here because  the disjointness function can be thought of as a way to compute the AND of $2$ bits $n$ times. Step 1 is relatively easy for this problem. Step 2 involved a clever way to split the dependence between random variables, and was accomplished using the subadditivity of entropy. Step 3 is also not too difficult.

The next chapter of the story was written during the study of  parallel repetition, a vital tool in the development of probabilistically checkable proofs. Raz \cite{Raz}  proved the first exponentially small bounds in this context using the Kullback-Liebler divergence as a measure of information. Given a distribution $p(xy)$, and a carefully chosen  event $W$, Raz  measured the divergence
\begin{align} &\E_{p(xy\vert W)}\Big [\mathsf{D}(p(x\vert yW)||p(x\vert y))+\mathsf{D}(p(y\vert xW)||p(y\vert x))\Big]\notag \\&=\E_{p(xy\vert W)}\Big [\log \Big(\frac{p(x|yW)}{p(x\vert y)}\cdot \frac{p(y|xW)}{p(y\vert x)}\Big)\Big].\label{divergence}
\end{align} 
In the proof, it is crucial that the event $W$ is \emph{rectangular}, meaning that if $x,y$ are independent, then they remain independent even after conditioning on $W$. Once again, Step 1 is not too difficult. Raz used the subadditivity of divergence and a similar set of clever random variables as in \cite{Razborov92} to split the dependence and accomplish Step 2. Later, Holenstein \cite{Holenstein} introduced a method called  \emph{correlated sampling} to simplify the analogue of Step 3 in Raz's proof, and obtained better bounds.  The second author used these tools to prove optimal bounds for parallel repetition in the setting relevant to probabilistically checkable proofs \cite{Rao}. 

Chakrabarti, Shi, Wirth and Yao \cite{CSWY} were the first to propose using general measures of  information complexity  to address the questions we consider in this paper. Let $xy$ denote the inputs, $m$ denote the public randomness and transcript of a communication protocol and $p(xym)$ denote the joint distribution induced by the  protocol\footnote{We often say $p(xym)$ is a protocol when we mean that it is a distribution induced by a protocol.}. \cite{CSWY} proposed to measure the mutual information
\begin{align*}
\mathsf{I}(M:XY)=
\E_{p(xym)}\Big[\log \frac{p(xy\vert m)}{p(xy)}\Big].
\end{align*}
Years later, this measure was renamed  \emph{external information} by \cite{BBCR}. The external information measures the information learned by an external observer  about the parties' inputs. Step 1 is easy for this measure of information. However, the subadditivity of Step 2 does not hold in general; the proof only goes through when the input distribution $p(xy)$ is a product distribution.  Jain, Radhakrishnan and Sen \cite{JRS}, and Harsha, Jain, McAllester and Radhakrishnan \cite{HJMR} gave ways to implement Step 3 that led to bounds on the success probability for computing $f^n$ in the setting where the inputs are assumed to come from a product distribution and the communication protocols are restricted to having a bounded number of rounds. Meanwhile, Bar-yossef, Jayram, Kumar and Sivakumar \cite{BJKS} showed how to reframe Razborov's proof using mutual information instead of entropy, and proved other results using this formulation which contained hints of the definition of information that came next.

The first upper bounds on the success probability in the  general setting  came when the second author together with Barak, Braverman and Chen \cite{BBCR} adapted the methods developed  in the study of parallel repetition to these problems. In contrast with the external information, they defined the \emph{internal information}, which is the sum of two mutual information terms
\begin{align}
\mathsf{I}(M:X\vert Y) + \mathsf{I}(M:Y\vert X)=
\E_{p(xym)}\Big[\log \Big(\frac{p(x\vert ym)}{p(x\vert y)}\cdot \frac{p(y\vert xm)}{p(y\vert x)}\Big) \Big].\label{internaldef}
\end{align}
The internal information measures what is learned by each party about the other's input. \Cref{divergence} was the inspiration for \Cref{internaldef}; indeed, each setting of $m$ corresponds to a rectangular event. When the inputs come from a product distribution, the internal and external information are the same, and \cite{BBCR}  proved that  subadditivity holds for internal information using an argument similar to the one used in the context of parallel repetition. Moreover, they showed how to leverage the technique of correlated sampling developed by  Holenstein to simulate protocols with information $I$ and communication $C$ using $\approx \sqrt{IC}/\log C$ communication. They gave near optimal simulations of $\approx I \log^2 C$  for protocols with small external information using rejection sampling and a variant of Azuma's concentration inequality. These results proved that there is a constant $\kappa$ such that if $\adv(C,f) < 1/2$, then 
\begin{align*}
\adv\Big(\frac{\kappa C  \sqrt{n}}{ \log(Cn)},f^{\oplus n}\Big) &< 1/2,
\end{align*}
which was the first result along the lines of  \Cref{maintheorem}.  Later, the second author and Braverman \cite{BR} argued that this is the \emph{right} definition of information, because the internal information cost of a function is equal to the amortized communication complexity of that function. This suggested that the internal information might well be the last word in this evolution of definitions, because it could be defined purely using the concept of communication complexity. It seemed like the only path to better results was through better methods to compress  internal information. This is a belief we no longer hold.

Nevertheless, a flurry of ideas about compressing protocols with internal information $I$ and communication $C$ followed. Braverman \cite{Braverman} showed how to obtain protocols with communication $\approx 2^{O(I)}$. The second author and Ramamoorthy  \cite{RR} showed that if $I_A,I_B$ denote the internal information learned by each party, then you can achieve communication $\approx I_A \cdot 2^{O(I_B)}$ and can also achieve communication $\approx I_A + \sqrt[4]{I_B \cdot C^3}$. Two excellent papers, the first by Kol \cite{Kol} and the second by  Sherstov \cite{Sherstov}, showed that  $\approx I \log^2 I$ communication can be achieved when the inputs come from a product distribution. Ganor, Kol and Raz \cite{GKR} (see also \cite{RS}) gave a nice counterexample: a function that can be computed with communication $\approx 2^{2^{O(I)}}$, and internal information $\approx I$, but cannot be computed with communication  $\approx 2^{I}$.  

The next definition to evolve was proposed by  the second author together with Braverman, Weinstein and Yehudayoff \cite{BRWY,BRWY-bounded-round}, inspired by the work of Jain, Pereszl\'enyi and Yao \cite{JPY}. Rather than bounding the information under the distribution $p(xym)$ induced by the protocol, they bounded the infimum of information achieved in the ball of distributions that are close to the protocol. They defined the information to be the infimum 
\begin{align}\inf_q \mathsf{I}_q(M:X\vert Y) + \mathsf{I}_q(M:Y\vert X)
=\inf_q \E_{q(xym)}\Big[\log \Big(\frac{q(x\vert ym)}{q(x\vert y)}\cdot \frac{q(y\vert xm)}{q(y\vert x)}\Big) \Big], \label{brwydef}\end{align}
where here the infimum is taken over all distributions $q(xym)$ that are close to $p(xym)$ in statistical distance. This quantity was ultimately bounded by setting $q(xym) = p(xym\vert W)$, where here $W$ is a reasonably large event (not necessarily rectangular) that implies that the protocol correctly computes the function. The bound on \Cref{brwydef} does not lead to a bound on the information according to $p(xym)$, because it is quite possible that the points outside $W$ reveal a huge amount of information. Still, \cite{BRWY} were able to follow all 3 steps of the high-level approach to prove their results. Step 1 remained easy, but Steps 2 and 3 became more difficult using \Cref{brwydef}. \cite{BRWY} obtained  exponentially small upper bounds for the success probability of computing $f^n$, but did not manage to prove new bounds on the advantage for $f^{\oplus n}$ using this approach. \Cref{brwydef} may not seem very different from \Cref{internaldef}, but it does involve a  proxy $q$, and we pursue the use of such proxies further in the definition of marginal information that we discuss next.

In a paper full of new ideas, Yu \cite{huacheng} recently proved exponentially small bounds on the advantage of bounded-round protocols computing $f^{\oplus n}$. Although Yu's paper involves a potential function that superficially looks like a definition of information,  his proof does not involve a method to compress protocols whose potential is small, and we are unable to  extract a definition of information from his work. Still, his ideas inspired many of the choices made in  our definition. To define the marginal information, we need the concept of a rectangular distribution, which was defined in \cite{huacheng}: 
\begin{definition}\label{rectangular}
    Given a set $Q$ consisting of triples $(xym)$, we say that $Q$ is \emph{rectangular} if its indicator function can be expressed as 
	$$\mathds{1}_Q(xym) = \mathds{1}_A(xm) \cdot \mathds{1}_B(ym),$$ for some Boolean functions $\mathds{1}_A,\mathds{1}_B$.
    Given a distribution $q(xym)$ and a distribution $\mu(xy)$, we say that $q$ is \emph{rectangular with respect to $\mu$} if it can be expressed as   
        \[q(xym) = \mu(xy)\cdot A(xm)\cdot B(ym),\]
    for some functions $A,B$.
\end{definition}
For intuition, it is helpful to think of a rectangular distribution as the result of conditioning a protocol distribution $p(xym)$ on a rectangular event. That would produce a rectangular distribution, but the space of rectangular distributions actually contains other distributions that cannot be obtained in this way.

From our perspective, the most useful  insight of Yu's work is that if $q$ is restricted to being rectangular, then one can allow $q$ to be quite far from $p$ in \Cref{brwydef} and still carry out a meaningful compression of a protocol $p$ to implement Step 3. That is because the rectangular nature of $q$ allows the parties to use hashing and rejection sampling to convert a protocol that samples from $p$ into a protocol that samples from $q$. If $q(xym) = p(xym \vert R)$ for a rectangular event $R$, this is easy to understand: the parties can communicate  $2$ bits to compute if  $xym \in R$ and output the most likely value of $f$ under $q$ with $xym \in R$. If $xym \notin R$ they can output a random guess for the value of $f$. So, it is enough to bound the information terms for $xym \in R$, and enough to guarantee that the compression is efficient for such points. This observation is very powerful, because it allows us to throw away problematic points in the support of the distributions we are working with and pass to appropriate sub-rectangles throughout our proofs.

For all of this to work, it is crucial that the protocol retains some advantage within the support of $q$. For this reason, we need to keep track of the information in the support of $q$ as well as the advantage within the support of $q$, and so, for the first time, the measure of information is going to depend on the function $f$ that the protocol computes. We are ready to state the definition:
\begin{definition}\label{marginaldef}
For  $I\geq 1$ and\footnote{Even though $\delta$ is a fixed constant, we choose to write it in the definition because it eases the notation throughout the paper.}  $\delta=1/15$, the \emph{marginal information} of a protocol $p$ for computing $f$ is defined as
\begin{align*}
    \marg_I(p,f) & = \inf_q \sup_{xym} \log \Big (\frac{q(x\vert ym)}{p(x \vert y)} \cdot \frac{q(y\vert xm)}{p(y \vert x)} \cdot \Big( \frac{q(xym)}{p(xym)}\Big)^I \cdot \Big |\E_{q(xy\vert m)}[(-1)^{f(xy)}]\Big |^{-12I/\delta}\Big ),
\end{align*}
where the infimum is taken over all distributions $q$ that are rectangular with respect to the input distribution $p(xy)$, and the supremum  is taken over all $xym$ in the support of $q$.  
\end{definition}

We use the letter $I$ above because it turns out that protocols computing $f$ can be efficiently compressed when $\marg_I = O(I)$, and any compression must have communication $\Omega(I)$.  Compare \Cref{marginaldef} with \Cref{internaldef,brwydef}. The fact that $q$ must be tethered to $p$ is ensured by including the term $q(xym)/p(xym)$. If $q(xym)=p(xym\vert R)$ for a rectangular event $R$, $q(xym)/p(xym)$ will be equal to $1/p(R)$. The last term in the product computes the advantage of $q$ for computing $f$, because under $q$ and given $m$, the best guess for the value of $f$ is determined by the sign of $\E_{q(xy\vert m)}[(-1)^{f(xy)}]$, and its advantage is the absolute value of this quantity. In words, the marginal information measures the supremum over all $xym$ of the information per unit of advantage, of the best rectangular approximation $q$. 

In analogy with the external information, we define the external marginal information:

\begin{definition}
For  $I\geq 1$ and $\delta=1/15$, the \emph{external marginal information} of a protocol $p$ for computing $f$ is defined as:
\begin{align*}
    \marg^{\mathsf{ext}}_I(p,f) & = \inf_q \sup_{xym} \log \Big (\frac{q(xy\vert m)}{p(x y)} \cdot \Big( \frac{q(xym)}{p(xym)}\Big)^I \cdot \Big |\E_{q(xy\vert m)}[(-1)^{f(xy)}]\Big |^{-12I/\delta}\Big ),
\end{align*}
where the infimum is taken over all distributions $q$ that are rectangular with respect to the input distribution $p(xy)$, and the supremum  is taken over all $xym$ in the support of $q$.  
\end{definition}
We prove that the external marginal information is equal to the marginal information when the distribution on inputs is a product distribution in \Cref{extproduct}.

To state our results about marginal information, we first define the average-case measure of advantage. 
Given a distribution $\mu(xy)$ on inputs, define
\begin{align*}
    \adv_\mu(C,f) &= \sup_{\|\pi \| \leq C} \E[(-1)^{\pi(xy)+f(xy)}],
\end{align*}
where here the expectation is over the choice of inputs $xy$ as well as the random coins of the communication protocol. 
 To study the more restricted setting where the protocols we are working with have a bounded number of rounds, define the worst-case and average case quantities:
\begin{align*}
    \adv^r(C,f) &= \sup_{\|\pi \| \leq C} \inf_{xy} \E[(-1)^{\pi(xy) +f(xy)}], \\
    \adv^r_\mu(C,f) &= \sup_{\|\pi \| \leq C} \E[(-1)^{\pi(xy) +f(xy)}], 
\end{align*}
where throughout, the supremums are taken over $r$-round protocols. 

Returning to our high-level approach, we prove the following results about marginal information, which allow us to carry out Steps 1,2,3:
\begin{enumerate}
\item In \Cref{step1}, we show that a protocol with small communication and large advantage has small marginal information, to handle Step 1:

\begin{theorem}\label{thm:construct-rect-dist}
    For every Boolean function $f(xy)$ and every protocol $p$ of communication complexity $C$,
    \[\marg_I (p,f) \leq 2 C - \left(1 + 12/\delta \right) \cdot I\cdot \log \Big( \E_{p(m)}\Big|\E_{p(xy\vert m)}\Big[(-1)^f\Big]\Big|\Big) + O(I).\]
\end{theorem}

For any fixed $m$, the quantity $|\E_{p(xy\vert m)}[(-1)^f]|$ measures the advantage of the protocol for computing $f$ conditioned on that value of $m$. So,  
if $\adv_\mu(C,f^{\oplus n}) \geq \exp(-m)$ via a protocol corresponding to the distribution $p$, then the above theorem implies that 
$\marg_I(p,f^{\oplus n}) \leq O(C+Im)$. Unlike all previous definitions, for marginal information Step 1 involves significant work. Our proof crucially uses the fact that the protocol  has bounded communication complexity: for example it would not be enough to start with a bound on the internal information.

\item In \Cref{step2}, we prove that marginal information  is sub-additive with respect to the $n$-fold xor of $f$. If the transcript $m=(m_0,m_1,\dotsc,m_C)$, where $m_j$ denotes the $j$'th message of the protocol, we show

\begin{theorem}\label{lem:single-shot-marg-bound} There is a universal constant $\Delta$ such that if  $I \geq 1$ and $p$ is a protocol distribution for computing $f^{\oplus n}$ with $p(xy) = \prod_{i=1}^n p(x_iy_i)$, then there is a protocol $p_i$ for computing $f$ such that $p_i(x_iy_i) = p(x_iy_i)$, $p_i$ has the same number of messages as $p$, for $j>1$  the support of  $m_j$ is identical in $p_i$ and $p$, and moreover 
   \[\marg_I(p_i,f) \leq \frac{ \marg_I(p,f^{\oplus n})}{n} +\Delta I \cdot  \Big(1+\log \frac{\marg_I(p,f^{\oplus n})}{n\cdot  I} \Big).\]
\end{theorem} 
If $\marg_I(p,f^{\oplus n}) \leq O(In)$, this theorem proves that $\marg_I(p_i,f)\leq O(I)$. 
This might well be the most technically novel  part of our proof; it is certainly where we spent the most time. The main challenge is proving the result for $n=2$, which is very delicate. This case is captured by   \Cref{thm:subadditivity}, and   \Cref{lem:single-shot-marg-bound} is a straightforward consequence. If $n=2$ and $\marg_I(p,f^{\oplus 2})$ is small, then there is a rectangular distribution $q$ such that the pair $$q(x_1x_2y_1y_2m),p(x_1x_2y_1y_2m)$$ leads to a small value of $\marg_I(p,f^{\oplus 2})$. We show how to use $q,p$ to generate a new pair $$q_1(x_1y_1\mone),p_1(x_1y_1\mone)$$ or a new pair  $$q_2(x_2y_2\mtwo),p_2(x_2y_2\mtwo)$$ proving that either $\marg_I(p_1,f)$ or $\marg_I(p_2,f)$ is more or less bounded by $\marg_I(p,f^{\oplus 2})/2$.  A significant first step is the construction of two pairs of rectangular/protocol distributions with the properties described in \Cref{eqn:rectsub,eqn:1infosub,eqn:2infosub,eqn:advsub}. Given this step, we need to  eliminate various problematic points from the support of the distributions while preserving the rectangular nature of the distribution  to ultimately construct the promised pair of distributions.

We are unable to bound the length of the first message of $p_i$ in terms of the length of the corresponding message of $p$ in \Cref{lem:single-shot-marg-bound}, because in our proof of \Cref{thm:subadditivity} the first message  $\mone_1$ or $\mtwo_1$ needs to encode one of the inputs of the original protocol. Fortunately, this is not a significant obstacle for the high-level strategy. 

\item In \Cref{step3a,step3b,step3c,sec:braverman-sim}, we show how to compress marginal information to handle Step 3. We have been able to match many of the prior results  \cite{BBCR,BR,Braverman} about compressing information and external  information with corresponding results about compressing marginal information and external marginal information, though our proofs are much more technical. Our most general simulation is captured by the following theorem:
\begin{theorem} \label{simulationthoerem}For every $\alpha>0$ there is a $\Delta>0$ such that if $\marg_I(p,f) \leq \alpha I$, $\mu(xy)=p(xy)$ and moreover the messages $m=(m_0,\dotsc, m_C)$ are such that  $m_2,\dotsc,m_C \in \{0,1\}$, then $\adv_\mu(\Delta(I + \sqrt{CI} \log (CI)),f) \geq 1/\Delta$.
\end{theorem}
\Cref{simulationthoerem} shows that if the marginal information is $O(I)$, then one can obtain a protocol with communication $\tilde O(\sqrt{CI})$ that has $\Omega(1)$ advantage for computing $f$. For the external marginal information, we prove:
\begin{theorem}\label{simulationtheorem2}
    For every $\alpha>0$ there is a $\Delta>0$ such that if  $\mext_I(p,f) \leq \alpha I$, $\mu(xy) = p(xy)$, and moreover the messages $m=(m_0,\dotsc, m_C)$ are such that  $m_2,\dotsc,m_C \in \{0,1\}$, then $\adv_{\mu}(\Delta I \log ^2 C, f) \geq 1/\Delta$.
\end{theorem}
This theorem gives improved results when the inputs come from a product distribution. 
It is quite possible that even better simulations can be obtained using the ideas of \cite{Kol,Sherstov,BK}, but we have not managed to obtain such results. We also obtain results that are independent of the communication complexity:
\begin{theorem}\label{thm:braverman-sim}
    For every $\alpha>0$ there is a $\Delta>0$ such that if  $\marg_I(p,f) \leq \alpha I$ and $\mu(xy)=p(xy)$, then $\mathsf{adv}_\mu(\Delta I, f) \geq \exp(- \Delta I)$.
\end{theorem}

When the number of rounds of the protocol is bounded, we prove:
\begin{theorem} \label{thm:bounded-round-sim} 
    For every $\alpha>0$ there is a $\Delta>0$ such that if $\marg_I(p,f) \leq \alpha I$, $\mu(xy) = p(xy)$, $p$ has $r$-rounds\textbf{} and $m_r \in \{0,1\}$, then $\adv^r_\mu(\Delta r(I+ \log r), f) \geq 1/\Delta$.
\end{theorem}

\end{enumerate}
These results about the marginal information cost allow us to prove \Cref{maintheorem}, as well as several other results of that flavor.

\subsection{Using marginal information to prove XOR lemmas}\label{usingmarginal}
To state all of our results, let us define the average-case and worst-case measures of success:
\begin{align*}
    \suc(C,f) & = \sup_{\|\pi \| \leq C} \inf_{xy} \Pr[\pi(xy)=f(xy)]\\
    \suc^r(C,f) & = \sup_{\|\pi \| \leq C} \inf_{xy}\Pr[\pi(xy)=f(xy)]\\
    \suc_\mu(C,f) & = \sup_{\|\pi \| \leq C} \Pr[\pi(xy)=f(xy)]\\
    \suc^r_\mu(C,f) & = \sup_{\|\pi \| \leq C} \Pr[\pi(xy)=f(xy)],
\end{align*}
where in $\suc^r,\suc_\mu^r$ the supremum is taken over $r$-round protocols, and in $\suc_\mu, \suc^r_\mu$ the probability is over inputs sampled from $\mu(xy)$. Yao's min-max theorem yields
\begin{align}
\adv(C,f) & = \inf_{\mu} \adv_{\mu}(C,f),\notag \\
    \suc(C,f) & = \inf_{\mu} \suc_\mu(C,f),\notag\\
    \adv^r(C,f) & = \inf_{\mu} \adv^r_{\mu}(C,f),\notag \\
    \suc^r(C,f) & = \inf_{\mu} \suc^r_\mu(C,f).
    \label{yao}
\end{align}

 Given any distribution $\mu$ on $\mathcal{X} \times \mathcal{Y}$, define the $n$-fold product distribution $\mu^n$ on $\mathcal{X}^n \times \mathcal{Y}^n$ by $\mu^n(xy)  = \prod_{j=1}^n \mu(x_jy_j)$.
 \Cref{maintheorem} is proved by proving this stronger bound:
\begin{theorem}\label{maintheorem2} There is a universal constant $\kappa>0$ such that
    if $C >1/\kappa $ and $\adv_\mu(C,f) \leq \kappa$, then  $\adv_{\mu^n}(\kappa C  \sqrt{n}/ \log(Cn),f^{\oplus n}) \leq \exp(-\kappa n).$
\end{theorem}
To prove \Cref{maintheorem2}, suppose that there is a protocol $p$ computing $f^{\oplus n}$ with advantage $\exp(-\kappa n)$ and communication $T = \kappa C\cdot\sqrt{n}/\log(Cn)$. If $T/n \geq 1$, we set $I = T/n$ and apply \Cref{thm:construct-rect-dist} to show that $\marg_{I}(p,f^{\oplus n}) \leq O(T+ \kappa I n) \leq O(In)$. Next, apply \Cref{lem:single-shot-marg-bound} to find a protocol $p'$ with $\marg_{I}(p',f) \leq O(I)$. Finally, apply \Cref{simulationthoerem} to obtain a protocol computing $f$ with advantage $\Omega(1)$ and communication proportional to
\begin{align*}
\frac{T}{n} + 2\sqrt{I T}\log(T) &\leq \frac{T}{n} + 2\frac{T\log T}{\sqrt{n}} \\
&\lesssim \frac{\kappa C}{\log nC}\cdot \log T \lesssim  \kappa C.
\end{align*}
If $T/n <1$, set $I = 1$ and apply \Cref{thm:construct-rect-dist} to show that $\marg_I(p,f^{\oplus n}) \leq O(I n)$. Next, apply \Cref{lem:single-shot-marg-bound} to find a protocol $p'$ with $\marg_{I}(p',f) \leq O(I) = O(1)$. Finally, we apply \Cref{thm:braverman-sim} to obtain a protocol computing $f$ with advantage $\Omega(1)$ and communication $O(1)$. Setting $\kappa$ sufficiently small, we obtain a contradiction in either case, which proves that there is no protocol $p$ as above. \Cref{maintheorem} can be obtained from \Cref{maintheorem2}  using \Cref{yao} and the fact that the worst-case success probability of a communication protocol can be increased by taking the majority outcome of several runs of the protocol. We leave these details to the reader. 

 \Cref{maintheorem,maintheorem2} yield bounds on the success probability for computing $f^n$ as well:
\begin{corollary}
    There is a universal constant $\kappa>0$ such that
    if $C >1/\kappa $ and $\adv(C,f) < \kappa$, then  $\suc(\kappa C  \sqrt{n}/ \log(Cn)),f^{ n}) < \exp(-\kappa n)$.
\end{corollary}
\begin{corollary}
    There is a universal constant $\kappa>0$ such that
    if $C >1/\kappa $ and $\adv_\mu(C,f) < \kappa$, then  $\suc_{\mu^n}(\kappa C  \sqrt{n}/ \log(Cn)),f^{ n}) < \exp(-\kappa n)$.
\end{corollary}
This matches the result proved by \cite{BRWY} mentioned earlier. These corollaries are obtained by observing that if $S \subseteq \{1,2,\dotsc,n\}$ is chosen uniformly at random, and $xy$ are sampled according to $\mu^n$, then 
\begin{align*}
    \E\Big[(-1)^{\sum_{j \in S} \pi(xy)_j + f(x_jy_j)}\Big] & = \Pr[\pi(xy)=f^n(xy)],
\end{align*}
so a protocol computing $f^n$ with success probability $\exp(-n/2)$ yields a set of $n' = \Omega(n)$ coordinates where the protocol computes $f^{\oplus n'}$ with advantage $\exp(-\Omega(n))$. Again, we leave the details to the reader. When the distribution $\mu(xy) = \mu(x) \cdot \mu(y)$ is a product distribution, we obtain stronger bounds:
\begin{theorem}\label{maintheorem3} 
   There is a universal constant $\kappa>0$ such that for every product distribution $\mu$, if $C >1/\kappa $ and 
    $\adv_\mu(C,f) < \kappa$, then  $\adv_{\mu^n}(\kappa C  n/ \log^2(Cn),f^{\oplus n}) < \exp(-\kappa n)$.
\end{theorem}
To prove \Cref{maintheorem3}, suppose we are given a protocol $p$ computing $f^{\oplus n}$ with advantage $\exp(-\kappa n)$ and communication $T=\kappa C n / \log^2(Cn)$. If $T/n \geq 1$, we set $I = T/n$ and apply \Cref{thm:construct-rect-dist} to show that $\marg_I(p,f^{\oplus n}) \leq O(n I)$. Next, apply \Cref{lem:single-shot-marg-bound} to find a protocol $p'$ with $\marg_{I}(p',f) \leq O(I)$. Finally, using the fact that for product distributions, $\mext_I(p,f) = \marg_I(p,f)$, we can apply \Cref{simulationtheorem2} to obtain a protocol computing $f$ with advantage $\Omega(1)$ and communication $O(I\log^2(Cn))\leq O(\kappa C)$. Otherwise, if $T/n < 1$, set $I = 1$ and apply \Cref{thm:construct-rect-dist} to show that $\marg_I(p,f^{\oplus n}) \leq O(n)$. Then, apply \Cref{lem:single-shot-marg-bound} to find a protocol $p'$ with $\marg_{I}(p',f) \leq O(I) = O(1)$. Lastly, we apply \Cref{thm:braverman-sim} to obtain a protocol computing $f$ with advantage $\Omega(1)$ and communication $O(1)$. Setting $\kappa$ to be small enough gives a contradiction in either case. 

As before, this yields a corollary for computing $f^n$:
\begin{corollary}
    There is a universal constant $\kappa>0$ such that for every product distribution $\mu$, 
    if $C >1/\kappa $ and  $\adv_\mu(C,f) < \kappa$, then  $\suc_{\mu^n}(\kappa C  n/ \log^2(Cn),f^n) < \exp(-\kappa n).$
\end{corollary}
Again, this is identical to a bound proved by \cite{BRWY} using a different approach. For the bounded-round setting, we prove:
\begin{theorem}\label{maintheoremround} 
    There is a universal constant $\kappa>0$ such that if $C > (r(\log r) +1)/\kappa$, and
    $\adv^r_\mu(C,f) < \kappa$, then  $\adv^r_{\mu^n}((\kappa C/r   - \log r)n,f^{\oplus n}) < \exp(-\kappa n).$
\end{theorem}
Yu \cite{huacheng} proves the same bound on the advantage with a communication budget that grows like $\Omega((C/r^r-O(1))n)$. Our bound eliminates the exponential dependence on $r$. To prove \Cref{maintheorem3}, set $T = (\kappa C/r-\log r)n$, and suppose there is a protocol computing $f$ with $r$ rounds, communication $T$ and advantage $\exp(-\kappa n)$. Set $I = T/n \geq 1$. Then, $\marg_I$ can be bounded by $O(T+\kappa I n)$ by \Cref{thm:construct-rect-dist}. Applying \Cref{lem:single-shot-marg-bound} gives an $r$-round protocol with $\marg_I$ bounded by $O(I)$, and applying \Cref{thm:bounded-round-sim} gives an $r$-round protocol with communication complexity $O(r(I+\log r))=O(\kappa C)$ computing $f$ with advantage $\Omega(1)$. Setting $\kappa$ to be small enough proves the result. As usual, we obtain the following corollaries:

\begin{corollary} There is a universal constant $\kappa>0$ such that if $C > 7(r\log r)/\kappa$ and $\adv^r_\mu(C,f) < \kappa$, then  $\suc^r_{\mu^n}((\kappa C/r   - \log r)n,f^{ n}) < \exp(-\kappa n).$
\end{corollary}

\begin{corollary} There is a universal constant $\kappa>0$ such that if $C > 7(r\log r)/\kappa$, and 
     $\adv^r(C,f) < \kappa$, then  $\suc^r((\kappa C/r   - \log r)n,f^{ n}) < \exp(-\kappa n).$
\end{corollary}

In the rest of the paper, we prove \Cref{thm:construct-rect-dist,lem:single-shot-marg-bound,simulationthoerem,simulationtheorem2,thm:bounded-round-sim,thm:braverman-sim}. We prove \Cref{thm:construct-rect-dist} in \Cref{step1}, and \Cref{lem:single-shot-marg-bound} in \Cref{step2}. In \Cref{trimming} we gather several results related to the \emph{trimming} technique borrowed from \cite{huacheng}  that are used in these first two steps. In \Cref{consequences} we gather several consequences of small marginal information that are used to analyze our compression schemes. We prove the general simulation for marginal information, \Cref{simulationthoerem}, in \Cref{step3a}. In \Cref{smoothing} we prove that if the external marginal information is small, then there is a \emph{smooth} protocol with small external marginal information, mirroring a similar result in \cite{BBCR}. We then show how to compress smooth protocols to prove  \Cref{simulationtheorem2} in \Cref{step3b}. We prove \Cref{thm:bounded-round-sim} in \Cref{step3c} and finally, in \Cref{sec:braverman-sim} we prove \Cref{thm:braverman-sim}. 

\subsection*{Acknowledgements}
Thanks to Paul Beame, Makrand Sinha, Oscar Sprumont, Michael Whitmeyer and  Amir Yehudayoff for helpful conversations.

\input{defs}

\input{reduction-comm}

\input{subadd}

\input{trimming}
\input{argmax}

\input{consequences}

\input{sim-2}

\input{sim-analysis-2}

\input{external-info}

\input{bounded-round}

\input{braverman-sim}

\bibliographystyle{alpha}

\bibliography{ref}

\end{document}

%% file: macros.tex
\usepackage{amssymb}
\usepackage{amsmath,amsthm, bbm, mathdots}
\usepackage{color, xcolor}
\usepackage{graphicx}
\usepackage{setspace}
\usepackage{xspace}
\usepackage{multirow}
\usepackage{array}
\usepackage{complexity}
\usepackage{geometry}
\usepackage{mathtools}
\usepackage{algorithm,algpseudocode}
\usepackage{dirtytalk}
\usepackage{tikz}
\usepackage
{hyperref}
\usepackage[capitalize]{cleveref}

\usepackage{dsfont}

\hypersetup{colorlinks=true,urlcolor=blue,linkcolor=magenta,citecolor=[rgb]{.42,.56,.14},}

\usepackage{indentfirst}
\makeatletter
\renewcommand{\@seccntformat}[1]{\csname the#1\endcsname.\quad}
\makeatother

\theoremstyle{plain}
\newtheorem{theorem}{Theorem}
\newtheorem{corollary}[theorem]{Corollary}
\newtheorem{proposition}[theorem]{Proposition}
\newtheorem{lemma}[theorem]{Lemma}
\newtheorem{claim}[theorem]{Claim}

\newtheorem{definition}[theorem]{Definition}

\theoremstyle{remark}

\theoremstyle{plain}
\newclass{\DNF}{DNF}
\newclass{\DNFs}{DNFs}
\newclass{\ACzero}{AC^0}
\newclass{\TCzero}{TC^0}

\renewcommand{\Pr}{\mathop{\bf Pr\/}}
\renewcommand{\E}{\mathop{\mathbb{E}}}

\newcommand{\adv}{\mathsf{adv}}

\newcommand{\abs}[1]{\left|#1\right|}

\newcommand{\eps}{\varepsilon}

\renewcommand{\tilde}{\widetilde}

\newcommand{\calE}{\mathcal{E}}

\newcommand{\calG}{\mathcal{G}}

\newcommand{\calM}{\mathcal{M}}

\newcommand{\calZ}{\mathcal{Z}}

\newcommand{\suc}{\mathsf{suc}}

\renewcommand{\epsilon}{\varepsilon}

\newcommand{\marg}{\mathsf{M}}

\newcommand{\mext}{\mathsf{M}^{\mathsf{ext}}}

\newcommand{\mone}{m^{(1)}}
\newcommand{\mtwo}{m^{(2)}}

%% file: defs.tex
\section{Preliminaries}

Throughout, we assume that $x \in \mathcal{X}$, $y \in \mathcal{Y}$ and $m\in \mathcal{M}$ for some finite sets $\mathcal{X},\mathcal{Y}, \mathcal{M}$.
Let $\mu(xy)$ be a distribution on pairs of inputs. To ease the notation, we often write $ab$ instead of the tuple $(a,b)$.
Everywhere in the paper, we assume that $\delta>0$ is a sufficiently small constant; $\delta = 1/15$ will suffice.

\begin{definition}    
   We say that $p(xym)$ is a  protocol distribution if it can be expressed as \[p(xym) =p(xy) \cdot p(m_0)\cdot \prod_{i=1,3,5,\dotsc}p(m_i\vert xm_{< i})\cdot p(m_{i+1}\vert ym_{\leq i}).\] 
\end{definition}

Every randomized worst-case protocol corresponds to some protocol distribution $p(xym)$, where $p(xy)$ can be taken to be the uniform distribution on all possible inputs. Given a distribution $\mu(xy)$ on inputs, and any protocol generating the messages $m$, the joint distribution of $xym$ corresponds again to a protocol distribution $p(xym)$, with $p(xy) = \mu(xy)$. 

Recall \Cref{rectangular}. 
Note that if $q$ is rectangular 
 with respect to $\mu(xy)$ and $p$ is a protocol with $p(xy) = \mu(xy)$, it is not necessary that $q(xy)=\mu(xy)$. For the purpose of intuition, it may be helpful to think of a rectangular distribution as the result of conditioning $\mu(xy)$ on the event that it lies in a disjoint union of rectangles indexed by $m$, though this statement is not without loss of generality, and we do use the full generality of \Cref{rectangular}.

Let $x=x_1x_2$ and $y = y_1y_2$. Let $\mu(xy) = \mu(x_1y_1) \cdot \mu(x_2y_2)$ be a product distribution. It will be helpful to define $w=(x_1y_2m)$. 
Given $m = (m_0,\ldots,m_r)$ and $y_2$, we denote  
\begin{align}
\mone&= (m_0,y_2m_1,m_2,\ldots,m_r),\notag\\
 \mtwo &= (m_0x_1,m_1,m_2,\ldots,m_r). \label{monemtwo}
 \end{align}

Let us gather some basic facts about rectangular distributions in this setting:

\begin{proposition}\label{cor:multiplicativity}
	If $v$ is rectangular, then 
	\begin{enumerate}
        \item $v(xy|w) = v(y_1|w) \cdot v(x_2|w)$,
		\item $v(xw)\cdot v(yw) = v(xym)\cdot v(w)$,
		\item $v(x_1\vert y_1 \mone )\cdot v(x_2 \vert y_2 \mtwo ) = v(x\vert ym)$, and
		\item $v(y_1\vert x_1 \mone)\cdot v(y_2\vert x_2 \mtwo) = v(y\vert xm)$.
	\end{enumerate}
\end{proposition}
\begin{proof}

For the first identity, let $A,B$ be such that $v(xym) = \mu(xy) \cdot A(xm) \cdot B(ym)$. Then 
\begin{align*}
v(xy|w) = \frac{v(xyw)}{v(w)}
&=\frac{\mu(x_1y_1) \cdot \mu(x_2y_2) \cdot A(xm) \cdot B(ym)}{\sum_{x'_2y'_1} \mu(x_1  y'_1 ) \cdot \mu(x'_2 y_2) \cdot A(x_1x'_2m) \cdot B(y'_1 y_2 m)}\\
&=\frac{\mu(x_1y_1)\cdot B(y_1y_2 m)}{\sum_{y'_1} \mu(x_1  y'_1 ) \cdot B(y'_1 y_2 m) } \cdot \frac{\mu(x_2y_2)  \cdot A(x_1x_2 m)}{\sum_{x'_2} \mu(x'_2 y_2) \cdot  A(x_1x'_2m)}   \\
&=v(y_1|w) \cdot v(x_2|w).  
\end{align*}

	For the second identity, 
	\begin{align*}
		v(xw)\cdot v(yw) &= v(w)\cdot v(yw) \cdot v(x|w) \\
		&= v(w) \cdot v(yw) \cdot v(x_2|x_1 y m)  \tag{by the first identity}\\
		&= v(w) \cdot v(xym).
	\end{align*}
	
	For the third identity, 
	\begin{align*}
		v(x_1\vert y_1\mone)\cdot v(x_2\vert y_2\mtwo) &= v(x_1\vert ym)\cdot v(x_2\vert x_1y_2m) \\
        &=  v(x_1\vert ym)\cdot v(x_2\vert x_1ym) \tag{by the first identity}\\
        &=  v(x\vert ym).
	\end{align*}
	A similar calculation yields the fourth identity.
\end{proof}

It is easy to check that the external marginal inforamtion and marginal information are the same when the distribution on inputs is product:
\begin{lemma}\label{extproduct}
If $p(xy)=\mu(xy)$ is a product distribution then we have  $\mext_I(p,f) = \marg_I(p,f).$
\end{lemma}
\begin{proof}
For all rectangular $q$, we have
    \begin{align*}
        q(xy\vert m) & = \frac{q(xym)}{q(m)}\\
        & = \frac{\mu(xy) \cdot A(xm) \cdot B(ym)}{\sum_{x'y'} \mu(x'y') \cdot A(x'm) \cdot B(y'm)}\\
        & = \frac{\mu(x) \mu(y) \cdot A(xm) \cdot B(ym)}{\sum_{x'y'} \mu(x') \mu(y') \cdot A(x'm) \cdot B(y'm)}\\
        & = \frac{\mu(x) \cdot A(xm)}{\sum_{x'} \mu(x')  \cdot A(x'm) } \cdot  \frac{\mu(y)  \cdot B(ym)}{\sum_{y'}  \mu(y') \cdot B(y'm)},
\end{align*}
    proving that $q(xy\vert m)$ is a product distribution.

Thus:
$$ \frac{q(x\vert ym)}{p(x\vert y)} \cdot \frac{q(y\vert xm)}{p(y \vert x)} = \frac{q(xy\vert m)}{p(xy)},$$
and so $\mext(p,f) = \marg_I(p,f)$.
\end{proof}

We recall the definition of divergence and some relevant facts from information theory.
\begin{definition}
    Given two distributions $a(u)$ and $b(u)$, the divergence of $a$ from $b$ is defined as
    \[\E_{a(u)}\bigg[\log \frac{a(u)}{b(u)}\bigg].\]
\end{definition}

\begin{proposition}
    For two distributions $a(u)$ and $b(u)$, we have that 
    \begin{align}
    \E_{a(u)}\bigg[\log \frac{a(u)}{b(u)}\bigg] &\geq 0\label{eq:div-non-neg},\\
    \sqrt{\E_{a(u)}\bigg[\log \frac{a(u)}{b(u)}\bigg]} &\geq \frac{\|a(u) - b(u)\|_1}{2}.\label{eq:pinsker-inequality}
    \end{align}
\end{proposition}
The proof of this proposition can be found in \cite{CoverT91}. 

The following version of a protocol due to \cite{BR} appears as Lemma 43 in \cite{huacheng}, the parameters we cite here are clear from the proof.
\begin{lemma}\label{lem:round-sim}
Let $u,v$ denote two distributions on some finite set $\mathcal{M}$. For every $\epsilon>0$, there is a $1$-round protocol distribution $\psi(uvs)$ (here $uv$ correspond to the inputs of the protocol, and $s$ corresponds to the transcript), and functions $a(us) \in \mathcal{M}, b(vs) \in \mathcal{M} \cup \{\bot\}$ (here $\bot \notin \mathcal{M}$), such that $\psi(uv)$ is supported on all pairs $uv$ and for every $uv$ and $z \in \mathcal{M}$,
\begin{enumerate}
\item $\psi(a(us)=z\vert uv) = u(z)$,
\item $\psi(a(us)\neq b(vs)\vert uv,a(us)) \leq \epsilon + \max\{0, 1- 2^L \cdot \frac{v(a(us))}{u(a(us))}\}$.
\item $\psi \left (b(vs) \notin \{a(us),\bot\}\vert uv\right ) \leq \epsilon$.
\end{enumerate}
Moreover, the communication complexity of $\psi$ is $L + \log\log 1/\eps + \log1/\eps + O(1)$.
\end{lemma}

Additionally, we need the following Lemma, which appears as Lemma 4.14 in \cite{BBCR}.

\begin{lemma}\label{lem:first-diff}
There is a randomized protocol $\tau$ with communication complexity at most $O(\log (C/\eps))$ such that on input two $C$-bit strings $m^A, m^B$, it outputs the first index $i\in [C]$ such that $m^A_i\neq m^B_i$ with probability at least $1-\eps$, if such an $i$ exists.
\end{lemma}

%% file: reduction-comm.tex
\section{Marginal information of efficient protocols} \label{step1}

In this section, we prove \Cref{thm:construct-rect-dist}.
    Let $R$ be a rectangular set that maximizes the quantity 
    \[p(R)^\delta\cdot \E_{p(m\vert R)}\Big|\E_{p(xy\vert m R)}\Big[(-1)^f\Big]\Big|.\]

    We shall use trimming to prove the following claim:
    \begin{claim}\label{claim:info-rect-bounds}
        There exists a rectangular set $T\subseteq R$ with $p(T\vert R) \geq 1/4$ such that for any $xym$ in the support of $T$, we have
        \begin{align*}
        \frac{p(xym\vert T)}{p(xym)} &\leq  \frac{4}{ p(R)},\\
         \log \frac{p(x\vert ym T)}{p(x\vert y)}, \log \frac{p(y\vert xm T)}{p(y\vert x)} &\leq \frac{96 \cdot 2^{C}}{p(R)^2},\\
         \E_{p(m\vert T)}\Big|\E_{p(xy\vert mT)}\Big[(-1)^f\Big]\Big| 
        &\geq \frac{\Omega(1)}{p(R)^\delta} \cdot   \E_{p(m)}\Big|\E_{p(xy\vert m)}\Big[(-1)^f\Big]\Big|.\\
        \end{align*}
    \end{claim}
 
    We defer the proof of the above  claim to the end of this section. Let $Q\subseteq T$ be the sub-rectangle obtained by keeping only the messages $m'$ for which the advantage is at least half of  the average advantage: 
       \[
        Q = \Big\{x'y'm' \in T: \Big|\E_{p(xy\vert m'T)}\Big[(-1)^f\Big]\Big| \geq  \frac{1}{2}\cdot \E_{p(m\vert T)}\Big|\E_{p(xy\vert mT)}\Big[(-1)^f\Big]\Big|\Big \}.
    \]
    Observe that 
    \begin{align*}
        \E_{p(m\vert T)}\Big|\E_{p(xy\vert mT)}\Big[(-1)^f\Big]\Big| &<  p(Q\vert T)+  \sum_{m':p(m'|Q)=0}p(m'\vert T)\cdot \frac{1}{2}\cdot \E_{p(m\vert T)}\Big|\E_{p(xy\vert mT)}\Big[(-1)^f\Big]\Big| \\
        &\leq p(Q\vert T) + \frac{1}{2}\cdot \E_{p(m\vert T)}\Big|\E_{p(xy\vert mT)}\Big[(-1)^f\Big]\Big|,
    \end{align*}
    and so by the choice of $R$,
    \begin{align}p(Q\vert T) &\geq \frac{1}{2} \cdot \E_{p(m\vert T)}\Big|\E_{p(xy\vert mT)}\Big[(-1)^f\Big]\Big| 
     \geq  \frac{\Omega(1)}{p(R)^\delta} \cdot   \E_{p(m)}\Big|\E_{p(xy\vert m)}\Big[(-1)^f\Big]\Big|.
    \label{eqn:qbound} \end{align}
    Define the rectangular distribution $q(xym) = p(xym\vert Q)$. By the definition of $Q$ and \Cref{claim:info-rect-bounds}, we have that for all $m$ in the support of $q$:
    \begin{align}
    \Big|\E_{p(xy\vert mQ)}\Big[(-1)^f\Big]\Big|
      \geq
    \frac{1}{2}\cdot \E_{p(m|T)}\Big|\E_{p(xy\vert mT)}\Big[(-1)^f\Big]\Big| 
    \geq \frac{\Omega(1)}{p(R)^\delta} \cdot   \E_{p(m)}\Big|\E_{p(xy\vert m)}\Big[(-1)^f\Big]\Big|. \label{eqn:qbound2}
    \end{align}
    
    Using \Cref{claim:info-rect-bounds,eqn:qbound,eqn:qbound2} and the definition of $Q$, we can bound the marginal information cost by
    \begin{align*}
        2^{\marg(p,f)} &\leq \sup_{xym} \left(\frac{p(x\vert ym Q)}{p(x\vert y)}\cdot \frac{p(y\vert xm Q)}{p(y\vert x)}\right)  \cdot  \left(\frac{p(xym\vert Q)}{p(xym)}\right)^I \cdot \left(\Big|\E_{p(xy\vert m Q)}\Big[(-1)^f\Big]\Big|\right)^{-12I/\delta}\\
        &\leq \sup_{xym}  \left(\frac{p(x\vert ym T)}{p(x\vert y)}\cdot \frac{p(y\vert xm T)}{p(y\vert x)}\right) \cdot \left(\frac{p(xym\vert T)}{p(xym) \cdot p(Q\vert T)} \right)^I \cdot   \left ( \Big|\E_{p(xy\vert mQ)}\Big[(-1)^f\Big]\Big|\right)^{ -12I/\delta}\\
        &\leq  O(1) \cdot 2^{2 C}\cdot p(R)^{-4} \cdot 2^{O(I)}\cdot p(R)^{-I(1-\delta)+12I} \cdot    \left(  \E_{p(m)}\Big|\E_{p(xy\vert m)}\Big[(-1)^f\Big]\Big|\right)^{-I(1+ 12/\delta)} \\
        &\leq   O(1) \cdot  2^{2 C} \cdot 2^{O(I)} \cdot \left(  \E_{p(m)}\Big|\E_{p(xy\vert m)}\Big[(-1)^f\Big]\Big|\right)^{-I(1+ 12/\delta)},
    \end{align*}
    where in the last inequality we used the fact that $I(12+\delta-1)-4>0$ since $I\geq 1$.
    This completes the proof of the theorem.

It only remains to prove \Cref{claim:info-rect-bounds}. 
    We have
    \[\E_{p(ym\vert R)}\left[\frac{1}{p(m\vert y)}\right] = \sum_{ym}\frac{p(ym\vert R)}{p(m\vert y)}  \leq \frac{1}{p(R)}\sum_{ym} \frac{p(ym)}{p(m\vert y)} 
        = \frac{\sum_{ym} p(y)}{p(R)} \leq \frac{2^{C}}{p(R)},
    \]
    since the communication complexity of $p$ is bounded by $C$. A similar argument proves 
        \[\E_{p(xm\vert R)}\left[\frac{1}{p(m\vert x)}\right]  \leq \frac{2^{C}}{p(R)}.
    \]
    
    Define the rectangular set
    \[G  =\left\{ xym \in R: \frac{1}{p(m\vert y)} , \frac{1}{p(m\vert x)}\leq 4 \cdot \frac{ 2^{C}}{p(R)} \right\}.\]  Markov's inequality implies that 
    $p(G\vert R) \geq 1/2$.
    We apply \Cref{lem:trimming-1} with $a(xym) =p(xym|G)$, $b(xym)=p(xym)$, and  $\kappa=1/6$ to obtain a rectangular set $T \subseteq G$ with $p(T|G)\geq 1/2$ and 
    \begin{align}
    \frac{p(xm\vert T)}{p(xm)} , \frac{p(ym\vert T)}{p(ym)} \geq \frac{1}{6}, \label{eqn:trimfirst}
    \end{align}
    for all points in the support of $T$. We compute
    \[p(T\vert R) =   p(G \vert R)\cdot p(T \vert G) \geq \frac{1}{4}.\]
    
    Let us verify that $T$ satisfies the remaining conditions promised by \Cref{claim:info-rect-bounds}. We have 
    \[ \frac{p(xym \vert T)}{p(xym)} = \frac{1}{p(T)} = \frac{1}{p(T\vert R)\cdot p(R)} \leq \frac{4}{p(R)}.\]
    To prove the second identity, use the first identity, the definition of $G$ and \Cref{eqn:trimfirst}:
    \begin{align*}
        \frac{p(x\vert ym T)}{p(x\vert y)} &= \frac{1}{p(ym\vert T)}\cdot \frac{p(xym\vert T)}{p(x\vert y)} \\
        &\leq \frac{6}{p(ym)}\cdot \frac{4\cdot p(xym)}{p(x\vert y)\cdot p(R)} \\
        &=\frac{24\cdot p(m\vert xy)}{p(m\vert y)\cdot p(R)} \\
        &\leq \frac{96 \cdot 2^C}{p(R)^2}.
    \end{align*}
 A similar calculation yields 
\begin{align*}
        \frac{p(y\vert xm T)}{p(y\vert x)} \leq \frac{96 \cdot 2^C}{p(R)^2}.
    \end{align*}
    Finally, applying \Cref{lem:adv-preserving} with $v(xym)=p(xym)$, $Z = T$, and noting that $p(Z\vert R) \geq 1/4$, we get
    \begin{align*}
        \E_{p(m\vert T)}\Big|\E_{p(xy\vert mT)}\Big[(-1)^f\Big]\Big| 
        &\geq \frac{1-\delta^2-4\delta}{p(R)^\delta}\cdot \E_{p(m)}\Big|\E_{p(xy\vert m)}\Big[(-1)^f\Big]\Big| \\
        &\geq \frac{\Omega(1)}{p(R)^\delta}\cdot \E_{p(m)}\Big|\E_{p(xy\vert m)}\Big[(-1)^f\Big]\Big| .
    \end{align*}
    This completes the proof of \Cref{claim:info-rect-bounds}.

%% file: subadd.tex
\section {Marginal information is subadditive} \label{step2}
In this section we prove \Cref{lem:single-shot-marg-bound}. Recall the definitions of $\mone, \mtwo$, which are given in \Cref{monemtwo}. The core of the proof is the following statement. 

\begin{theorem}\label{thm:subadditivity}
  Let $f(x_1y_1)$ and $g(x_2y_2)$ be two Boolean functions and let $p(xym)$ be a protocol distribution such that $p(x_1y_1x_2y_2)=p(x_1y_1) \cdot p(x_2y_2)$. Then, for every $1/3 \leq \gamma \leq 2/3$, there are protocol distributions $p_1(x_1y_1\mone)$, $p_2(x_2y_2\mtwo)$ such that 
  $p_1(x_1y_1) = p(x_1y_1)$,
  $p_2(x_2y_2) = p(x_2y_2)$, 
  and 
  \begin{align*} &\min\big\{\marg_I(p_1,f) - \gamma \cdot \marg_I(p,f\oplus g),\marg_I(p_2,g)-(1-\gamma) \cdot \marg_I(p,f\oplus g)\big\} \\& \qquad \leq  3I \cdot \log \frac{\marg_I(p,f\oplus g)}{I} + O(I).\end{align*}
\end{theorem}

We shall prove \Cref{lem:single-shot-marg-bound} assuming \Cref{thm:subadditivity}, whose proof we supply right after.

\subsection*{Proof of \Cref{lem:single-shot-marg-bound}}
Let $k_0>1$ be a large constant, to be determined. Define $f_i(x_iy_i) = f(x_iy_i)$. For $\ell=1,2,\dotsc,n$, define $$k(\ell)= \max\Big \{\inf_{\substack{S \subset [n], |S| = \ell\\ p'}} \marg(p',\oplus_{i\in S}f_i),k_0 I\Big \},$$
where the infimum is taken over all protocols $p'$ with $C$ messages such that the support of $m_2,\dotsc,m_C$ is the same as in $p$. Define $$T = \max\{k(n), k_0 n I\}.$$ 
Note that $$k(n) \leq T \leq \frac{nT}{n} + 12I \cdot\log \frac{nT}{I \cdot n}.$$
For any $\ell>1$, suppose we have 
\begin{align}
k(\ell)\leq \frac{\ell T}{n} + 12I \cdot \log \frac{\ell T}{In},\label{eqn:lcondition}
\end{align}
then set $\gamma = \lceil \ell/2 \rceil/\ell$. Since $1/3 \leq \gamma \leq 2/3$, for $k_0$ chosen large enough, \Cref{thm:subadditivity} shows that for some $\ell' \in \{\lceil \ell/2 \rceil,\lfloor \ell/2 \rfloor\}$, we have
\begin{align*}
    k(\ell') &\leq \max\{ \frac{\ell'}{\ell}\cdot  k(\ell) + 3I \log \frac{k(\ell)}{I}, k_0 I \}\\
    & \leq \frac{\ell'}{\ell} \cdot \frac{\ell T}{n} + \frac{\ell'}{\ell} \cdot 12I \cdot \log \frac{\ell T}{I n}+ 3I \log \Big( \frac{\ell T}{In} + 12 \log \frac{\ell T}{In}\Big) \tag{by the choice of $T$ and \Cref{eqn:lcondition}}\\
    & \leq  \frac{\ell' T}{n} +  8I \cdot \log \frac{\ell T}{I n}+ 3I \log \Big( 13\cdot \frac{\ell T}{In} \Big) \\
    & =  \frac{\ell' T}{n} +  11I \cdot \log \frac{\ell' T}{I n}  + 11 I \cdot \log \frac{\ell}{\ell'} + 3I \log 13\\
    & \leq  \frac{\ell' T}{n} +  11I \cdot \log \frac{\ell' T}{I n}  + 11 I \cdot \log \frac{3}{2} + 3I \log 13\\
    & \leq  \frac{\ell' T}{n} +  12I \cdot \log \frac{\ell' T}{I n} ,
\end{align*}
for $k_0$ chosen large enough.

So, starting with $\ell=n$, we obtain a smaller and smaller $\ell$ satisfying \Cref{eqn:lcondition}, until $\ell=1$, which completes the proof.

\subsection*{Proof of \Cref{thm:subadditivity}}
Given a Boolean function $h(xy)$, a protocol distribution $p(xym)$ and $q(xym)$  that is rectangular with  respect to  $p(xy)$, it will be convenient to define 
\[ \marg_I(q,p,h)= \sup_{xym \in \mathsf{supp}(q)} \log \Bigg( \frac{q(x\vert ym)}{p(x\vert y)}\cdot\frac{q(y\vert xm)}{p(y
        \vert x)} \cdot \bigg(\frac{q(xym)}{p(xym)}\bigg)^I  \cdot   \bigg|\E_{q(xy\vert m)}\left[(-1)^{h}\right]\bigg|^{-12I/\delta} \Bigg),
\]
so $\marg_I(p,h) = \inf_q \marg_I(q,p,h)$. We shall prove that there are protocol distributions $p_1(x_1y_1\mone)$, $p_2(x_2y_2 \mtwo)$ with $p_1(x_1y_1)=p(x_1y_1)$ and  $p_2(x_2y_2)=p(x_2y_2)$ and rectangular distributions $r_1,r_2$ such that 
  \begin{align*} &\min\big\{\marg_I(r_1,p_1,f) - \gamma \cdot \marg_I(q,p,f\oplus g),\marg_I(r_2,p_2,g)-(1-\gamma) \cdot \marg_I(q,p,f\oplus g)\big\} \\& \qquad \leq  3I \cdot \log \frac{\marg_I(q,p,f\oplus g)}{I} + O(I),\end{align*}
from which the theorem follows.
 
Before we give the actual proof, let us give a high level overview of all the steps. Recall the definitions of $\mone, \mtwo$ and $w$, given in the preliminaries.  We start by defining rectangular distributions $q_1(x_1y_1\mone),q_2(x_2y_2\mtwo)$ and protocol distributions $p_1(x_1y_1 \mone),p_2(x_2y_2\mtwo)$ that satisfy the identities described in \Cref{eqn:rectsub,eqn:1infosub,eqn:2infosub,eqn:advsub}. The distributions $q_1,q_2$ are not be the same as our final rectangular distributions $r_1,r_2$, but they are closely related. We would like to prove that $$\marg(q_1,p_1,f)+ \marg(q_2,p_2,g) \leq \marg(q,p,f\oplus g),$$ but the advantage terms do not add nicely in the marginal information cost: in \Cref{eqn:advsub}, the advantage is computed with respect to $w$, and not $\mone, \mtwo$ or $m$. For example, there may be some $mw$ in the support for which 
$$\abs{\E_{q(xy\vert w)}\left[(-1)^{f\oplus g}\right]} \ll \abs{\E_{q(xy\vert m)}\left[(-1)^{f\oplus g}\right]}.$$
To resolve this issue, we define a subset $G$ whose indicator function $1_G(xym)$ depends only on $w$, and yet for all $mw$ in the support of $G$, the advantage is preserved in the sense of \Cref{eqn:advstep1}. This allows us to convert the advantage term in $\marg(q,p,f \oplus g)$ into the kind of term where \Cref{eqn:advsub} can be applied, and we use it to get subadditivity as described in \Cref{eqn:subadd1}. This equation shows that the costs add up pointwise, and so we can pass to a large subset $U\cap L$ where the costs in, say, the first coordinate are a $\gamma$-fraction of the total, see  \Cref{eqn:subadd2}. We are left with our final obstacle: once again the advantage term that we have control over is not exactly the one we want, it may well be that \begin{align*}\Big|\E_{q_1(x_1y_1\vert \mone)}\Big[(-1)^f\Big]\Big| \ll \Big|\E_{q_1(x_1y_1\vert w)}\Big[(-1)^f\Big]\Big|.
\end{align*}
To address this, we show that after passing to a suitable set $U' \cap L'$ (whose indicator function depends only on $w$), the advantage for each fixed $w$ is at least $2^{-\Omega(\marg(q,p,f \oplus g))}$ (\Cref{claim:largeadv}). We then cluster the $w$ and pass to a  subset $B$ of density $\Omega(\marg(q,p,f \oplus g)^{-1})$ where the advantage terms for each $w$ are within a factor of $2$ of each other. The low density of this set is what leads to the $\log \marg$ factor in the statement of the theorem. This allows us to show that the advantage with respect to $\mone$ is comparable to the advantage with respect to $w$ (\Cref{eq:B-mone-prop}). All of these steps leave us with a subset of the inputs $xym$ where the proof gives good control on the quantity $\marg(q_1,p_1,f)$, but we now need to define a distribution $r_1$ supported on these points where $\marg(r_1,p_1,f)$ can be bounded. To do this we need to carefully control the sizes of all the sets we encounter during the proof, and define the distribution of $r_1$ carefully.

Now we begin the actual proof. Define  
    \begin{align*}
    q_1(x_1y_1\mone) &= q(x_1y_1\mone) = q(yw)\\
    q_2(x_2y_2\mtwo) &= q(x_2y_2\mtwo) =q(xw)\\
    	p_1(x_1y_1\mone) &= p(x_1y_1)\cdot p(m_0)\cdot q(y_2\vert x_1m_0)\cdot \prod_{i =1,3,5,\dotsc}q(m_i\vert x_1y_2m_{< i})\cdot p(m_{i+1}\vert ym_{\leq i})\\
        p_2(x_2y_2\mtwo) &= p(x_2y_2)\cdot  q(m_0x_1)\cdot \prod_{i=1,3,5,\dotsc}p(m_i\vert xm_{< i})\cdot q(m_{i+1}\vert x_1y_2m_{\leq i})
    \end{align*} 
    
   These distributions have been carefully chosen to have many nice properties. First, observe that $p_1(x_1y_1\mone),p_2(x_2y_2 \mtwo)$ are protocol distributions, with the same number of rounds of communication as $p$, though the length of the $\mone_1$ is longer than $m_1$, and the length of $\mtwo_0$ is longer than $m_0$. Since $q$ is rectangular with respect to $p(xy)$, we have $q(xym) = p(xy) \cdot A(xm) \cdot B(ym)$ for some functions $A,B$. So, we get 
   \begin{align*}
       &q_1(x_1y_1\mone) \\
       &= \sum_{x_2}q(xym) \\
       &= \sum_{x_2}p(x_1y_1)\cdot p(x_2y_2)\cdot A(xm)\cdot B(ym)\\
       &= p(x_1y_1)\cdot \bigg(\sum_{x_2} p(x_2y_2)\cdot A(xm)\bigg)\cdot B(ym) = p(x_1y_1)\cdot A'(x_1\mone)\cdot B'(y_1\mone),
   \end{align*}
   proving that $q_1$ is rectangular with respect to $p(x_1y_1)$.
    A similar calculation calculation shows that $q_2(x_2y_2\mtwo)$ is rectangular with respect to $p(x_2y_2)$.
    Using the fact that $p$ is a protocol and $x_1y_1$ and $x_2y_2$ are independent under $p$, we can compute:
        \begin{align}
    	   p_1(x_1y_1\mone)\cdot p_2(x_2y_2\mtwo) 
        &= p(xy)\cdot p(m_0) \cdot q(x_1y_2m_0)\cdot \prod_{i>0}p(m_{ i}\vert xym_{< i})\cdot q(m_{i}\vert x_1y_2m_{< i}) \notag \\
	       &= p(xym)\cdot q(x_1y_2m) \notag \\
        &= p(xym) \cdot q(w). \label{eqn:p_1.p_2-feature}
        \end{align}

   The pairs $p_1,p_2,q_1,q_2$ have been engineered so that the various terms in the marginal cost add up nicely across the pairs. We have:

     \begin{align}
         \frac{q_1(x_1y_1\mone)}{p_1(x_1y_1\mone)}\cdot \frac{q_2(x_2y_2\mtwo)}{p_2(x_2y_2\mtwo)} &= \frac{q(xym)}{p(xym)}, \label{eqn:rectsub} \\
	       \frac{q_1(x_1\vert y_1\mone)}{p_1(x_1\vert y_1)}\cdot \frac{q_2(x_2\vert y_2\mtwo)}{p_2(x_2\vert y_2)} &= \frac{q(x\vert ym)}{p(x\vert y)}, \label{eqn:1infosub}\\
	        \frac{q_1(y_1\vert x_1\mone)}{p_1(y_1\vert x_1)}\cdot \frac{q_2(y_2\vert x_2\mtwo)}{p_2(y_2\vert x_2)} &= \frac{q(y\vert xm)}{p(y\vert x)}, \label{eqn:2infosub}\\
              \abs{\E_{q_1(x_1y_1\vert w)}\left[(-1)^f\right]}\cdot \abs{\E_{q_2(x_2y_2\vert w)}\left[(-1)^g\right]} &= \abs{\E_{q(xy\vert w)}\left[(-1)^{f\oplus g}\right]}.\label{eqn:advsub}
              \end{align}
  
     To prove \Cref{eqn:rectsub},  use  \Cref{eqn:p_1.p_2-feature,cor:multiplicativity} to obtain 
        \begin{align*}
            \frac{q_1(x_1y_1\mone)}{p_1(x_1y_1\mone)}\cdot \frac{q_2(x_2y_2\mtwo)}{p_2(x_2y_2\mtwo)} &= \frac{q(xym)}{p(xym)}\cdot \frac{q(w)}{q(w)} 
            = \frac{q(xym)}{p(xym)}.
        \end{align*}
         \Cref{eqn:1infosub,eqn:2infosub} follow directly from \Cref{cor:multiplicativity}. We use the fact that $q(xy\vert w)=q(x_2\vert w) \cdot  q(y_1\vert w)$ from \Cref{cor:multiplicativity} to prove \Cref{eqn:advsub}:
        \begin{align*}
            \abs{\E_{q_1(x_1y_1\vert w)}\left[(-1)^f\right]}\cdot \abs{\E_{q_2(x_2y_2\vert w)}\left[(-1)^g\right]} = \abs{\E_{q(y_1\vert w)}\left[(-1)^{f}\right]\cdot \E_{q(x_2\vert w)}\left[(-1)^g\right]} = \abs{\E_{q(xy\vert w)}\left[(-1)^{f\oplus g}\right]}.  
        \end{align*}

These identities suggest that  the costs in the first and second coordinates should sum to $\marg(q,p,f\oplus g)$. 
    The main challenge in applying this intuition is that the advantage terms in \Cref{eqn:advsub} are not the ones needed for $\marg(q_1,p_1,f)$ and $ \marg(q_2,p_2,g)$. To resolve this, we need to remove some problematic points in the support of $q$. We need to do this while retaining the rectangular structure of $q_1,q_2$ and preserving the sub-additivity of the other terms in the marginal cost.

    Let $G$ be a subset of triples $xym$ such that the indicator function $1_{G}(xyw)$ depends only on $w$, and for each fixed $m$, the  set $G$ maximizes 
    \begin{align}
    q(G\vert m)^\delta\cdot \Big|\E_{q(xy\vert mG)}\Big[(-1)^{f\oplus g}\Big]\Big|,\label{eqn:gmax}
    \end{align}
  among all such sets. In \Cref{lem:advantage-preserving-1}, we prove that for all $w$ in the support of $G$:
        \begin{align}\abs{\E_{q(xy\vert w)}\left[(-1)^{f\oplus g}\right]}
        \geq (1-\delta)\cdot q(G\vert m)^{-\delta}\cdot \abs{\E_{q(xy\vert m)}\left[(-1)^{f\oplus g }\right]}. \label{eqn:advstep1}\end{align}

    This gives us an effective way to split the costs for $q,p$. Using \Cref{eqn:rectsub,eqn:1infosub,eqn:2infosub,eqn:advsub}, we obtain that for all $xym$ in the support of $G$,

    \begin{align}
    &\frac{q_1(x_1|y_1\mone)}{p_1(x_1|y_1)} \frac{q_1(y_1|x_1\mone)}{p_1(y_1|x_1)}  &\cdot& \bigg(\frac{q_1(x_1y_1\mone)}{p_1(x_1y_1\mone)}\bigg)^I &\cdot&   \abs{\E_{q_1(x_1y_1\vert w)}\left[(-1)^{f}\right]}^{-12I/\delta} \notag \\
    \times &\frac{q_2(x_2|y_2\mtwo)}{p_2(x_2|y_2)} \frac{q_2(y_2|x_2\mtwo)}{p_2(y_2|x_2)}  &\cdot& \bigg(\frac{q_2(x_2y_2\mtwo)}{p_2(x_2y_2\mtwo)}\bigg)^I &\cdot&  \abs{\E_{q_2(x_2y_2\vert w)}\left[(-1)^{ g}\right]}^{-12I/\delta}\notag \\
    =&\frac{q(x|ym)}{p(x|y)} \frac{q(y|xm)}{p(y|x)} &\cdot&  \bigg(\frac{q(xym)}{p(xym)}\bigg)^I &\cdot&   \abs{\E_{q(xy\vert w)}\left[(-1)^{f\oplus  g}\right]}^{-12I/\delta}\notag \\
    \leq &\frac{q(x|ym)}{p(x|y)} \frac{q(y|xm)}{p(y|x)} &\cdot&  \bigg(\frac{q(xym)}{p(xym)}\bigg)^I &\cdot& \abs{\E_{q(xy\vert m)}\left[(-1)^{f\oplus  g}\right]}^{-12I/\delta} \cdot O(q(G|m))^{12I} \notag \\
    \leq &2^{\marg(q,p,f\oplus g)} &\cdot& O(q(G|m))^{12I} && \notag
    \end{align}

    In this product, the quantity in the first line does not depend on the choice of $x_2$, and the quantity in the second line does not depend on $y_1$. Thus, for every fixed value of $w$, we obtain:

    \begin{align}
    &\abs{\E_{q_1(x_1y_1\vert w)}\left[(-1)^{f}\right]}^{-12I/\delta} \cdot \sup_{y_1} \bigg(\frac{q_1(x_1y_1\mone)}{p_1(x_1y_1\mone)}\bigg)^I \cdot  \frac{q_1(x_1|y_1\mone)}{p_1(x_1|y_1)} \frac{q_1(y_1|x_1\mone)}{p_1(y_1|x_1)}   \notag \\
     \times & \abs{\E_{q_2(x_2y_2\vert w)}\left[(-1)^{ g}\right]}^{-12I/\delta}\cdot \sup_{x_2} \bigg(\frac{q_2(x_2y_2\mtwo)}{p_2(x_2y_2\mtwo)}\bigg)^I \cdot \frac{q_2(x_2|y_2\mtwo)}{p_2(x_2|y_2)} \frac{q_2(y_2|x_2\mtwo)}{p_2(y_2|x_2)}  \notag \\
    \leq & 2^{\marg(q,p,f\oplus g)} \cdot O(q(G|m))^{12I}. \label{eqn:subadd1}
    \end{align}
    
    Let $L \subseteq G$ be a subset whose indicator function $1_L(xym)$ depends only on $w$, such that $1_{L}(w)=1$ if and only if 

\begin{align}
    & \left(\abs{\E_{q_1(x_1y_1\vert w)}\left[(-1)^{f}\right]}^{-12I/\delta} \cdot \sup_{y_1} \bigg(\frac{q_1(x_1y_1\mone)}{p_1(x_1y_1\mone)}\bigg)^I \cdot  \frac{q_1(x_1|y_1\mone)}{p_1(x_1|y_1)} \frac{q_1(y_1|x_1\mone)}{p_1(y_1|x_1)} \right)^{1/\gamma}  \notag \\
    & \leq \left(\abs{\E_{q_2(x_2y_2\vert w)}\left[(-1)^{ g}\right]}^{-12I/\delta}\cdot \sup_{x_2} \bigg(\frac{q_2(x_2y_2\mtwo)}{p_2(x_2y_2\mtwo)}\bigg)^I \cdot \frac{q_2(x_2|y_2\mtwo)}{p_2(x_2|y_2)} \frac{q_2(y_2|x_2\mtwo)}{p_2(y_2|x_2)}\right)^{1/(1-\gamma)} \notag .
    \end{align}

    Let $U$ denote the set whose indicator function depends only on $m$, such that $1_U(m)=1$ if and only if $q(L |mG) \geq 1/2$. If $q(U) \geq 1/2$, we carry out the reduction in the first coordinate. Otherwise, we carry out the reduction in the second coordinate, using the complements of $U,L$ instead. Without loss of generality, we assume that $q(U) \geq 1/2$. By the definition of $U,L$, and by \Cref{eqn:subadd1}, for all $w$ in the support of $U \cap L$ we have 
\begin{align}
    & \abs{\E_{q_1(x_1y_1\vert w)}\left[(-1)^{f}\right]}^{-12I/\delta}  \cdot \bigg(\sup_{y_1} \frac{q_1(x_1y_1\mone)}{p_1(x_1y_1\mone)}\bigg)^I \cdot \frac{q_1(x_1|y_1\mone)}{p_1(x_1|y_1)} \frac{q_1(y_1|x_1\mone)}{p_1(y_1|x_1)}  \notag \\
    &\leq  2^{\gamma\cdot \marg(q,p,f\oplus g)} \cdot O(q(G|m))^{12\gamma \cdot I}. \label{eqn:subadd2}
    \end{align}

Our next barrier is that in $\marg(q_1,p_1,f)$ the advantage term is not exactly the same as what we have bounded in the above expressions; it might well be that for most $w$ consistent with $\mone$
    \begin{align}\Big|\E_{q_1(x_1y_1\vert \mone)}\Big[(-1)^f\Big]\Big| \ll \Big|\E_{q_1(x_1y_1\vert w)}\Big[(-1)^f\Big]\Big|.\label{eqn:wlarge}
    \end{align}

    To resolve this issue, we condition on a dense subset $B$ of the $w$'s such that given any two $w,w' \in B$ that are consistent with the same $m$, 
    $$\Big|\E_{q_1(x_1y_1\vert w)}\Big[(-1)^f\Big]\Big| \geq \frac{1}{2}\cdot  \Big|\E_{q_1(x_1y_1\vert w')}\Big[(-1)^f\Big]\Big|.$$
    This will ensure that \Cref{eqn:wlarge} does not happen. 
    To find this subset $B$, we first prune away some problematic points to ensure all advantage terms in \Cref{eqn:subadd2} are reasonably large. This is accomplished by \Cref{claim:largeadv} below. 
    
\begin{claim} \label{claim:largeadv} There are subsets $U' \subseteq U, L' \subseteq L$ such that $1_{U'}(xym)$ only depends on $m$,  $1_{L'}(xym)$ only depends on $w$, $q(U') \geq 1/4$, and for all $mw$ in the support of  $U' \cap L'$, we have $q(L'|mG)\geq 1/4$ and $$
     \abs{\E_{q_1(x_1y_1\vert w)}\left[(-1)^{f}\right]}^{-1}
    \leq  \alpha,$$ for some $\alpha \leq 2^{O(\marg(q,p,f\oplus g)/I)}\cdot O(1).$
\end{claim}

We defer the proof of \Cref{claim:largeadv} to the end of this section. Assuming the claim, we can now bucket the $w$ according to their advantage. For each fixing of $\mone$, partition the space of $w$ consistent with  $\mone$ into disjoint buckets  according to the sign of $\E_{q_1(x_1y_1\vert w)}\left[(-1)^{f}\right]$, and the value of $$\left \lceil \log \abs{\E_{q_1(x_1y_1\vert w)}\left [(-1)^{f}\right]} \right \rceil.$$
   There can be at most $O( \log \alpha )$ such buckets, and so by picking the heaviest bucket for each $\mone$, we obtain a set $B\subseteq L'$ whose indicator function $1_B(x_1y_1\mone)$ is determined by $w$, such that for every $\mone$,
   \begin{align} \label{eqn:bucketbound} q_1(B\vert \mone L') \geq \frac{1}{ O(\log \alpha )},\end{align}
and moreover, for every $w\mone$ in the support of $B$, 

    \begin{align}\label{eq:B-mone-prop}
        \abs{\E_{q_1(x_1y_1\vert \mone B)}\Big[(-1)^f\Big]}  \geq \frac{1}{2} \cdot \abs{\E_{q_1(x_1y_1\vert w)}\Big[(-1)^f\Big]}.
    \end{align}

    Let $R\subseteq B\subseteq L'$ be the rectangular set in $x_1y_1 \mone$ such that for every $\mone$, $R$ maximizes 
    \begin{align}\label{eq:def-R}
        q_1(R\vert \mone B)^\delta\cdot \Big|\E_{q_1(x_1y_1\vert  \mone R) }\Big[(-1)^{f}\Big]\Big|. 
    \end{align}

Define the rectangular distribution 
\begin{align}
    	r(x_1y_1\mone) &= q_1(x_1y_1\mone)\cdot \frac{1_{U'}(m) \cdot  1_{R}(x_1y_1\mone)}{q_1(U')\cdot q_1(L'\vert m)\cdot q_1(R\vert \mone L')}\notag \\
     &=  \frac{q_1(m)1_{U'}(m)}{q_1(U')} \cdot \frac{q_1(y_2|m) q_1(L'|\mone) }{ q_1(L'\vert m) } \cdot \frac{ q_1(x_1y_1\vert \mone) \cdot 1_{R}(x_1y_1\mone)}{q_1(L'\vert\mone) 
 \cdot q_1(R\vert \mone L')}\notag \\
 &=  q_1(m|U') \cdot q_1(y_2|mL') \cdot \frac{ q_1(x_1y_1\vert \mone) \cdot 1_{R}(x_1y_1\mone)}{
  q_1(R\vert \mone )}\notag \\
  &=  q_1(m|U') \cdot q_1(y_2|mL') \cdot q_1(x_1y_1|\mone R). \label{eqn:rdef}
    \end{align}

Because $r$ is defined as the  product of a rectangular distribution with a function that is also rectangular, $r$ 
 is rectangular. From the last line in \Cref{eqn:rdef}, it is clear that $r$ is a distribution. We have the following bound:
\begin{align}\label{eqn:rqratio}
\frac{r(x_1y_1\mone)}{q_1(x_1y_1\mone)} &=  \frac{1_{U'}(m) \cdot  1_{R}(x_1y_1\mone)}{q_1(U')\cdot q_1(L'\vert m)\cdot q_1(R\vert \mone L')} \notag \\
&\leq  \frac{O(1)}{q_1(L'\vert mG)\cdot q_1(G\vert m) \cdot q_1(R\vert \mone B) \cdot q_1(B\vert \mone L')}  \notag\\
&\leq \frac{O(\log \alpha )}{ q_1(G\vert m) \cdot q_1(R\vert \mone B)}, 
\end{align}
where here we used \Cref{claim:largeadv} and  \Cref{eqn:bucketbound}.
Apply \Cref{lem:trimming-1} with $a = r$, $b=q_1$ and $\kappa = 1/6$ to obtain a rectangular set $T$ with $r(T) \geq 1/2$ such that 
    \begin{align}\label{eq:bound-from-trimming}
    \frac{r(x_1\mone\vert T)}{q_1(x_1\mone)}, \frac{r(y_1\mone\vert T)}{q_1(y_1\mone)}, \frac{r(\mone \vert T)}{r(\mone)} \geq \frac{1}{6}.
    \end{align}
    
Finally, we define $r_1(x_1y_1\mone) = r(x_1y_1\mone\vert T)$. Because $r$ is a rectangular distribution and $T$ is a rectangular set, $r_1$ is a rectangular distribution. It only remains to bound $\marg(r_1,p_1,f)$. For all $x_1y_1 \mone$ in the support of $r_1$, we have
\begin{align}\label{eqn:rectboundf}
\frac{r_1(x_1y_1 \mone)}{p_1(x_1y_1\mone)} &= \frac{q_1(x_1y_1\mone)}{p_1(x_1y_1\mone) } \cdot \frac{r(x_1y_1\mone)}{q_1(x_1y_1\mone)} \cdot \frac{1}{r(T)}\notag \\
&\leq \frac{q_1(x_1y_1\mone)}{p(x_1y_1\mone)} \cdot \frac{O(\log \alpha )}{ q_1(G\vert m) \cdot q_1(R\vert \mone B)}, 
\end{align}
using \Cref{eqn:rqratio} and the fact that $r(T) \geq 1/2$. For the next term,
\begin{align} \label{eqn:infoboundf1}
\frac{r_1(x_1|y_1 \mone)}{p_1(x_1|y_1 \mone)} & = \frac{q_1(x_1|y_1 \mone)}{p_1(x_1|y_1 \mone)} \cdot \frac{r_1(x_1|y_1 \mone)}{q_1(x_1|y_1 \mone)} \notag \\
& = \frac{q_1(x_1|y_1 \mone)}{p_1(x_1|y_1 \mone)} \cdot \frac{r(x_1 y_1 \mone)}{q_1(x_1y_1 \mone)} \cdot \frac{1}{r(T)} \cdot \frac{q_1(y_1 \mone)}{r(y_1\mone|T)}\notag \\
& \leq \frac{q_1(x_1|y_1 \mone)}{p_1(x_1|y_1 \mone)}  \cdot \frac{O(\log \alpha )}{ q_1(G\vert m) \cdot q_1(R\vert \mone B)},
\end{align}
using \Cref{eqn:rqratio,eq:bound-from-trimming}, and the fact that $r(T) \geq 1/2$. The symmetric argument gives:
\begin{align} \label{eqn:infoboundf2}
\frac{r_1(y_1|x_1 \mone)}{p_1(y_1|x_1 \mone)} 
& \leq \frac{q_1(y_1|x_1 \mone)}{p_1(y_1|x_1 \mone)}  \cdot \frac{O(\log \alpha )}{ q_1(G\vert m) \cdot q_1(R\vert \mone B)}.
\end{align}

To bound the advantage, first note that 
\begin{align} \label{eqn:qtbound}
q_1(T\vert \mone R) & = \frac{q_1(T\vert R) \cdot q_1(\mone\vert T)}{q_1(\mone \vert R)} = \frac{r(T) \cdot r(\mone\vert T)}{r(\mone)} \geq  \frac{1}{12},
\end{align}
by \Cref{eq:bound-from-trimming}. For each $\mone$, we apply \Cref{lem:adv-preserving} with $v(x_1y_1\mone) = q_1(x_1y_1\mone\vert \mone B)$, $R$ and $Z = T$. Note here that $v(\mone)=1$. We obtain the bound:
        \begin{align}\label{eqn:finaladv}
    \Big|\E_{r_1(x_1y_1\vert \mone)}\big[(-1)^f\big]\Big| &= \Big|\E_{q_1(x_1y_1\vert \mone T)}\big[(-1)^f\big]\Big| 
            \notag \\
            &\geq  \frac{1-\delta^2 - \delta/ q_1(T\vert  \mone R)}{q_1(R\vert  \mone B)^\delta} \cdot \Big|\E_{q_1(x_1y_1\vert \mone B)}\Big[(-1)^{f}\Big]\Big|\notag \\
            &\geq  \frac{\Omega(1)}{q_1(R\vert  \mone B)^\delta} \cdot \Big|\E_{q_1(x_1y_1\vert \mone B)}\Big[(-1)^{f}\Big]\Big|,
        \end{align}
        by the choice of $\delta$ and \Cref{eqn:qtbound}.

Now, we are ready to put all these bounds together to complete the proof of the theorem. By \Cref{eq:B-mone-prop,eqn:rectboundf,eqn:infoboundf1,eqn:infoboundf2,eqn:finaladv}, we get that for every $x_1y_1\mone$ in the support of $r_1$,

\begin{align*}
&\frac{r_1(x_1\vert y_1 \mone)}{p_1(x_1\vert y_1)}\cdot \frac{r_1(y_1\vert x_1 \mone)}{p(y_1\vert x_1)}\cdot \bigg(\frac{r_1(x_1y_1\mone)}{p_1(x_1y_1\mone)}\bigg)^{I} \cdot  \Big|\E_{r_1(x_1y_1\vert \mone)}\big[(-1)^f\big]\Big|^{-12I/\delta} \\
&\leq \frac{q_1(x_1\vert y_1 \mone)}{p_1(x_1\vert y_1)}\cdot \frac{q_1(y_1\vert x_1 \mone)}{p(y_1\vert x_1)}\cdot \bigg(\frac{q_1(x_1y_1\mone)}{p_1(x_1y_1\mone)}\bigg)^I \cdot  \Big|\E_{q_1(x_1y_1\vert w)}\big[(-1)^f\big]\Big|^{-12I/\delta} \\ & \qquad \times  \frac{O(\log(\alpha))^{I+2}}{q_1(G\vert m)^{I + 2} \cdot  q_1(R\vert \mone B)^{I+2-12I}}\\
&\leq 2^{\gamma\cdot \marg(q,p,f\oplus g)}\cdot O(\log(\alpha))^{3I}\cdot q_1(G\vert m)^{12\gamma\cdot I -I - 2}\cdot 2^{O(I)}\cdot q_1(R\vert \mone B)^{11I-2}\\
&\leq 2^{\gamma\cdot \marg(q,p,f\oplus g)} \cdot O(\log(\alpha))^{3I}\cdot q_1(G\vert m)^{3I - 2}\cdot  2^{O(I)}\\
&\leq 2^{\gamma\cdot \marg(q,p,f\oplus g)} \cdot O\left(\frac{\marg(q,p,f\oplus g)}{I}\right)^{3I},
\end{align*}
where in the last three inequalities we used \Cref{eqn:subadd2}, the fact that $I \geq 1$ and \Cref{claim:largeadv}. This implies that 
$$ \marg(r_1,p_1,f) \leq \gamma\cdot \marg(q,p,f \oplus g) + 3I \log \frac{\marg(q,p,f \oplus g)}{I} + O(I),$$
completing the proof of the theorem. 

\textit{Proof of \Cref{claim:largeadv}:} 
We have $$\E_{q_1(m)}\left[\frac{p_1(m)}{q_1(m)}\right] \leq 1,$$ and so by Markov's inequality, the total mass of  $m\in \mathsf{supp}(q)$ for which \begin{align} \frac{q_1(m)}{p_1(m)} \leq 1/4 \label{eqn:Urule} \end{align} is at most $1/4$. We delete all such $m$ from the support of $U$. We are left with a set $U'$ with \begin{align}q(U') \geq 1/2-1/4 =1/4.\label{eqn:Umass}\end{align}
Next, we delete $w$ from $L$ if either 
\begin{align}\frac{q_1(w\vert m)}{p_1(w\vert m)}<  \frac{q_1(G\vert m)}{8},\label{eqn:Lrule1}\end{align}
or \begin{align}\E_{q_1(y_1\vert w)}\left[\log \frac{q_1(x_1\vert y_1\mone)}{p_1(x_1\vert y_1)}\right]< \log  \frac{q_1(G\vert m)}{8}.\label{eqn:Lrule2}\end{align}

We claim that for all $m$ in the support of $U'$, $q(L'|mG) \geq 1/4$. To see this, observe that
\begin{align*}
q_1(G\vert m)\cdot  \E_{q_1(w\vert m G)}\left[\frac{p_1(w\vert m)}{q_1(w\vert m)}\right] \leq \E_{q_1(w\vert m )}\left[\frac{p_1(w\vert m)}{q_1(w\vert m)}\right]\leq 1,
\end{align*}
so Markov's inequality implies that for each $m$ the total mass of $w$ for which \Cref{eqn:Lrule1} is violated is at most $1/8$. By the concavity of the $\log$ function, the $w$ deleted because of \Cref{eqn:Lrule2} satisfy
\begin{align*}\log \E_{q_1(y_1\vert w)}\left[ \frac{p_1(x_1\vert y_1\mone)}{q_1(x_1\vert y_1)}\right] \geq \E_{q_1(y_1\vert w)}\left[\log \frac{p_1(x_1\vert y_1\mone)}{q_1(x_1\vert y_1)}\right]> \log  \frac{8}{q_1(G\vert m)}.\end{align*}

On the other hand, because $w$ determines $1_G$,
\begin{align*}
q_1(G\vert m)\cdot \E_{q_1(w\vert m G)} \left [\E_{q_1(y_1\vert w )}\left[\frac{p_1(x_1\vert y_1)}{q_1(x_1\vert y_1\mone)}\right] \right]
            &\leq  \E_{q_1(y_1 w\vert m) }\left[\frac{p_1(x_1\vert y_1)}{q_1(x_1\vert y_1\mone)}\right] \\
            &=  \E_{q_1(y\vert m) }\left[\E_{q_1(x_1 \vert ym) }\left[\frac{p_1(x_1\vert y_1)}{q_1(x_1\vert y_1\mone)}\right]\right] \leq 1,
\end{align*}
so once again, Markov's inequality implies that the total mass of $w$ deleted using this rule is at most $1/8$. This gives \begin{align}q(L'\vert mG) \geq 1/2 - 1/8-1/8 = 1/4. \label{eqn:Lmass}\end{align}

The result of these pruning steps is that we are left with large sets $U',L' \subseteq G$ such that for all $m,w$ that are consistent with $U',L'$, we have
\begin{align*}
            &\sup_{y_1}   \bigg(\frac{q_1(x_1y_1\mone)}{p_1(x_1y_1\mone)}\bigg)^I\cdot \frac{q_1(x_1\vert y_1\mone)}{p_1(x_1\vert y_1)}\cdot\frac{q_1(y_1\vert x_1\mone)}{p_1(y_1\vert x_1)}   \\
            &=  \bigg(\frac{q_1(m) \cdot q_1(w\vert m)}{p_1(m) \cdot p_1(w \vert m)}\bigg)^I  \cdot \sup_{y_1}  \bigg(\frac{q_1(y_1\vert w)}{p_1(y_1\vert w)}\bigg)^I \cdot  \frac{q_1(x_1\vert y_1\mone)}{p_1(x_1\vert y_1)}\cdot\frac{q_1(y_1\vert x_1\mone)}{p_1(y_1\vert x_1)} \\
            &=  \bigg(\frac{q_1(m) \cdot q_1(w\vert m)}{p_1(m) \cdot p_1(w \vert m)}\bigg)^I  \cdot \exp \Bigg( \sup_{y_1} \log \bigg(\bigg( \frac{q_1(y_1\vert w)}{p_1(y_1\vert w)} \bigg)^I\cdot  \frac{q_1(x_1\vert y_1\mone)}{p_1(x_1\vert y_1)}\cdot\frac{q_1(y_1\vert x_1\mone)}{p_1(y_1\vert x_1)} \bigg)\Bigg)\\
            &\geq  \bigg(\frac{q_1(m) \cdot q_1(w\vert m)}{p_1(m) \cdot p_1(w \vert m)}\bigg)^I  \cdot \exp \Bigg( \E_{q(y_1\vert w)} \log \bigg(\bigg( \frac{q_1(y_1\vert w)}{p_1(y_1\vert w)} \bigg)^I\cdot  \frac{q_1(x_1\vert y_1\mone)}{p_1(x_1\vert y_1)}\cdot\frac{q_1(y_1\vert x_1\mone)}{p_1(y_1\vert x_1)} \bigg)\Bigg),
            \end{align*}
            where here $\exp(z)$  denotes  $2^z$. Now we use the fact that all the $m,w$ violating \Cref{eqn:Urule,eqn:Lrule1,eqn:Lrule2} have been deleted and use \Cref{eq:div-non-neg} to bound
    \begin{align*}
            &\geq   \bigg(\frac{1}{4} \cdot \frac{q_1(G|m)}{8}\bigg)^I \\
            &\qquad \cdot \exp \Bigg(\E_{q_1(y_1\vert w)} \left [I\log \frac{q_1(y_1\vert w)}{p_1(y_1\vert w)} + \log \left(\frac{q_1(x_1\vert y_1\mone)}{p_1(x_1\vert y_1)}\right)+  \log \left( \frac{q_1(y_1\vert x_1\mone)}{p_1(y_1\vert x_1)}\right) \right]\Bigg)\\
            &\geq   \bigg(\frac{1}{4}\cdot \frac{q_1(G|m)}{8} \bigg)^{I} \cdot \exp \Bigg( \E_{q_1(y_1\vert w)} \left [ \log \left(\frac{q_1(x_1\vert y_1\mone)}{p_1(x_1\vert y_1)}\right) \right]\Bigg)\\
            &\geq    \Omega(q_1(G|m))^{1+I}. 
    \end{align*} 

    Combining this bound with  \Cref{eqn:subadd2}, we get that for all $w$ consistent with $U',L'$,
    \begin{align}
    \abs{\E_{q_1(x_1y_1\vert w)}\left[(-1)^{f}\right]}^{-12I/\delta}
    \leq  2^{\gamma\cdot \marg(q,p,f\oplus g)} \cdot O(q(G|m))^{12\cdot\gamma I - I - 1}, \notag 
    \end{align}
    so since $I \geq 1$ and $\gamma\geq 1/3$, this implies 
    \begin{align}
    \abs{\E_{q_1(x_1y_1\vert w)}\left[(-1)^{f}\right]}^{-1}
    \leq   O(2^{\marg(q,p,f\oplus g)\cdot (\delta\gamma/12I)}) = \alpha,  
    \label{eqn:abound}
    \end{align}
as required.

%% file: trimming.tex
 \section{Trimming and advantage preserving sets} \label{trimming}
    In this section, we gather a few lemmas about trimming rectangular sets to pass to subrectangles with nice features. The idea of trimming comes from the work of Yu \cite{huacheng}.
    
    \begin{lemma}\label{lem:trimming-1}
    For every $1>\kappa>0$, if $a(xym),b(xym)$ are two distributions, there exists a rectangular set $T$ such that  $a(T) \geq 1-3\kappa$ and for all $xym \in T$, we have
		\[ \frac{a(xm\vert T)}{b(xm)}, \frac{a(ym\vert T)}{b(ym)},\frac{a(m\vert T)}{a(m)} \geq \kappa.\]
    \end{lemma}
    \begin{proof}
    The set $T$ is constructed by an iterative process. Initially, $T$ is the set of all triples $xym$. In each iteration, if there is $xm$ such that
	\begin{align} \frac{a(xm\vert T)}{b(xm)} < \kappa,\label{eqn:trimrule}\end{align} then delete $xm$ from the support of $T$, if there is $ym$ such that $$\frac{a(ym\vert T)}{b(ym)} < \kappa,$$ then delete $ym$ from the support of $T$, and if there is $m$ such that $$\frac{a(m\vert T)}{a(m)} < \kappa,$$ then delete $m$ from the support of $T$. The process halts when there are no more elements to delete. Because the distributions we are working with have finite support, this process must eventually terminate. Initially, $T$ is rectangular, and each deletion step leaves us with another rectangular set $T$, so the final $T$ is also rectangular.

    Let us bound $a(T)$. For each pair $xm$ that was deleted from the support of $T$ because of \Cref{eqn:trimrule}, let  $T_{xm}$ denote the set $T$ right before $xm$ was deleted. If $xm$ was not deleted, let $T_{xm}$ denote the empty set. 
    
    The total mass deleted using \Cref{eqn:trimrule} is exactly 
     \begin{align*}
     \sum_{xm}a(xm T_{xm})=  \sum_{xm}a(T_{xm}) \cdot a(xm|T_{xm}) < \sum_{xm} \kappa \cdot b(xm)=\kappa.
    \end{align*}

    Similarly, the total mass deleted using each of the other rules is also at most $\kappa$. By the union bound, this proves that $a(T) \geq 1-3 \kappa$ when the process terminates.
    \end{proof}

%% file: argmax.tex
Next, we gather a couple of nice lemmas about finding subrectangles with nice properties. 

\begin{lemma}\label{lem:adv-preserving}
    For any distribution $v(xym)$ and a Boolean function $h(xy)$, suppose $R$ is a rectangular set maximizing
    \begin{align}
    v(R)^\delta\cdot \E_{v(m\vert R)}\Big|\E_{v(xy \vert mR)}\big[(-1)^h\big]\Big|.\label{eqn:rmax}
    \end{align} Then, for any rectangular $Z\subseteq R$, we have
    \[\E_{v(m\vert Z)}\Big|\E_{v(xy\vert mZ)}\big[(-1)^h\big]\Big| \geq \frac{1-\delta^2  - \delta/v(Z|R)}{v(R)^{\delta}} \cdot  \E_{v(m)}\Big|\E_{v(xy\vert m)}\big[(-1)^h\big]\Big|.\]
\end{lemma}
\begin{proof}
    Since $R$ and $Z$ are rectangular, we have  
    \[1_{R}(xym) = 1_A(xm) \cdot 1_B(ym),\]and \[ 1_Z(xym) = 1_{A'}(xm)\cdot 1_{B'}(ym),\] for appropriate sets $A,A'$ and $B,B'$. 
     $R$ can be partitioned into three rectangular sets, $Z_0 = Z, Z_1$ and $Z_2$, where
    \begin{align*}
        1_{Z_1}(xym) = 1_{A\setminus A'}(xm)\cdot 1_B(ym)
        \end{align*}
        and \begin{align*}
            1_{Z_2}(xym) = 1_{A'}(xm)\cdot 1_{B\setminus B'}(ym).
    \end{align*}

    By the triangle inequality, we get  
    \begin{align}
            &\E_{v(m\vert R)}\Big|\E_{v(xy\vert Rm)}\Big[(-1)^{h}\Big]\Big| \leq 
             \sum_{i=0}^2v(Z_i \vert R)\cdot \E_{v(m\vert Z_i)} \Big|\E_{v(xy\vert m Z_i)}\Big[(-1)^{h}\Big]\Big|\label{eqn:trianglL}
    \end{align}
   Let us bound the contribution of $Z_1,Z_2$:
    \begin{align*}
    &\sum_{i=1}^2  v(Z_i|R)\cdot \E_{v(m\vert Z_i)}\Big|\E_{v(xy\vert m Z_i)}\Big[(-1)^{h}\Big]\Big| \\&= \sum_{i=1}^2  v(Z_i|R)^{1-\delta}\cdot \Big(\frac{v(Z_i)}{v(R)}\Big)^\delta \cdot \E_{v(m\vert Z_i)}\Big|\E_{v(xy\vert m Z_i)}\Big[(-1)^{h}\Big]\Big| \\
    &\leq \sum_{i=1}^2  v(Z_i|R)^{1-\delta}\cdot  \E_{v(m\vert R)}\Big|\E_{v(xy\vert m R)}\Big[(-1)^{h}\Big]\Big|\tag{because $R$ is the maximizer of \Cref{eqn:rmax}}\\
    &\leq 2^\delta \cdot \Big(\sum_{i=1}^2  v(Z_i|R)\Big)^{1-\delta}\cdot  \E_{v(m\vert R)}\Big|\E_{v(xy\vert m R)}\Big[(-1)^{h}\Big]\Big| \tag{by H\"older's inequality}\\
    &= 2^\delta \cdot \Big(1-v(Z|R)\Big)^{1-\delta}\cdot  \E_{v(m\vert R)}\Big|\E_{v(xy\vert m R)}\Big[(-1)^{h}\Big]\Big|.
    \end{align*}
    Using the inequalities $(1-t)^{1-\delta} \leq 1-t(1-\delta))$ and $2^\delta\leq 1+\delta$ which hold for $t,\delta \in [0,1]$:
    \begin{align*}
    &\leq  (1+\delta) \cdot (1-(1-\delta) v(Z|R))\cdot  \E_{v(m\vert R)}\Big|\E_{v(xy\vert m R)}\Big[(-1)^{h}\Big]\Big|. 
    \end{align*}

    Putting this back into \Cref{eqn:trianglL} and reaarranging, we get:
    
    \begin{align*}
            \E_{v(m\vert Z)}\Big|\E_{v(xy\vert mZ)}\Big[(-1)^{h}\Big]\Big| &\geq \frac{-\delta +(1-\delta^2)v(Z|R)}{v(Z|R)} \E_{v(m\vert R)} \Big|\E_{v(xy\vert m R)}\Big[(-1)^{h}\Big]\Big|\\
            &\geq  (1-\delta^2 - \delta/v(Z|R)) \cdot v(R)^{-\delta}\cdot \E_{v(m)}  \Big|\E_{v(xy\vert m )}\Big[(-1)^{h}\Big]\Big|,
    \end{align*}
    where we again used the fact that $R$ maximizes \Cref{eqn:rmax}.
\end{proof}

     \begin{lemma}\label{lem:advantage-preserving-1}Let $q(xym)$ be a rectangular distribution, with $x=x_1x_2$ and $y=y_1y_2$. Let $f(x_1y_1),g(x_2y_2)$ be Boolean functions.
    Let $G$ be a subset of triples $xym$ such that the indicator function $1_{G}(xym)$ depends only on $w=x_1y_2 m$, and for each $m$, $G$ maximizes 
    \begin{align}
    q(G\vert m)^\delta\cdot \Big|\E_{q(xy\vert mG)}\Big[(-1)^{f\oplus g}\Big]\Big|,\label{eqn:gmax2}
    \end{align} among all such sets.
        Then for any $w$ in the support of $G$, we have
        \[\abs{\E_{q(xy\vert w)}\left[(-1)^{f\oplus g}\right]}
        \geq (1-\delta)\cdot q(G\vert m)^{-\delta}\cdot \abs{\E_{q(xy\vert m)}\left[(-1)^{f\oplus g }\right]}.\]
    \end{lemma}
    \begin{proof}
    Fix $w=x_1y_2m$ and define $G' \subseteq G$ to be the subset of $G$ obtained by deleting all triples $xym$ consistent with $w$. Using the triangle inequality, we can write    
    \begin{align*}
        \Big|\E_{q(xy\vert mG)}\Big[(-1)^{f\oplus g}\Big]\Big| 
        &\leq q(w\vert  mG)\cdot \Big|\E_{q(xy\vert w)}\Big[(-1)^{f\oplus g}\Big]\Big| + 
         q(G'\vert mG)\cdot \Big|\E_{q(xy\vert mG')}\Big[(-1)^{f\oplus g}\Big]\Big| \\
         &= q(w\vert  mG)\cdot \Big|\E_{q(xy\vert w)}\Big[(-1)^{f\oplus g}\Big]\Big| + 
         q(G'\vert mG)^{1-\delta} \cdot \frac{q(G'\vert m)^{\delta}}{q(G\vert m)^\delta}\cdot \Big|\E_{q(xy\vert mG')}\Big[(-1)^{f\oplus g}\Big]\Big| \\
         &\leq q(w\vert m G)\cdot \Big|\E_{q(xy\vert w)}\Big[(-1)^{f\oplus g}\Big]\Big| + 
         q(G'\vert mG)^{1-\delta} \cdot  \Big|\E_{q(xy\vert mG)}\Big[(-1)^{f\oplus g}\Big]\Big|,
    \end{align*}
    where in the last line we used the fact that $G$ is the maximizer of \Cref{eqn:gmax2}. 

Because $q(G'|mG) = 1-q(w|mG)$, and using the inequality $(1-t)^{\gamma} \leq 1-t\gamma$, which holds for $t,\gamma \in [0,1]$, we obtain 

\begin{align*}
        \Big|\E_{q(xy\vert mG)}\Big[(-1)^{f\oplus g}\Big]\Big| 
         &\leq q(w\vert  mG)\cdot \Big|\E_{q(xy\vert w)}\Big[(-1)^{f\oplus g}\Big]\Big| + (1- 
         (1-\delta) q(w\vert mG)) \cdot  \Big|\E_{q(xy\vert mG)}\Big[(-1)^{f\oplus g}\Big]\Big|.
    \end{align*}

    Rearranging gives:
    \begin{align*}
    \Big|\E_{q(xy\vert w)}\Big[(-1)^{f\oplus g}\Big]\Big| &\geq (1-\delta) \cdot \Big|\E_{q(xy\vert mG)}\Big[(-1)^{f\oplus g}\Big]\Big|\\
    & \geq (1-\delta) \cdot q(G|m)^{-\delta}\cdot  \Big|\E_{q(xy\vert m)}\Big[(-1)^{f\oplus g}\Big]\Big|,
    \end{align*}

    where in the second inequality  we once again used the fact that $G$ is the maximizer of \Cref{eqn:gmax2}.
    \end{proof}

%% file: consequences.tex
\section{Consequences of small marginal information}\label{consequences}
Let $q$ be a rectangular distribution achieving $\marg_I(p,f)$. Since $q$ is rectangular, we can write 
\begin{align} \label{eq:g1g2}
\frac{q(xym)}{p(xym)} = \frac{\mu(xy)\cdot A(xm)\cdot B(ym)}{\mu(xy)\cdot p(m_0)\cdot \prod_{i=1,3,5,\dotsc}p(m_i\vert xm_{< i})\cdot p(m_{i+1}\vert ym_{\leq i})} = g_1(xm)\cdot g_2(ym),
\end{align}
for appropriate functions $g_1$ and $g_2$. 

For every $K\geq 1$, define the sets
\begin{align}
S_K &= \{xym: \abs{\lceil \log g_1(xm)\rceil + \log g_2(ym)} \leq 3(\marg_I(p,f) + KI)/I\} \label{eq:def-S},\\
R_K &= \{xym: p(m_1\vert xm_0) \leq 2^{6(\marg_I(p,f) + KI)}\cdot p(m_1\vert ym_0)\}. \label{eq:R}.
\end{align}

\begin{proposition}
    For $xym\in S_K$, 
    \begin{align}\label{eq:S-ub-lb}
     -\frac{3(\marg_I(p,f)+KI)}{I} - 1 \leq \log \frac{q(xym)}{p(xym)} \leq \frac{3(\marg_I(p,f) + KI)}{I}.
    \end{align}
\end{proposition}
\begin{proof}
    Because $\log (q(xym)/p(xym)) = \log g_1(x m) + \log g_2(y m)$, 
    \begin{align*}
    \log \frac{q(xym)}{p(xym)} \geq \lceil \log g_1(x m)\rceil + \log g_2(y m) - 1\geq - \frac{3(\marg_I(p,f) + KI)}{I} - 1,
    \end{align*}and 
    \begin{align*}
    \log \frac{q(xym)}{p(xym)} \leq \lceil \log g_1(x m)\rceil + \log g_2(y m) \leq \frac{3(\marg_I(p,f) + KI)}{I}.
    \end{align*}
\end{proof}

\begin{claim}\label{claim:q(S)-lb}
If $K\geq 3$, $q(S_K^c), q(R_K^c\vert S_K) \leq 5\cdot2^{-(\marg_I(p,f) + KI)/I}$. 
\end{claim}
\begin{proof}
    Define
    \begin{align*}
    G_1 &= \{xym: q(x\vert ym) \geq 2^{-(\marg_I(p,f) + KI)/I}\cdot p(x\vert y)\},\\
    G_2 &= \{xym: q(y\vert xm) \geq 2^{-(\marg_I(p,f) + KI)/I}\cdot p(y\vert x)\},\\
    G_3 &= \{xym: q(xym) \geq 2^{-3(\marg_I(p,f) + KI)/I}\cdot p(xym)\}, \text{ and }\\
    G_4 &= \{xym: q(x\vert ym) \geq 2^{-(\marg_I(p,f) + KI)/I}\cdot p(x\vert m_0m_1 y)\}.
    \end{align*}

If $xym\in G_1\cap G_2$, then
    \begin{align*}
    \marg_I(p,f) &\geq \log \Bigg(\frac{q(x\vert ym)}{p(x\vert y)}\cdot\frac{q(y\vert xm)}{p(y\vert x)}\cdot  \left(\frac{q(xym)}{p(xym)}\right)^I  \cdot     \bigg|\E_{q(xy\vert m)}\left[(-1)^{f}\right]\bigg|^{-12I/\delta} \Bigg) \\
    &\geq -\frac{2(\marg_I(p,f)+KI)}{I} +  I\cdot \log \frac{q(xym)}{p(xym)} \tag{because $xym \in G_1 \cap G_2$}\\
    &\geq  -2\marg_I(p,f)-2KI +  I\cdot (\lceil\log g_1(xm) \rceil + \log g_2(ym) -1) \tag{using $I\geq 1$}
    \end{align*}
    and rearranging this and using the fact that $KI \geq 1$ gives 
\begin{align*}
    \lceil\log g_1(xm) \rceil + \log g_2(ym) \leq   3(\marg_I(p,f)+KI)/I. 
\end{align*}
    
    Moreover, for $xym\in G_3$,
    \[\lceil\log g_1(xm) \rceil + \log g_2(ym) \geq   \log\frac{q(xym)}{p(xym)} \geq -3(\marg_I(p,f)+KI)/I,\]
    so we have $G_1\cap G_2\cap G_3 \subseteq S_K$. We shall prove that $q(S_K^c) \leq 3\cdot2^{-(\marg_I(p,f) + KI)/I}$ by proving that  $q(G_1^c), q(G_2^c), q(G_3^c)$ and $q(G_4^c)$ are all less than $2^{-(\marg_I(p,f) + KI)/I}$. We have
    \begin{align*}
        q(G_1^c) = \sum_{xym\notin G_1}q(xym) &< \sum_{xym\notin G_1}q(ym)\cdot p(x\vert y) \cdot 2^{-(\marg_I(p,f) + KI)/I} \\
        &\leq 2^{-(\marg_I(p,f) + KI)/I}\cdot\sum_{xym}q(ym)\cdot p(x\vert y) \\&\leq  2^{-(\marg_I(p,f) + KI)/I},
    \end{align*}
    and similar calculations show that $q(G_2^c),q(G_3^c), q(G_4^c)<  2^{-(\marg_I(p,f)+KI)/I}$. 

    It only remains to bound $q(R_K^c\vert S_K)$. We have
    \begin{align*}
    \frac{p(m_1\vert xm_0)}{p(m_1\vert ym_0)} = \frac{p(m_1\vert xym_0)}{p(m_1\vert ym_0)} 
    = \frac{p(x\vert ym_0m_1)}{p(x\vert ym_0)} 
    = \frac{p(x\vert ym_0m_1)}{q(x\vert ym)}\cdot \frac{q(x\vert ym)}{p(x\vert y)}, 
    \end{align*} 
    so, for every $xym \in G_2\cap G_3\cap G_4$,
    \begin{align*}
    \marg_I(p,f) &\geq \log \Bigg(\frac{q(x\vert ym)}{p(x\vert y)}\cdot\frac{q(y\vert xm)}{p(y\vert x)}\cdot  \left(\frac{q(xym)}{p(xym)}\right)^I  \cdot     \bigg|\E_{q(xy\vert m)}\left[(-1)^{f}\right]\bigg|^{-12I/\delta} \Bigg) \\
    &\geq \log \Bigg(\frac{p(m_1\vert xm_0)}{p(m_1\vert ym_0)} \cdot \frac{q(x\vert ym)}{p(x \vert y m_0m_1)}\cdot\frac{q(y\vert xm)}{p(y\vert x)}\cdot  \left(\frac{q(xym)}{p(xym)}\right)^I  \Bigg) \\
    &\geq  \log \frac{p(m_1\vert xm_0)}{p(m_1\vert ym_0)} - \frac{2(\marg_I(p,f) + KI)}{I} - 3(\marg_I(p,f) + KI)\\
    &\geq \log \frac{p(m_1\vert xm_0)}{p(m_1\vert ym_0)} - 5(\marg_I(p,f) + KI),
    \end{align*}
    since $I\geq 1$. Rearranging, we get $p(m_1\vert xm_0) \leq 2^{6(\marg_I(p,f) +KI)}\cdot p(m_1\vert ym_0)$, so $G_2\cap G_3\cap G_4 \subseteq R_K$. The union bound then gives:
    \begin{align*}
        q(R_K^c\vert S_K) 
        &\leq \frac{q(G_2^c) + q(G_3^c) + q(G_4^c)}{q(S_K)} 
        < \frac{3\cdot 2^{-(\marg_I(p,f) + KI)/I}}{1 - 3\cdot 2^{-(\marg_I(p,f) + KI)/I}} \leq 5\cdot2^{-3(\marg_I(p,f) + KI)/I},
    \end{align*}
    since $K\geq 3$.
\end{proof}

An argument analogous to the one in the previous claim allows us to obtain similar bounds if marginal information cost is replaced by external marginal information cost:

\begin{claim}\label{claim:mext-bounds}
    Let $q$ be a rectangular distribution achieving $\mext_I(p,f)$ and let $g_1,g_2$ be as defined in \Cref{eq:g1g2}. For every $K$, define 
    \begin{align}
        S_K &= \{xym: \abs{\lceil \log g_1(xm)\rceil + \log g_2(ym)} \leq 3(\mext_I(p,f) + KI)/I\} \text{ and } \label{eq:mext-def-S}\\
        R_K &= \{xym: p(m_1\vert xm_0) \leq 2^{5(\mext_I(p,f) + KI)}\cdot p(m_1\vert m_0)\}. \label{eq:mext-R}
    \end{align}
    Then, for all $K\geq 2$, it holds that $q(S_K^c), q(R_K^c\vert S_K) \leq 4\cdot2^{-(\mext_I(p,f) + KI)/I}$. 
\end{claim}
\begin{proof}
    Define
    \begin{align*}
    G_1 &= \{xym: q(xy\vert m) \geq 2^{-(\mext_I(p,f) + KI)/I}\cdot p(xy)\},\\
    G_2 &= \{xym: q(xym) \geq 2^{-3(\mext_I(p,f) + KI)/I}\cdot p(xym)\}, \\
    G_3 &= \{xym: q(xy\vert m) \geq 2^{-(\mext_I(p,f) + KI)/I}\cdot p(xy\vert m_0m_1)\}.
    \end{align*}

    For $xym\in G_1 \cap G_2$, we have
    \begin{align*}
    \mext_I(p,f) &\geq \log \Bigg(\frac{q(xy\vert m)}{p(xy)}\cdot  \left(\frac{q(xym)}{p(xym)}\right)^I  \cdot     \bigg|\E_{q(xy\vert m)}\left[(-1)^{f}\right]\bigg|^{-12I/\delta} \Bigg) \\
    &\geq -\frac{(\mext_I(p,f)+KI)}{I} +  I\cdot  \log \frac{q(xym)}{p(xym)}\\
    &\geq  -(\mext_I(p,f)+KI) +  I\cdot (\lceil\log g_1(xm) \rceil + \log g_2(ym) -1),
    \end{align*}
    since $K,I \geq 1$. Rearranging gives 
    \[\lceil\log g_1(xm) \rceil + \log g_2(ym) \leq  \frac{3(\mext_I(p,f)+KI)}{I}.\]
    Moreover, for $xym\in G_2$
    \[\lceil\log g_1(xm) \rceil + \log g_2(ym) \geq \log\frac{q(xym)}{p(xym)} \geq -\frac{3(\mext_I(p,f)+KI)}{I},\]
    proving that $G_1\cap G_2 \subseteq S_K$.

    We show that $q(G_1^c), q(G_2^c)$ and $q(G_3^c)$ are all less than $2^{-(\mext_I(p,f) + KI)/I}$, which  implies that $q(S_K^c) \leq 2\cdot2^{-(\mext_I(p,f) + KI)/I}$ as desired. To see the upper bound on $q(G_1^c)$, we may write
    \begin{align*}
        q(G_1^c) = \sum_{xym\notin G_1}q(xym) &< \sum_{xym\notin G_1}q(m)\cdot p(xy) \cdot 2^{-3(\mext_I(p,f) + KI)/I} 
        \leq 2^{-(\mext_I(p,f) + KI)/I}.
    \end{align*}
    A similar calculation shows that $q(G_2^c)$ and $q(G_3^c)<  2^{-(\mext_I(p,f)+KI)/I}$.

    Now, we prove that $q(R_K^c\vert S_K) \leq 5\cdot2^{-(\mext_I(p,f) + KI)/I}$. We have
    \begin{align*}
    \frac{p(m_1\vert xm_0)}{p(m_1\vert m_0)} = \frac{p(m_1\vert xym_0)}{p(m_1\vert m_0)} 
    = \frac{p(xy\vert m_0m_1)}{p(xy)} 
    = \frac{p(xy\vert m_0m_1)}{q(xy\vert m)}\cdot \frac{q(xy\vert m)}{p(xy)}, 
    \end{align*} 
    so for every $xym \in G_2\cap G_3$,
    \begin{align*}
    \mext_I(p,f) &\geq \log \Bigg(\frac{q(xy\vert y)}{p(xy)}\cdot  \left(\frac{q(xym)}{p(xym)}\right)^I  \cdot     \bigg|\E_{q(xy\vert m)}\left[(-1)^{f}\right]\bigg|^{-12I/\delta} \Bigg) \\
    &\geq \log \Bigg(\frac{p(m_1\vert xm_0)}{p(m_1\vert m_0)} \cdot \frac{q(xy \vert m)}{p(xy \vert m_0m_1)}\cdot  \left(\frac{q(xym)}{p(xym)}\right)^I  \Bigg) \\
    &\geq  \log \frac{p(m_1\vert xym_0)}{p(m_1\vert m_0)} - \frac{(\mext_I(p,f) + KI)}{I} - (\mext_I(p,f) + KI)\\
    &\geq \log \frac{p(m_1\vert xym_0)}{p(m_1\vert m_0)} - 4(\mext_I(p,f) + KI),
    \end{align*}
    since $I\geq 1$. Rearranging, we get $p(m_1\vert xym_0) \leq 2^{5(\mext_I(p,f) +KI)}\cdot p(m_1\vert m_0)$ for all $xym \in G_2\cap G_3$, and so $G_2\cap G_3 \subseteq R_K$. The union bound then gives:
    \begin{align*}
        q(R_K^c\vert S_K) 
        &\leq \frac{q(G_2^c) + q(G_3^c)}{q(S_K)} 
        < \frac{2\cdot 2^{-(\mext_I(p,f) + KI)/I}}{1 - 2\cdot 2^{-(\mext_I(p,f) + KI)/I}} \leq 4\cdot2^{-(\mext_I(p,f) + KI)/I},
    \end{align*}
    since $K\geq 2$.
\end{proof}

For the bounded-round simulation protocol, we need the following claim.  

\begin{claim}\label{claim:R-dense-in-S}
    Let $K \geq 3$, and $S_K$ be the set defined in \Cref{eq:def-S}.
    Let $p(xym)$ be an $r$-round protocol and define 
    \begin{align*}
        T_K = \Big\{xym: \forall i,  \frac{p(m_i\vert xym_{< i})}{p(m_i\vert ym_{< i})}, \frac{p(m_i\vert xym_{< i})}{p(m_i\vert xm_{< i})} \leq 2^{14(\marg_I(p,f) + KI)}\cdot (r+1)^5 \Big\}.
    \end{align*}
    Then  $q(T_K^c\vert S_K)\leq 22\cdot 2^{-(\marg_I(p,f) + KI)/I}$.
\end{claim}
\begin{proof}
    Define the sets
    \begin{align*}
        G_1 &= \{xm: \forall i, q(xm_{\leq i}) \geq 2^{-(\marg_I(p,f) + KI)/I}\cdot (r+1)^{-1}\cdot p(xm_{\leq i})\},\\
        G_1'&= \{ym: \forall i: q(ym_{\leq i}) \geq 2^{-(\marg_I(p,f) + KI)/I}\cdot (r+1)^{-1}\cdot p(ym_{< i})\},\\
        G_2 &= \{xym: \forall i, q(x\vert ym_{\leq i}) \geq 2^{-(\marg_I(p,f) + KI)/I}\cdot (r+1)^{-1}\cdot p(x\vert y)\},\\
        G_2'&= \{xym: \forall i, q(y\vert xm_{\leq i}) \geq 2^{-(\marg_I(p,f) + KI)/I}\cdot (r+1)^{-1}\cdot p(y\vert x)\},\\
        G_3 &= \{xym: \forall i, q(x\vert ym) \geq 2^{-(\marg_I(p,f) + KI)/I}\cdot (r+1)^{-1}\cdot p(x\vert ym_{\leq i})\},\\
        G_3' &= \{xym: \forall i, q(y\vert xm) \geq 2^{-(\marg_I(p,f) + KI)/I}\cdot (r+1)^{-1}\cdot p(y\vert xm_{\leq i})\},\\
        G_4 &= \{xym: \forall i, q(m_{\geq i}\vert xym_{< i}) \geq 2^{-(\marg_I(p,f) + KI)/I}\cdot (r+1)^{-1}\cdot p(m_{\geq i}\vert xym_{< i})\}.
    \end{align*}
    We claim that 
    \begin{align}\label{eq:R-large}
        \bigcap_{j=1}^3(G_j\cap G_j') \cap G_4\cap S_K \subseteq T_K.
    \end{align}
    For $xym \in  G_2' \cap G_2 \cap S_K$,
 \begin{align*}
        \marg_I(p,f) &\geq \log \Bigg(\frac{q(x\vert ym)}{p(x\vert y)}\cdot\frac{q(y\vert xm)}{p(y\vert x)}\cdot  \left(\frac{q(xym)}{p(xym)}\right)^I  \cdot     \bigg|\E_{q(xy\vert m)}\left[(-1)^{f}\right]\bigg|^{-12I/\delta} \Bigg) \\
        &\geq  \log \frac{q(x\vert ym)}{p(x\vert y)} - \frac{(\marg_I(p,f) + KI)}{I} - \log (r+1) - 3(\marg_I(p,f) + KI) - I \\
        &\geq \log\frac{q(x\vert ym)}{p(x\vert y)} - 4\marg_{I}(p,f) - 5KI - \log (r+1),
    \end{align*}
    becase $I\geq 1$, and by the definition of $G_2'$ and \Cref{eq:S-ub-lb}. Rearranging implies the first inequality below, and the second has a similar proof:
    \begin{align}\label{eq:info-ubs}
        \frac{q(x\vert ym)}{p(x\vert y)}, \frac{q(y\vert xm)}{p(y\vert x)} \leq 2^{5(\marg_I(p,f) + KI)/I}\cdot (r+1).
    \end{align}
   

    By \Cref{eq:info-ubs}, for  $xym\in \bigcap_{j=1}^3(G_j\cap G_j') \cap G_4\cap S_K$ and all $i$,
    \begin{align*}
        \frac{p(m_{\leq i}\vert xy)}{p(m_{\leq i}\vert y)} =  \frac{p(x\vert ym_{\leq i})}{p(x\vert y)} =  \frac{p(x\vert ym_{\leq i})}{q(x\vert ym)}\cdot \frac{q(x\vert ym)}{p(x\vert y)} \leq 2^{(\marg_I(p,f) + KI)/I}\cdot 2^{5(\marg_I(p,f) + KI)}\cdot (r+1)^2.
    \end{align*}
    Moreover, 
    \begin{align*}
        \frac{p(m_{\leq i}\vert xy)}{p(m_{\leq i}\vert y)} &=  \frac{p(x\vert ym_{\leq i})}{p(x\vert y)} \\
        &= \frac{p(x\vert ym_{\leq i})}{q(x\vert ym_{\leq i})}\cdot \frac{q(x\vert ym_{\leq i})}{p(x\vert y)} \\
        &= \frac{p(xym_{\leq i})}{q(xym_{\leq i})}\cdot \frac{q(ym_{\leq i})}{p(ym_{\leq i})}\cdot \frac{q(x\vert ym_{\leq i})}{p(x\vert y)} \\
        &= \frac{p(xym)}{q(xym)}\cdot \frac{q(m_{> i}\vert xym_{\leq i})}{p(m_{> i}\vert xym_{\leq i})}\cdot \frac{q(ym_{\leq i})}{p(ym_{\leq i})}\cdot \frac{q(x\vert ym_{\leq i})}{p(x\vert y)} \geq  \frac{2^{-6(\marg_I(p,f) + KI)/I}}{(r+1)^3},
    \end{align*}
    where we used \Cref{eq:S-ub-lb} as well as the definitions of $G_4,G_1'$ and $G_2$ in the last step. Thus,
    \begin{align*} 
    \frac{p(m_i\vert xym_{< i})}{p(m_i\vert ym_{< i})} = \frac{p(m_{\leq i}\vert xy)}{p(m_{\leq i}\vert y)}\cdot \frac{p(m_{< i}\vert y)}{p(m_{< i}\vert xy)} &\leq 2^{9(\marg_I(p,f) + KI)/I}\cdot 2^{5(\marg_I(p,f) + KI)}\cdot (r+1)^5 \\
    &\leq 2^{14(\marg_I(p,f) + KI)}\cdot (r+1)^5,
    \end{align*}
    since $I\geq 1$. A similar calculation shows that $\frac{p(m_i\vert xym_{< i})}{p(m_i\vert xm_{< i})} < 2^{14(\marg_I(p,f) + KI)}\cdot(r+1)^5$. We conclude that \Cref{eq:R-large} holds. 

Next,  we show that $q(G_4^c) < 2^{-(\marg_I(p,f) + KI)/I}$. Define
    \begin{align*}
    t(xym) & = \begin{cases}
        \min\{i: q(m_{\geq i}\vert xym_{< i})< \frac{2^{-(\marg_I(p,f) + KI)/I}\cdot  p(m_{\geq i}\vert xym_{< i})}{r+1}\} & \text{if such $i$ exists,}\\
        \bot & \text{otherwise.}
    \end{cases}
    \end{align*}
    We have,
    \begin{align*}
        q(G_4^c)=q(t \neq \bot) 
        &= \sum_{i=0}^r\sum_{\substack{xym,\\ t(xym) = i}}q(xym) \\
        &< \frac{2^{-(\marg_I(p,f) + KI)/I}}{r+1}\cdot \sum_{i=0}^r\sum_{\substack{xym\\ t(xym) = i}}q(xym_{< i})\cdot p(m_{\geq i}\vert xym_{<i}) \\
        &\leq 2^{-(\marg_I(p,f) + KI)/I}.
    \end{align*}
    A similar argument shows that $q(G_j^c),q(G_j'^c)  < 2^{-(\marg_I(p,f) + KI)/I}$, for all $j\in \{1,2,3\}$. 
    Thus, we can bound
    \begin{align*}
        q(T_K^c\vert S_K) &\leq \sum_{j=1}^3\frac{q(G_j^c) + q(G_j'^c)}{q(S_K)} + \frac{q(G_4^c) + q(S_K^c)}{q(S_K)} \\
        &\leq 11\cdot \frac{2^{-(\marg_I(p,f) + KI)/I}}{1 - 2^{-(\marg_I(p,f) + KI)/I+2}} \leq 22\cdot 2^{-(\marg_I(p,f) + KI)/I},
    \end{align*}
    where we used \Cref{claim:q(S)-lb} and the fact that $K\geq 3$.
\end{proof}

\begin{claim}\label{claim:R-bound-for-Braverman}
    For any $K\geq 1$, let $S_K$ be the set defined in \Cref{eq:def-S} and define 
    \begin{align}\label{eq:def-R-b}
    T_K &= \{xym: p(m\vert xy) \leq 2^{6(\marg_I(p,f)+KI)}\cdot \min\{p(m\vert x),p(m\vert y)\}\}.
    \end{align}
    Then, for all $K\geq 3$, $q(T_K^c\vert S_K) \leq 6\cdot 2^{-(\marg_I(p,f)+KI)/I}$.
\end{claim}
\begin{proof}
    Define the sets
    \begin{align*}
        G_1 &= \{xym: q(m\vert xy) < 2^{-(\marg_I(p,f)+KI)/I}p(m\vert xy)\} \\
        G_2 &= \{xym: q(x\vert ym) < 2^{-(\marg_I(p,f)+KI)/I}p(x\vert y)\} \\
        G_3 &= \{xym: q(y\vert xm) < 2^{-(\marg_I(p,f)+KI)/I}p(y\vert x)\}. 
    \end{align*}
    We claim that $q(G_1^c),q(G_2^c)$ and $q(G_3^c)$ are all smaller than $2^{-(\marg_I(p,f)+KI)/I}$. Indeed, to bound $q(G_1^c)$, we see that
    \begin{align*}
        q(G_1^c) = \sum_{xym\in G_1}q(xym) <  2^{-(\marg_I(p,f)+KI)/I}\cdot \sum_{xym\in G_1}q(xy)\cdot p(m\vert xy) \leq 2^{-(\marg_I(p,f)+KI)/I}.
    \end{align*}
    The proof for the bounds on $q(G_2^c)$ and $q(G_3^c)$ are similar. For any $xym\in G_1\cap G_2\cap S_K$, we have
    \begin{align*}
        \marg_I(p,f) &\geq \log \Bigg(\frac{q(x\vert ym)}{p(x\vert y)}\cdot\frac{q(y\vert xm)}{p(y\vert x)}\cdot  \left(\frac{q(xym)}{p(xym)}\right)^I  \cdot     \bigg|\E_{q(xy\vert m)}\left[(-1)^{f}\right]\bigg|^{-12I/\delta} \Bigg) \\
        &\geq  \log \frac{q(x\vert ym)}{p(x\vert y)} - \frac{(\marg_I(p,f) + KI)}{I} - 3(\marg_I(p,f) + KI) - I \tag{by \Cref{eq:S-ub-lb} and definition of $G_2$}\\
        &\geq \log \frac{q(x\vert ym)}{p(x\vert ym)} + \log \frac{p(x\vert ym)}{p(x\vert y)}  - 4(\marg_I(p,f) + KI) - I \tag{since $I\geq 1$}\\
        &\geq -\frac{(\marg_I(p,f) + KI)}{I} + \log \frac{p(m\vert xy)}{p(m\vert y)}  - 4(\marg_I(p,f) + KI) - I,
    \end{align*} 
    where in the last step we used the fact $p(m\vert xy)/p(m\vert y) = p(x\vert ym)/p(x\vert y)$. Rearranging, we get that for every $xym \in G_1\cap G_2\cap S_K$
    \begin{align*}
        \log \frac{p(m\vert xy)}{p(m\vert y)} \leq 6(\marg_I(p,f) + KI).
    \end{align*}
    A similar calculation shows that for every $xym\in G_1\cap G_3\cap S_K$ it holds that 
    \begin{align*}
        \log \frac{p(m\vert xy)}{p(m\vert x)} \leq 6(\marg_I(p,f) + KI).
    \end{align*}
    Therefore,
    \begin{align*}
        q(T_K^c\vert S_K) \leq \frac{q(G_1^c) + q(G_2^c) + q(G_3^c)}{q(S_K)} \leq 6\cdot 2^{-(\marg_I(p,f) + KI)/I}.
    \end{align*}
\end{proof}

Additionally, the bound on the marginal information cost implies the following lemma which will be useful in our simulation.
\begin{lemma}
    \begin{align}
        \E_{q(xym)}\bigg[\sum_{i \geq 2}^C\|p(m_i\vert xm_{< i}) - p(m_i\vert ym_{< i})\|_{1}\bigg] &\leq 8\sqrt{C\cdot \marg_I(p,f)} \label{eqn:pinskerbound-1}\\
        \E_{q(xym)}  \bigg[\bigg|\E_{q(xy\vert m)}\left[(-1)^{f}\right]\bigg|\bigg] &\geq 2^{-\delta\marg_I(p,f)/12I}\label{eqn:expectadbound}.
    \end{align}
\end{lemma}
\begin{proof}
    By our bound on the marginal information cost, we get 
\begin{align}
    \marg_I(p,f) &= \max_{xym \in \mathsf{supp}(q)} \log \Bigg(\frac{q(x\vert ym)}{p(x\vert y)}\cdot\frac{q(y\vert xm)}{p(y\vert x)}\cdot  \left(\frac{q(xym)}{p(xym)}\right)^I  \cdot     \bigg|\E_{q(xy\vert m)}\left[(-1)^{f}\right]\bigg|^{-12I/\delta} \Bigg) \notag \\
    &\geq   \E_{q(xym)}\bigg[\log \frac{q(x\vert ym)}{p(x\vert y)}\bigg]+\E_{q(xym)}\bigg[\log \frac{q(y\vert xm)}{p(y\vert x)}\bigg] \notag \\ & \qquad + I \cdot  \E_{q(xym)}\bigg[\log \frac{q(xym)}{p(xym)}\bigg] +  \E_{q(xym)}  \bigg[\log \bigg|\E_{q(xy\vert m)}\left[(-1)^{f}\right]\bigg|^{-12I/\delta}\bigg]\label{eqn:expect}
\end{align}
By \Cref{eq:div-non-neg} and the fact that the advantage is always at most $1$, each of the expectations appearing above is non-negative, and so each term is bounded by $\marg_I(p,f)$. This implies
\[\log \E_{q(xym)}  \bigg[\bigg|\E_{q(xy\vert m)}\left[(-1)^{f}\right]\bigg|\bigg] \geq \E_{q(xym)}  \bigg[\log \bigg|\E_{q(xy\vert m)}\left[(-1)^{f}\right]\bigg|\bigg] \geq -\frac{\delta\marg_I(p,f)}{12I},\]
thus giving \Cref{eqn:expectadbound}. For \Cref{eqn:pinskerbound-1}, we have
\begin{align}
    &\E_{q(xym)}\bigg[\sum_{i \geq 2}^C\|p(m_i\vert xm_{< i}) - p(m_i\vert ym_{< i})\|_{1}\bigg]\notag \\
    &\leq \E_{q(xym)}\bigg[\sum_{i}\|p(m_i\vert xm_{< i}) - q(m_i\vert xym_{< i})\|_{1} + \|q(m_i\vert xym_{< i}) - p(m_i\vert ym_{< i})\|_{1}\bigg]\notag \\
    &\leq 2\E_{q(xym)}\bigg[\sum_{i}\sqrt{\E_{q(m_i\vert xym_{< i})}\bigg[\log \frac{q(m_i\vert xym_{< i})}{p(m_i\vert xm_{< i})}\bigg]} + \sqrt{\E_{q(m_i\vert xym_{< i})}\bigg[\log \frac{q(m_i\vert xym_{< i})}{p(m_i\vert ym_{< i})}\bigg]}\bigg] \tag{by \Cref{eq:pinsker-inequality}}\\
    &\leq 2\sqrt{C\cdot \E_{q(xym)}\bigg[\sum_{i}\log \frac{q(m_i\vert xym_{< i})}{p(m_i\vert xm_{< i})}\bigg]} + 2\sqrt{C\cdot \E_{q(xym)}\bigg[\sum_{i}\log \frac{q(m_i\vert xym_{< i})}{p(m_i\vert ym_{< i})}\bigg]} \tag{by concavity of $\sqrt{\cdot }$}\notag \\
    &= 2\sqrt{C\cdot\E_{q(xym)}\bigg[\log \frac{q(m\vert xy)}{p(m\vert x)}\bigg]} + 2\sqrt{C\cdot \E_{q(xym)}\bigg[\log \frac{q(m\vert xy)}{p(m\vert y)}\bigg]}\notag.
\end{align}
To complete the proof, we claim that 
\[\E_{q(xym)}\bigg[\log \frac{q(m\vert xy)}{p(m\vert x)}\bigg] , \E_{q(xym)}\bigg[\log \frac{q(m\vert xy)}{p(m\vert y)}\bigg] \leq \marg_I(p,f) \cdot (1 + 1/I).\]
We show this for the first term; the proof for the second term is identical. First, we have
\begin{align*}
   \frac{q(m\vert xy)}{p(m\vert x)} &= \frac{q(m\vert xy)}{p(m\vert xy)}\cdot \frac{p(m\vert xy)}{p(m\vert x)} \\
   &= \frac{q(m\vert xy)}{p(m\vert xy)}\cdot  \frac{p(x\vert ym)}{p(x\vert y)} \\
   &= \frac{q(m\vert xy)}{p(m\vert xy)}\cdot  \frac{p(x\vert ym)}{q(x\vert ym)}\cdot \frac{q(x\vert ym)}{p(x\vert y)} \\
   &= \frac{q(xym)}{p(xym)}\cdot \frac{p(xy)}{q(xy)} \cdot \frac{p(x\vert ym)}{q(x\vert ym)}\cdot \frac{q(x\vert ym)}{p(x\vert y)}.
\end{align*}
Therefore,
\begin{align*}
    &\E_{q(xym)}\bigg[\log \frac{q(m\vert xy)}{p(m\vert x)}\bigg] \\
    &= \E_{q(xym)}\bigg[\log \frac{q(xym)}{p(xym)}\bigg] + \E_{q(xym)}\bigg[\log \frac{p(xy)}{q(xy)}\bigg] +\E_{q(xym)}\bigg[\log \frac{p(x\vert ym)}{q(x\vert ym)}\bigg] + \E_{q(xym)}\bigg[\log \frac{q(x\vert ym)}{p(x\vert y)}\bigg]\\
    &\leq \E_{q(xym)}\bigg[\log \frac{q(xym)}{p(xym)}\bigg] + \log\E_{q(xym)}\bigg[ \frac{p(xy)}{q(xy)}\bigg] + \log \E_{q(xym)}\bigg[ \frac{p(x\vert ym)}{q(x\vert ym)}\bigg] + \E_{q(xym)}\bigg[\log \frac{q(x\vert ym)}{p(x\vert y)}\bigg] \\
    &\leq \E_{q(xym)}\bigg[\log \frac{q(xym)}{p(xym)}\bigg] + \E_{q(xym)}\bigg[\log \frac{q(x\vert ym)}{p(x\vert y)}\bigg] \leq \frac{\marg_I(p,f)}{I} + \marg_I(p,f), 
\end{align*}
where in the first inequality, we used the concavity of $\log(\cdot)$ and in the last one, we used \Cref{eqn:expect}.
\end{proof}

For the simulation of external marginal information, we need a claim analogous to the previous one. 
\begin{lemma}
Let $q$ be a distribution achieving $\mext_I(p,f)$. Then,
    \begin{align}
        \E_{q(xym)}\bigg[\log \frac{p(m\vert xy)}{p(m)}\bigg] &\leq \mext_I(p,f), \label{eq:mext-log-ratio-bound}\\
        \E_{q(xym)}\bigg[\bigg|\E_{q(xy\vert m)}\left[(-1)^{f}\right]\bigg|\bigg] &\geq 2^{-\delta\mext_I(p,f)/(12I)}\label{eq:mext-adbound}.
    \end{align}
\end{lemma}
\begin{proof}
   By our bound on the marginal information cost, we get 
\begin{align}
    \mext_I(p,f) &= \max_{xym \in \mathsf{supp}(q)} \log \Bigg(\frac{q(xy\vert m)}{p(xy)}\cdot \left(\frac{q(xym)}{p(xym)}\right)^I  \cdot  \bigg|\E_{q(xy\vert m)}\left[(-1)^{f}\right]\bigg|^{-12I/\delta} \Bigg) \notag \\
    &\geq   \E_{q(xym)}\bigg[\log \frac{q(xy\vert m)}{p(xy)}\bigg]+ I \cdot  \E_{q(xym)}\bigg[\log \frac{q(xym)}{p(xym)}\bigg] 
    - \frac{12I}{\delta}\cdot\E_{q(xym)}  \bigg[\log \bigg|\E_{q(xy\vert m)}\left[(-1)^{f}\right]\bigg|\bigg]. \notag
\end{align} 
By \Cref{eq:div-non-neg} and the fact that the advantage is always at most $1$, each of the expectations appearing above is non-negative, and so each term is bounded by $\marg_I(p,f)$. This implies
\[\log \E_{q(xym)}  \bigg[\bigg|\E_{q(xy\vert m)}\left[(-1)^{f}\right]\bigg|\bigg] \geq \E_{q(xym)}  \bigg[\log \bigg|\E_{q(xy\vert m)}\left[(-1)^{f}\right]\bigg|\bigg] \geq -\frac{\delta\marg_I(p,f)}{12I},\]
thus giving \Cref{eq:mext-adbound}. Moreover, 
\[ \mext_I(p,f) \geq \E_{q(xym)}\bigg[\log \frac{q(xy\vert m)}{p(xy)}\bigg] = \E_{q(xym)}\bigg[\log \frac{q(xy\vert m)}{p(xy\vert m)}\bigg] + \E_{q(xym)}\bigg[\log \frac{p(xy\vert m)}{p(xy\vert m)}\bigg],\]
and this implies \Cref{eq:mext-log-ratio-bound} since the first term in the sum is non-negative. 
\end{proof}

%% file: sim-2.tex
\section{Compressing marginal information} \label{step3a}

Here we prove \Cref{simulationthoerem}. 
Let $p(xym)$ be a protocol distribution such that $p(xy) = \mu(xy)$, and $\marg_I(p,f) = \alpha\cdot I$. Let $q(xym)$ be a rectangular distribution that realizes $\marg_I(p,f)$. For a large constant $K$, let $\marg = \marg_I(p,f)+ KI$. Since $\marg(p,f) \geq 0$, we have $\marg \geq KI$. Let $g_1,g_2$ be as in \Cref{eq:g1g2}.

Let $\eps$ be a parameter such that $\eps \gg (2^{11\marg/I}\sqrt{C\cdot\marg})^{-1}$. We define a protocol $\Gamma$ whose communication complexity is bounded by 
\begin{align*}
O(\marg + \log 1/\eps + 2^{7\marg/I}\cdot \sqrt{C \marg} \cdot  \log (C/\epsilon)).
\end{align*}
Using the assumption that $\marg_I(p,f) \leq \alpha I$ we see that $\marg \leq \Delta_1\cdot I$ and $\log 1/\eps \leq \Delta_2\cdot \log (CI)$, where $\Delta_1, \Delta_2$ only depend on $\alpha$. This implies the bound on the communication in the theorem.

Here is a description of $\Gamma$:
\begin{enumerate}
    \item Jointly sample $p(m_0)$. Alice sets $m^A_0 = m_0$ and Bob sets $m^B_0 = m_0$. Jointly sample $\eta^A,\eta^B \in [0,1]$ uniformly. Jointly sample uniformly random $\rho  \in [0,1]^C$. Jointly sample a uniformly random function $h : \mathbb{Z} \to \{1,\ldots,\lceil 1/\epsilon\rceil \}$.
    \item Run the protocol $\psi$ from \Cref{lem:round-sim} with $u = p(m_1\vert m^A_0x), v = p(m_1\vert m^B_0y), L = 6\marg$, error parameter $\epsilon$, to obtain functions $a,b$ and transcript $s$. Alice sets $m^A_1 = a(\psi(us))$, Bob sets $m^B_1 = b(\psi(vs))$. If $m^B_1 = \bot$, the protocol terminates. Bob sends a bit to Alice to indicate whether or not this occurs. The communication complexity of this step is $L + O(\log (1/\epsilon))$.
    \item Alice and Bob compute $m^B$, $m^A$ by setting  
    \begin{align}
    m^A_{i} &= 
    \begin{cases}
        1 & \text{if $\rho_i \leq p(m_i=1\vert  x m^{A}_{< i})$,}\\
        0 & \text{otherwise.}
    \end{cases}\label{eq:alicerule}\\
    m^B_{i} &= 
    \begin{cases}
        1 & \text{if $\rho_i \leq p(m_i=1\vert ym^{B}_{< i})$,}\\
        0 & \text{otherwise.}
    \end{cases} \label{eq:bobrule},
    \end{align}
    for $i = 2,\dotsc,C$.
   
    \item Run $\tau$ from \Cref{lem:first-diff} to find the smallest  $j$ with $m^A_j \neq m^B_j$. If $j$ is even, Alice flips the value of $m^A_j$ to $1-m^A_j$ and recomputes $m^A_i$ for $i=j+1,\dotsc, C$ using \Cref{eq:alicerule}.
    If $j$ is odd, Bob flips the value of $m^B_j$ to $1-m^B_j$ and recomputes $m^B_i$ for $i=j+1,\dotsc, C$ using \Cref{eq:bobrule}. The players repeat this process at most $2^{7\marg/I}\sqrt{C \marg}$ times. If by this point $\tau$ reports that $m^A \neq  m^B$, the players abort. Otherwise, they continue. Let $\langle m^A\rangle$, $\langle m^B\rangle$ denote the final values of $m^A, m^B$ after this step. The communication complexity of this step is at most 
    $O(2^{7\marg/I} \cdot \sqrt{C\marg}\cdot \log(C/\epsilon))$.
    
    \item If $\eta^A \leq g_1(x\langle m^A\rangle )\cdot2^{-\lceil \log g_1(x\langle m^A\rangle )\rceil}$, Alice sends $h(\lceil \log g_1(x\langle m^A\rangle )\rceil)$ to Bob, and otherwise she sends $\bot$ to indicate that the protocol should be aborted. 
    \item If there is a unique  integer $z$ 
 such that
 \begin{align*}
    |z + \log g_2(y\langle m^B\rangle )| &\leq 3\marg/I ,\\
  h(z)&=h(\lceil \log g_1(x\langle m^A\rangle )\rceil),\\ \eta^B &\leq g_2(y\langle m^B\rangle)\cdot2^{z - 3\marg/I}, 
  \end{align*}
 Bob sends  $\mathsf{sign} \Big( \E_{q(xy\vert \langle m^B\rangle )}[(-1)^f]\Big) \in \{\pm 1\}$ to Alice. Otherwise, he sends $\bot$ to abort the protocol.
\end{enumerate}

Let $\Gamma$ denote the joint distribution of the inputs and transcript of the above protocol. In order to analyze the protocol, define $m$ by setting $m_0 = m_0^A=m_0^B$, $m_1 = m^A_1$, and setting 
 \begin{align*}
    m_{i} &= 
    \begin{cases}
        1 & \text{if $\rho_i \leq p(m_i=1\vert  x m_{< i})$,}\\
        0 & \text{otherwise}
    \end{cases},
\end{align*}
when $i>1$ is even, and setting 
\begin{align*}
    m_{i} &= 
    \begin{cases}
        1 & \text{if $\rho_i \leq p(m_i=1\vert ym_{< i})$,}\\
        0 & \text{otherwise}
    \end{cases},
\end{align*}
when $i>1$ is odd. This definition ensures that 
\begin{align*}
     \Gamma(xym) = p(xym).   
\end{align*}

For $i=2,3,\dotsc,C$ define
\begin{align*}
    E_i &= \begin{cases}
         1 & \text{if $\rho_i$ is in between the numbers $p(m_i=1\vert xm_{<i})$ and $p(m_i=1\vert ym_{<i})$,}\\
         0 & \text{otherwise.}
     \end{cases}
\end{align*}

Let $S$ and $R$ be the sets defined in \Cref{eq:def-S,eq:R} for our choice of $K$. In addition to $S$ and $R$, we need the following sets to analyze the simulating protocol: 
\begin{align*}
    Q &= \Big \{xym \eta^A \eta^B: \eta^A \leq g_1(xm) \cdot 2^{-\lceil \log g_1(xm)\rceil}, \eta^B \leq g_2(ym)\cdot 2^{\lceil \log g_1(xm)\rceil - 3\marg/I}  \Big\},\\
    \mathcal{E} &= \{ \langle  m^A \rangle\langle m^B\rangle m : \langle m^A\rangle= \langle m^B\rangle =m\},\\
    \mathcal{Z} &= \{xymh:\exists \text{ a unique integer $z$ with }|z + \log g_2(ym )|\leq 3\marg/I \text{ and }  h(z)= h(\lceil \log g_1(xm) \rceil)\}.
\end{align*}
Let $\mathcal{G}$ denote the event that the protocol reaches the final step without aborting, and define $\mathcal{A}(xym) \in \{\pm 1\}$ by 
$$ \mathcal{A}(xym) = \mathsf{sign}\Big(\E_{q(xy \vert  m)}[(-1)^{f(xy)}] \Big) \cdot (-1)^{f(xy)}.$$
Our protocol computes $f(xy)$ correctly when: $\mathcal{G}$ happens, $\mathcal{A}(xym)=1$ and $m = m^B$. Since $\mathcal{EZ}SQ \subseteq \mathcal{G}$, and $\mathcal{E}$ implies $m=m^B$, the advantage of our protocol is at least:
\begin{align} 
&\Gamma(\mathcal{EZ}SQ) \cdot  \E_{\Gamma(xym|\mathcal{EZ}SQ)}[\mathcal{A}(xym)] - \Gamma(\mathcal{G} (\mathcal{EZ}SQ)^c). \label{finaladvantage-1}
\end{align}

We shall prove:
\begin{align}
    \E_{\Gamma(xym\vert \mathcal{EZ} QS)}[\mathcal{A}(xym)]&\geq \Omega(2^{-\delta \marg /(12I)}) ,\label{eqn:expectadbound-1}\\
    \Gamma(\mathcal{EZ}SQ) & \geq \Omega(2^{-3 \marg/I}), \label{eqn:qszebound-1}\\
    \Gamma(\mathcal{G}(\mathcal{EZ}SQ)^c) & \leq O(2^{-4\marg /I}). \label{eqn:negative-1}
\end{align}
By \Cref{finaladvantage-1}, since $\delta \leq 1$, we can choose $K$ to be large enough to prove the theorem, since $(\alpha+K) \geq \marg/I \geq K$. 

We first upper bound $\Gamma(\mathcal{G} (\mathcal{EZ}SQ)^c)$.  By the union bound, we have:
\begin{align*}
 \Gamma(\mathcal{G} (\mathcal{EZ}SQ)^c)
  & \leq  
   \Gamma(\mathcal{G}\mathcal{E}^c) +  \Gamma( \mathcal{Z}^c\vert \mathcal{GE})+\Gamma(S^c  \mathcal{GEZ})+  \Gamma(Q^c \vert \mathcal{G EZ}S).
\end{align*}

The definition of the protocol ensures that $\Gamma(\mathcal{Z}^c\vert \mathcal{GE})=0$. Moreover, we claim that $\Gamma(Q^c \vert \mathcal{G EZ}S) = 0$, because if the event $\mathcal{EZ}S$ happens and the parties do not abort, then:
\begin{align*}
\eta^A &\leq g_1(x\langle m^A\rangle)\cdot 2^{\lceil \log g_1(x\langle m^A\rangle) \rceil} = g_1(xm)\cdot 2^{\lceil \log g_1(xm) \rceil},\\ 
\eta^B &\leq g_2(y\langle m^B\rangle)\cdot 2^{z - 3\marg/I } = g_2(ym)\cdot 2^{\lceil \log g_1(xm) \rceil  - 3\marg/I }.
\end{align*} 
The event $\mathcal{G}\mathcal{E}^c$ implies that $\psi$ or $\tau$ made an error, leaving Alice and Bob with strings that were not equal in some step. The probability that this happens is at most $$O(\epsilon \cdot (1+2^{7\marg/I} \sqrt{C\cdot  \marg})) \leq 2^{-4\marg/I},$$ by our choice of $\epsilon$. Finally, $\Gamma(S^c\mathcal{GEZ}) \leq \Gamma(S^c\mathcal{EZ}) \leq O(\eps\marg/I)$, since if $S^c\mathcal{GEZ}$ happens then there must have been a hash collision, which happens with probability at most $O(\eps\marg/I)$. This implies \Cref{eqn:negative-1}.

Now, we turn to proving \Cref{eqn:qszebound-1}. Let us first estimate $\Gamma(QS)$. We have,
\begin{align*}
    \Gamma(QS) &=\sum_{xym\in S}\Gamma(xym)\cdot \Gamma(Q\vert xym) =\sum_{xym\in S}p(xym)\cdot \Gamma(Q\vert xym). 
\end{align*}
For $xym\in S$,
\begin{align}\label{eq:Q|xym}
    \Gamma(Q\vert xym) = \frac{g_1(xm) \cdot 2^{-\lceil \log g_1(xm)\rceil}\cdot g_2(ym)\cdot 2^{\lceil \log g_1(xm)\rceil}}{2^{3\marg/I}} 
    &= \frac{q(xym)}{p(xym)} \cdot \frac{1}{2^{3\marg/I}},
\end{align}
where the first equality follows from the fact that $$g_2(ym) \cdot 2^{\lceil \log g_1(xm)\rceil} = 2^{\lceil \log g_1(xm)\rceil + \log g_2(ym)} \leq 2^{3\marg/I},$$ by the definition of $S$. Therefore,
\begin{align}\label{eq:qs}
    \Gamma(QS) =\sum_{xym\in S}p(xym)\cdot \frac{q(xym)}{p(xym)} \cdot \frac{1}{2^{3\marg/I}} &= \frac{q(S)}{2^{3\marg/I}} \notag\\
    &\geq \frac{(1-5\cdot 2^{-\marg/I})}{2^{3\marg/I}}= \Omega(2^{-3\marg/I}),
\end{align}
where in the last line, we used \Cref{claim:q(S)-lb}.

We claim that for all $xym\in S$,
\begin{align}\label{eq:Z|xymQS-1}
    \Gamma(\calZ\vert xymQS) = \Gamma(\mathcal{Z}\vert xym) \geq 1 - O(\eps\marg/I).
\end{align}
The equality follows by observing that $xym$ determine $S$ and given $xym$, $\calZ$ just depends on the choice of $h$, which is independent of $Q$. The event $\mathcal{Z}^c$ can happen only if there exists an integer $z$ distinct from $\lceil \log g_1(xm) \rceil$ such that $h(\lceil g_1(xm) \rceil)=h(z)$ and $\abs{z + \log g_2(ym)} \leq 3\marg/I$. The probability that this happens is at most $O(\epsilon\cdot \marg /I )$. Therefore, $\Gamma(\mathcal{Z}\vert xym) \geq 1- O(\epsilon \marg / I)\geq 1/2$, by our choice of $\epsilon$.
We conclude that
\begin{align}
 \Gamma(QS\mathcal{Z}) &= \Gamma(QS) \cdot \Gamma (\mathcal{Z}\vert QS) \geq \Omega(2^{-3 \marg/I}), \label{eqn:zqsbound-1}
\end{align} 

For all $xym \in S$,
\begin{align}
    \Gamma(xym \vert QS\mathcal{Z}) &= \frac{\Gamma(xym)\cdot \Gamma(QS \mathcal{Z}\vert xym )}{\Gamma(Q S \mathcal{Z})}\notag \\
    &= \frac{p(xym)}{\Gamma(Q S) }  \cdot \Gamma(Q\vert xym) \cdot \frac{\Gamma(\mathcal{Z}\vert xymQS)}{\Gamma(\mathcal{Z}|QS)} \notag \\
    &= \frac{p(xym)}{\Gamma(Q S) } \cdot \frac{q(xym)}{p(xym)\cdot 2^{3\marg/I}} \cdot \frac{\Gamma(\mathcal{Z}\vert xymQS)}{\Gamma(\mathcal{Z}|QS)} \tag{By \Cref{eq:Q|xym}}\notag \\
    &=  \frac{q(xym)}{q(S)} \cdot \frac{\Gamma(\mathcal{Z}\vert xymQS)}{\Gamma(\mathcal{Z}|QS)} \tag{By \Cref{eq:qs}} \\
    &=  q(xym\vert S) \cdot (1\pm O(\epsilon \marg /I)), \label{eq:Gamma|QSZ=q|S}
\end{align}
where the last line follows by \Cref{eq:Z|xymQS-1}.

Given \Cref{eqn:zqsbound-1}, to complete the proof of \Cref{eqn:qszebound-1}, it will be enough to prove that $\Gamma(\calE \vert QS\mathcal{Z}) \geq 1/2$. We shall prove that 
\begin{align}
    &\Gamma(\mathcal{E}^c\vert QS\mathcal{Z})\notag \\
    &\leq \Gamma(R^c\vert QSZ) + \Gamma\bigg(\mathcal{E}^c\bigg\vert QS\mathcal{Z} R, \sum_{i=2}^CE_i \leq 2^{7\marg/I}\cdot\sqrt{C\marg}\bigg) + \Gamma\bigg(\sum_{i=2}^CE_i >2^{7\marg/I}\cdot\sqrt{C\marg} \Big\vert QS\calZ\bigg)\notag \\
    & \leq O(2^{-\marg/I}).\label{eqn:ecomplementbound} 
\end{align}

By \Cref{eq:Gamma|QSZ=q|S} and \Cref{claim:q(S)-lb},
\begin{align}
    \Gamma(R^c\vert QS\calZ) \leq q(R^c\vert S)(1+ O(\eps \marg/I)) \leq 2^{-\marg/I + 3}. \label{eqn:rcomplementbound}
\end{align}

Given $QS\mathcal{Z}R$ and the event $\sum_{i=2}^CE_i \leq  2^{7\marg/I}\cdot\sqrt{C\marg}$, the event $\calE^c$ can happen only if $\tau$ or $\psi$ make an error that leaves Alice and Bob with inconsistent messages, or if $\psi$ aborts. We claim that the probability that $\psi$ makes an error or aborts is at most $2\eps$. This is because every $xym\in R$ satisfies $p(m_1\vert xm_0) \leq 2^{6\marg}\cdot p(m_1\vert m_0y)$, 
so we can apply \Cref{lem:round-sim}. Moreover, the probability that $\tau$ ever makes an error is at most $O(\eps 2^{7\marg/I}\sqrt{C\marg})$ by a union bound. So, we conclude that
\begin{align}
\Gamma\bigg(\mathcal{E}^c\bigg\vert QS\mathcal{Z} R, \sum_{i=2}^CE_i \leq 2^{7\marg/I}\cdot\sqrt{C\marg}\bigg) & \leq O(\eps 2^{7\marg/I}\sqrt{C\marg}). \label{ecomplementgivenbound}
\end{align}

We shall prove at the end of this section that
\begin{align}\label{eq:mistakes-bound}
    \Gamma\bigg(\sum_{i=2}^CE_i > 2^{7\marg/I}\cdot\sqrt{C\marg} \Big\vert QS\calZ\bigg) &< O(2^{-\marg/I}).
\end{align}

\Cref{eqn:rcomplementbound,ecomplementgivenbound,eq:mistakes-bound} together prove \Cref{eqn:ecomplementbound}, and so conclude the proof of \Cref{eqn:qszebound-1}. Next, we prove \Cref{eqn:expectadbound-1}. Since $|\mathcal{A}(xym)| \leq 1$, we have
\begin{align*}
     \E_{\Gamma(xym\vert QS\mathcal{Z})}[\mathcal{A}(xym)] \leq  \Gamma(\mathcal{E}\vert QS\mathcal{Z}) \cdot  \E_{\Gamma(xym\vert QS\mathcal{ZE})}[\mathcal{A}(xym)] + \Gamma(\mathcal{E}^c\vert QS\mathcal{Z}),
\end{align*}
and since $\Gamma(\mathcal{E}\vert QS\mathcal{Z}) \leq 1$, this gives
\begin{align*}
    \E_{\Gamma(xym\vert QS\mathcal{ZE})}[\mathcal{A}(xym)] & \geq  \E_{\Gamma(xym\vert QS\mathcal{Z})}[\mathcal{A}(xym)]- \Gamma(\mathcal{E}^c\vert QS\mathcal{Z})\\
    & \geq  \E_{q(xym\vert S)}[\mathcal{A}(xym)]- O(\epsilon \marg/I) - O(2^{-\marg/I}) \tag{using \Cref{eqn:ecomplementbound,eq:Gamma|QSZ=q|S}}\\
    & \geq  \E_{q(xym)}[\mathcal{A}(xym)]- O(2^{-\marg/I}) \tag{by \Cref{claim:q(S)-lb}}\\
    & = \E_{q(xym)}\bigg[\mathsf{sign}\bigg(\E_{q(x'y'\vert m)}\Big[(-1)^{f}\Big]\bigg) \cdot (-1)^{f(xy)}\bigg] - O(2^{-\marg/I})\\
    &= \E_{q(m)}\bigg[\bigg|\E_{q(xy\vert m)}\Big[(-1)^{f(xy)}\Big]\bigg|\bigg] - O(2^{-\marg/I}) \geq \Omega(2^{-\delta \marg/(12I)}),
\end{align*}
by \Cref{eqn:expectadbound}. This completes the proof of \Cref{eqn:expectadbound-1}.

It only remains to prove \Cref{eq:mistakes-bound}. Define the function
\begin{align*}
t(xym) = \begin{cases}\min\{j : q(xym_{\leq j}) < 2^{-3 \marg/I} \cdot p(xym_{\leq j})\} & \text{if such $j$ exists,}\\
\bot & \text{otherwise.}
\end{cases}
\end{align*}
Note that the function $t(xym)$ is determined by $xym_{\leq t(xym)}$.

We have
    \begin{align*}
    \Gamma\bigg(\sum_{i=2}^CE_i> 2^{7\marg/I}\cdot \sqrt{C\marg}\Big\vert QS\calZ \bigg) 
        &\leq \Gamma(t\neq \bot\vert QSZ)+ \Gamma\bigg(\sum_{i=2}^CE_i > 2^{7\marg/I}\cdot \sqrt{C\marg}\Big\vert QS\calZ, t=\bot\bigg) , 
    \end{align*}
so let us bound each of these terms. 

\begin{align*}
        \Gamma(t \neq \bot \vert QSZ) &\leq q(t \neq \bot \vert S)\cdot (1+ O(\eps\marg/I)) \tag{by \Cref{eq:Gamma|QSZ=q|S}}\\
        &\leq q(t \neq \bot )\cdot (1 + O(2^{-3\marg/I}))\cdot (1+ O(\eps\marg/I)), \tag{by \Cref{claim:q(S)-lb}}
\end{align*}
and
\begin{align*}
    q(t\neq \bot) = \sum_{j=0}^C q(t=j)
     &= \sum_{j=0}^C \sum_{\substack{xym_{\leq j}\\ t(xym)=j}} q(xym_{\leq j})\\
    &< 2^{-3 \marg/I} \cdot \sum_{j=0}^C \sum_{\substack{xym_{\leq j}\\ t(xym)=j}} p(xym_{\leq j})
    = 2^{-3 \marg/I} \cdot p(t\neq \bot) \leq 2^{-3 \marg/I},
\end{align*}
so we conclude that
\begin{align}
    \Gamma(t \neq \bot \vert QS\calZ) &\leq O(2^{-3 \marg/I}). \label{eqn:tbotbound}
\end{align}

    Next, we show that 
    \begin{align*}
        \Gamma\bigg(\sum_{i=2}^CE_i >  2^{7\marg/I}\cdot \sqrt{C\marg} \Big\vert QS\calZ, t=\bot\bigg) < O(2^{-\marg/I}),
    \end{align*}
    which would complete the proof of \Cref{eq:mistakes-bound}. This follows from Markov's inequality and the bound
    \begin{align}
\E_{\Gamma}\bigg[\sum_{i=2}^CE_i\Big\vert QS\calZ,t =\bot \bigg] \leq O(2^{6 \marg/I} \sqrt{C \marg}),\label{eqn:econdbound} \end{align}which we prove next.
    We have:
    \begin{align}\label{eq:mistake-bound|QSZR}
        \E_{\Gamma}\bigg[\sum_{i=2}^CE_i\Big\vert QS\calZ,t =\bot \bigg]  &= \sum_{i=2}^C\frac{\Gamma(E_i=1,QS\calZ,t=\bot)}{\Gamma(QS\calZ,t=\bot)} \\
    &\leq O(2^{3 \marg/I}) \cdot \sum_{i=2}^C\Gamma(E_i=1,t=\bot). \tag{by \Cref{eqn:zqsbound-1,eqn:tbotbound}}
    \end{align}
    Moreover,
    \begin{align}
    \Gamma(E_i=1,t=\bot) &= \sum_{xym_{< i}}\Gamma(xym_{< i})\cdot \Gamma(E_i=1\vert xym_{< i})\cdot \Gamma(t=\bot\vert E_i=1,xym_{< i}) \notag \\
    &= \sum_{xym_{< i}}p(xym_{< i})\cdot \|p(m_i\vert xm_{< i}) - p(m_i\vert ym_{< i})\|_1\cdot \Gamma(t=\bot\vert E_i=1,xym_{< i}) \notag \\
    &\leq O(2^{3\marg/I}) \cdot\sum_{xym_{< i}}q(xym_{< i})\cdot \|p(m_i\vert xm_{< i}) - p(m_i\vert ym_{< i})\|_1 
    \end{align}
    Therefore, 
    \begin{align*}
        \E_{\Gamma}\bigg[\sum_{i=2}^CE_i\Big\vert QS\calZ,t =\bot \bigg]  &\leq O(2^{3\marg/I}) \cdot\E_{q(xym)}\bigg[\sum_{i=2}^C \|p(m_i\vert xm_{< i}) - p(m_i\vert ym_{< i})\|_1\bigg] \\ 
        &\leq O(2^{3\marg/I} \cdot \sqrt{C \marg}),
    \end{align*}
    by \Cref{eqn:pinskerbound-1}, which completes the proof of \Cref{eqn:econdbound}.

%% file: sim-analysis-2.tex
\section{Smoothing protocols}\label{smoothing}

A smooth protocol is a protocol where each message bit is close to being uniformly distributed:
\begin{definition}
Given a protocol distribution $p(xym)$ with $C$ messages satisfying $m_2,\dotsc,m_C \in \{0,1\}$, we say that the distribution is $\beta$-smooth if for all $i>1$, $|p(m_i\vert xym_{<i}) - 1/2| \leq \beta$.    
\end{definition}

Here we prove the following theorem:
\begin{theorem} \label{thm:smoothing}
    For every Boolean function $f$, every protocol distribution $p(xym)$ with $C$ messages satisfying $m_2,\dotsc,m_C \in \{0,1\}$, and every $\beta >0$, assuming that  $\marg^{\mathsf{ext}}_I(p,f)$ is finite, there is a $\beta$-smooth protocol $p'(xym')$ with $C'\leq O(C \cdot \log(IC)/\beta^2)$ messages such that $\marg^{\mathsf{ext}}_I(p',f) \leq \marg^{\mathsf{ext}}_I(p,f) + 1$, and $m'_2,\dotsc, m'_{C'} \in \{0,1\}$.
\end{theorem}
\begin{proof}
Let $q(xym)$ be a rectangular distribution realizing $\marg^{\mathsf{ext}}_I(p,f)$. Let $L>1$ be a large odd number to be determined. Define the pair of distributions $q'(xym'),p'(xym')$ as follows. Let $m'_0,m'_1$ have the same support as $m_0,m_1$, and let  $m'_2,\dotsc, m'_C \in \{0,1\}^{L}$. In $p',q'$, the $i$'th message will correspond to $m'_i$.

For $a \in \{0,1\}$, define the following distributions supported on $\{0,1\}^L$:
\begin{align*}
    s_a(r) & = \prod_{j=1}^L \frac{1}{2}+ (-1)^{a+r_j} \cdot \beta\\
    t_a(r) & = s_a\Big (r \big \vert (-1)^a \cdot \sum_{j=1}^L(-1)^{r_{j}}\geq 0 \Big)\\
    t'_a(r) & = s_a\Big (r \big \vert (-1)^a \cdot \sum_{j=1}^L(-1)^{r_{j}}< 0 \Big).
\end{align*}
In words, $s_a(r)$ is the distribution of $L$ independent bits that are biased towards being equal to $a$, $t_a(r)$ is this distribution conditioned on the event that the majority of the bits is equal to $a$ and $t'_a(r)$ is the distribution conditioned on the event that the majority is not $a$.

Now we define a protocol distribution $p'(xym')$ and a rectangular distribution $q'(xym')$. Given $m'$, let $D(m'_i)$ denote the unique string satisfying $D(m'_0,m'_1) = (m'_0,m'_1)$, and $$(-1)^{D(m'_i)} \cdot \sum_{j=1}^L (-1)^{m'_{i,j}}\geq 0.$$
In other words, $D$ decodes each block of $L$ bits by taking the majority. Below we abuse notation and write $D(m) = D(m_0),D(m_1),\dotsc,D(m_C)$.

Define
\begin{align*}
    q'(xym') &= q(xyD(m')) \cdot \prod_{i=2}^C t_{D(m')_i}(m'_{i}).
\end{align*}
The definition ensures that $q'(xym')$ is rectangular, and that conditioned on $D(m')$, $xy$ is independent of $m'$. Define the distribution $p'(xymm')$ as follows:
\begin{align*}
 p'(xym'_0m'_1) &= p(xym'_0m'_1),
 \end{align*}
 and for $i>1$,
 \begin{align*}
 p'(m_im'_i\vert xym_{<i}m'_{<i} ) & = p(m_i \vert xy,m_{<i}=D(m')_{<i})\cdot s_{m_i}(m'_i) 
 \end{align*}
 In words, in the protocol $p'(xym')$, the parties privately sample each message bit $m_i$ according to the protocol distribution $p$. However, instead of sending this sampled bit, they send $m'_i$ sampled according to $s_{m_i}(m'_i)$. After this transmission, they continue the protocol using $D(m')_{<i}$. Strictly speaking, in order to ensure that the new protocol is a protocol distribution, we require that all the odd bits are transmitted by Alice and all even bits are sent by Bob. This can be easily achieved by inserting random bits into the transcript, but we leave out the details here. 
 
 There is some small chance that for $i>1$, $D(m')_i \neq m_i$, but by the Chernoff bound,
 $$p'(D(m')_i \neq m_i)\leq \exp(-\Omega(\beta^2 L)).$$

We have that for all $xym'$ in the support of $q'$,
\begin{align*}
    \frac{q'(xym')}{p'(xym')} =  \frac{q(xyD(m'))}{p'(xyD(m'))} \cdot \prod_{i=2}^C \frac{ t_{D(m')_i}(m'_i)}{ p'(m'_i \vert xy D(m')_{\leq i})}.
    \end{align*}
For any $xyD(m')$ such that $q(xyD(m')) > 0$, we can bound
   \begin{align*}
       \frac{q(xyD(m'))}{p'(xyD(m'))} &= \frac{q(xyD(m'))}{p(xyD(m'))} \cdot \prod_{i=2}^C \frac{p(D(m')_i\vert xyD(m')_{<i} )}{p'(D(m')_i\vert xyD(m')_{<i} )}\\
       & \leq \frac{q(xyD(m'))}{p(xyD(m'))} \cdot \prod_{i=2}^C \frac{1}{1-\exp(-\Omega(\beta^2 L))},
   \end{align*}
   where we assumed that $p(xyD(m')) > 0$; if $p(xyD(m')) = 0$ then $q(xyD(m')) = 0$ for otherwise the marginal information cost would be unbounded. Next,
   \begin{align*}
    \prod_{i=2}^C \frac{ t_{D(m')_i}(m'_i)}{ p'(m'_i \vert xy D(m')_{\leq i})} &=   \prod_{i=2}^C \frac{ t_{D(m')_i}(m'_i)}{ \E_{p'(m_i \vert xy D(m')_{\leq  i})}[p'(m'_i \vert xy D(m')_{\leq  i}, m_i)] } \\
    & =  \prod_{i=2}^C \frac{ t_{D(m')_i}(m'_i)}{ \E_{p'(m_i \vert xy D(m')_{\leq  i})}[s_{m_i}(m'_i\vert D(m'_i))]} \\
    & =    \prod_{i=2}^C \frac{ t_{D(m')_i}(m'_i)}{ p'(m_i=D(m')_i) \cdot  t_{D(m')_i}(m'_i)+p'(m_i\neq D(m')_i)\cdot t'_{D(m')_i}(m'_i)} \\
    & =   \prod_{i=2}^C \frac{1}{ p'(m_i=D(m')_i) +p'(m_i\neq D(m')_i)\cdot \frac{t'_{D(m')_i}(m'_i)}{t_{D(m')_i}(m'_i)}} \\
    & \leq   \prod_{i=2}^C \frac{1}{ p'(m_i=D(m')_i) } \\
    & \leq    \prod_{i=2}^C \frac{1}{1- \exp(-\Omega(\beta^2 L)) }. 
\end{align*} 
So, we obtain the bound:
\begin{align*}
    \frac{q'(xym')}{p'(xym')} =  \frac{q(xyD(m'))}{p(xyD(m'))} \cdot (1+2C\exp(-\Omega(\beta^2 L)) ).
    \end{align*}

Moreover, for all $xym'$ in the support of $q'$, we have
\begin{align*}
q'(x y\vert m') & = \frac{q'(xym')}{q'(m')}\\
&= \frac{q(xyD(m')) \cdot \prod_{i=1}^C t_{D(m')_i}(m'_i) }{q(D(m')) \cdot \prod_{i=1}^C t_{D(m')_i}(m'_i) }\\
&=q(xy\vert  D(m')).
\end{align*}

Finally, since $q'(xy\vert m') = q(xy \vert D(m'))$, we have
\begin{align*}
    \Big |\E_{q'(xy\vert m')}[(-1)^f] \Big | &= |\E_{q'(xy\vert D(m'))}[(-1)^f] \Big |.
\end{align*}

Thus, we get that $$\marg^{\mathsf{ext}}_I(p',f) \leq \marg^{\mathsf{ext}}_I(p,f) + IC \cdot \exp(-\Omega(\beta^2 L)).$$
Setting $L = O(\log(IC)/\beta^2)$ proves the theorem.
\end{proof}

Smooth protocols have the feature that the log-ratios of the information terms are tightly concentrated. To explain this phenomenon, we need to introduce a few definitions. For every $xym$ in the support of $p$, and $j\geq 2$, define the $j$-th divergence costs:
\begin{align*}
d^A_j(xm) &= \sum_{\substack{2 \leq i \leq j\\ i \text{ odd}}} \E_{p(m_i\vert xm_{<i})}\Big [ \log \frac{p(m_i\vert xm_{<i})}{p(m_i\vert m_{<i})}\Big],\\
d^B_j(ym) &= \sum_{\substack{2 \leq i \leq j\\ i \text{ even}}} \E_{p(m_i\vert ym_{<i})}\Big [ \log \frac{p(m_i\vert ym_{<i})}{p(m_i\vert m_{<i})}\Big],\\
d_j(xym) &= d^A_j(xm) + d^B_j(ym).
\end{align*}

By the non-negativity of divergence, the divergence costs are monotone i.e. $d_j(xym) \leq d_{j+1}(xym)$. Since the protocol is $\beta$-smooth, we have
\begin{align}
d_{j+1}^A(xm)-d_{j}^A(xm) &\leq \log\frac{1/2+ \beta}{1/2 - \beta} \leq 5\beta,\notag\\
d_{j+1}^B(ym)-d_{j}^B(ym) &\leq \log\frac{1/2+ \beta}{1/2 - \beta} \leq 5\beta. \label{eqn:smoothness}
\end{align}

We say a function $r(xym)$ taking values in $\{1,\ldots,C\}$ is a \emph{frontier} if every $m$ contains exactly one prefix of the type $m'_{\leq r(xym')}$, and that is the prefix $m_{\leq r(xym)}$. Alternatively, for every $m,m'$ such that $r(xym) \neq r(xym')$, it holds that both $r(xym)$ and $r(xym')$ are larger than the length of the longest common prefix of $m$ and $m'$.
Given a frontier $r(xym)$, define
\begin{align}
    F_{r,\alpha} & = \Big\{xym: \Big |\sum_{i\geq 2}^{r(xym)}\log \frac{p(m_{i}\vert xym_{< i})}{p(m_{i}\vert m_{< i})} - d_{r(xym)}(xym)\Big | \geq  \alpha \Big\},\notag \\
    F^A_{r,\alpha} & = \Big\{xym: \Big |\sum_{i \geq 2 \text{ odd}}^{r(xym)}\log \frac{p(m_{i}\vert xym_{< i})}{p(m_{i}\vert m_{\leq i})} - d^A_{r(xym)}(xm)\Big | \geq  \alpha \Big\},\notag \\
    F^B_{r,\alpha} & = \Big\{xym: \Big |\sum_{i \geq 2\text{ even}}^{r(xym)}\log \frac{p(m_{i}\vert xym_{< i})}{p(m_{i}\vert m_{\leq i})} - d^B_{r(xym)}(ym)\Big | \geq  \alpha \Big\}.\label{badfrontiers}
\end{align}

\begin{lemma}\label{divconcentrate}
    Let $r(xym)$ be a frontier such that for every $xym$, it holds that $d_{r(xym)}(xym) \leq \tau$.
    Then   $p(F_{r,\alpha}),p(F^A_{r,\alpha})$ and $p(F^B_{r,\alpha})$ are all at most  $2\exp(-\Omega(\alpha^2/\tau))$.
\end{lemma}
\begin{proof}
We prove the inequality for $p(F_{r,\alpha})$; the proofs for the other two terms are similar.
Define the random variable $z_0,z_1\ldots$ where $z_0 = z_1 = 0$ and for every $i\geq 2$,
\begin{align*}
z_i = \begin{cases}\log \frac{p(m_i\vert xym_{<i})}{p(m_i \vert m_{<i})} & \text{if $i\leq r(xym)$}\\
0 & \text{otherwise.}
\end{cases}
\end{align*}
and let $t_i = z_i - \E_{p(m_i\vert xy m_{<i})}[z_i]$. Then by definition $\E[t_i\vert t_{<i}]=0$. Moreover, we have
\begin{align*}
    \sup(z_i\vert xy m_{<i}) &\leq \max_{m_i}\Big \{ \log \frac{p(m_i\vert xym_{<i})}{p(m_i\vert m_{<i})} \Big\}\\
    &\leq \log \frac{1/2-\beta + \sqrt{d_i(xym) - d_{i-1}(xym)}}{1/2-\beta}\\
    & \leq O(\sqrt{d_i(xym)-d_{i-1}(xym)}).
\end{align*}
Similarly, 
\begin{align*}
    \inf(z_i\vert xy m_{<i}) &\geq \log \frac{1/2-\beta }{1/2-\beta- \sqrt{d_i(xym) - d_{i-1}(xym)}}\\
    & \geq -O(\sqrt{d_i(xym)-d_{i-1}(xym)}).
\end{align*}

So, if we define $L$ as below, we have
\begin{align*}
    L & = \sup_{xym} \sum_{i=2}^C (\sup(t_i\vert xym_{<i}) -\inf(t_i \vert xym_{<i}))^2 \\
    & =\sup_{xym} \sum_{i=2}^C (\sup(z_i\vert xym_{<i}) -\inf(z_i \vert xym_{<i}))^2 \\
    & \leq O(\tau).
\end{align*}

It is well known that if $\E[t_i]=0$, then $\E[\exp(t_i)] \leq \exp((\sup(t_i)-\inf(t_i))^2/8)$ (see Lemma 2.6 in \cite{JerrumHMRR98}). We can use this inequality to bound:
\begin{align*}
    \E_{p(m\vert xy)}\Big[\exp(\frac{4\alpha}{L} \cdot \sum_{i=2}^C t_i)\Big] & \leq   \E_{p(m_{\leq 2}\vert xy) }\Big[\exp(\frac{4\alpha}{L} t_2)\cdot \E_{p(m\vert xym_{<3})}\Big [\exp(\frac{4\alpha}{L} \cdot \sum_{i=3}^C t_i)\Big] \Big]\\
    & \leq \dotsc\\
    &\leq \exp\Big(\frac{(4\alpha/L)^2 \sup_{xym} \sum_{i=2}^C (\sup(t_i\vert xym_{<i})-\inf(t_i\vert xym_{<i}))^2}{8}\Big)\\
    &\leq \exp\Big(2 \alpha^2/L \Big).
\end{align*}

So by Markov's inequality, we get:
\begin{align*}
p(\sum_{i=2}^C t_i> \alpha) \leq \E\Big[\exp(\frac{4\alpha}{L} \cdot \sum_{i=2}^C t_i)\Big]\cdot  \exp(-4\alpha^2/L) \leq \exp(-\Omega( \alpha^2/\tau)).  
\end{align*}
Applying the same argument with $t_i=-t_i$ proves the other inequality. Defining $z_i,t_i$ appropriately proves the other inequalities using the same proof. 
\end{proof}

%% file: external-info.tex
\section{Compressing  external marginal information }\label{step3b}

Here we prove \Cref{simulationtheorem2}. 
Set $\mext = \mext_{I}(p,f)+ KI$, for a large constant $K$ to be chosen later. 
By \Cref{thm:smoothing}, it is no loss of generality to assume that $p$ is $\beta $-smooth, with $\beta = 1/(K\log (C2^{5\mext/I}))$. Let $g_1,g_2$ be as in \Cref{eq:g1g2}.

Define:
\begin{align*}
    r^A_i(xm) &= \begin{cases}
\min\{j : d_j^A(xm) > 20  \beta+d^A_{i-1}(xm)\}& \text{if such $j$ exists,}\\
C & \text{otherwise.}
\end{cases}
\\
    r^B_i(ym)  &= \begin{cases}
\min\{j : d_j^B(ym) > 20  \beta+d^B_{i-1}(ym)\}&\text{if such $j$ exists,}\\
C & \text{otherwise.}
\end{cases}\\
\end{align*}
Note that $r_i^A(xm)$ is always odd, and $r_i^B(ym)$ is always even.

%
%

Because $p$ is a protocol, we have 
$$\E_{p(xy\vert m)}[d_j^A(xm)+ d_j^B(ym)] = d_j(m).$$

Let $\eps$ be a parameter such that $\eps \ll 2^{-5\mext/I}$. Now, we describe a protocol $\Gamma$ for computing $f(xy)$. Throughout this protocol, the parties will maintain a partial transcript $m_{<i}$. These partial transcripts may be inconsistent with each other, but we describe the protocol assuming that they are consistent with each other. In the analysis we shall show that the probability that the parties end up with inconsistent transcripts is negligible.
\begin{enumerate}
    \item The parties sample $m_0$ using the distribution $p(m_0)$. The parties also sample a uniformly random function $h:\mathbb{Z} \rightarrow \{1,2,\dotsc, \lceil 1/\epsilon \rceil\}$.
    \item Run the protocol $\psi$ from \Cref{lem:round-sim} with $u = p(m_1\vert m_0x), v = p(m_1\vert m_0), L = 5\mext$, error parameter $\epsilon$, to obtain functions $a,b$ and transcript $s$. Alice sets $m^A_1 = a(\psi(us))$, Bob sets $m^B_1 = b(\psi(vs))$. If $m^B_1 = \bot$, the protocol terminates. Bob sends a bit to Alice to indicate whether or not this occurs. The communication complexity of this step is $L + O(\log (1/\epsilon))$.
    \item Let $m_{\leq \ell}$ denote the part of the transcript sampled so far. Alice and Bob repeat the following steps until $m$ corresponds to an entire transcript. 
    \begin{enumerate}
        \item Both parties use shared randomness to sample a full transcript $\tilde m$ according to $p( m\vert m_{\leq \ell})$. They exchange the values of  $r^A_{\ell+1}(x \tilde m)$ and $r^B_{\ell+1}(y \tilde m)$ to determine 
        \[k = \min\{r^A_{\ell+1}(x \tilde m),r^B_{\ell+1}(y \tilde m)\}.\] 
        \item Alice privately samples a number $\zeta^A\in [0,1]$  and sends $1$ to Bob if 
        \[\zeta^A \leq \frac{1}{2}\cdot \prod_{\substack{i=\ell+2\\ i \text{ odd}}}^k\frac{p(\tilde m_i\vert x\tilde m_{< i})}{p(\tilde m_i\vert \tilde m_{< i})},\]
        and otherwise sends $0$.
        \item Bob privately samples a number $\zeta^B\in [0,1]$  and sends $1$ to Alice if 
        \[\zeta^B \leq \frac{1}{2}\cdot \prod_{\substack{i=\ell+2\\ i \text{ even}}}^k\frac{p(\tilde m_i\vert y\tilde m_{< i})}{p(\tilde m_i\vert \tilde m_{< i})},\]
        and otherwise sends $0$.
        \item If both players receive $1$ then, set $m_{\leq k} \leftarrow \tilde m_{\leq k}$.
    \end{enumerate}
    \item  If $\eta^A \leq g_1(x m )\cdot2^{-\lceil \log g_1(x m )\rceil}$, Alice sends $h(\lceil \log g_1(x m )\rceil)$ to Bob, and otherwise she sends $\bot$ to indicate that the protocol should be aborted. 
    \item If there is a unique integer $z$ 
    such that
    \begin{align*}
        |z + \log g_2(y m)| &\leq 3\mext/I ,\\
        h(z)&=h(\lceil \log g_1(x m )\rceil),\\ \eta^B &\leq g_2(y m)\cdot2^{z - 3\mext/I}, 
    \end{align*}
    Bob sends  $\mathsf{sign} \Big( \E_{q(xy\vert  m )}[(-1)^f]\Big) \in \{\pm 1\}$ to Alice. Otherwise, he sends $\bot$ to abort the protocol.
\end{enumerate}

To ensure the communication of the protocol is small, in our final protocol the parties abort and output a random bit if the communication in step 3 exceeds $(\mext\cdot 2^{15\mext/I}\cdot \log C)/\beta$. Then, the total communication is at most 
\[5\mext + \frac{\mext\cdot 2^{15\mext/I}\cdot \log C}{\beta} + O(\log 1/\eps) = O(\mext\cdot 2^{15\mext/I}\cdot \log^2 (C\cdot 2^{5\mext/I})) \leq \Delta\cdot I\log^2 C,\]
for some $\Delta$ that depends only on $\alpha$ since $\mext \leq (\alpha+K)I$.

Throughout the analysis below, we assume that in step 2, Alice always samples a message according to $u$, and Bob either accepts this sample or aborts, but never samples an inconsistent message. We can afford to make this assumption, because the probability of Bob sampling an inconsistent message without aborting is bounded by $\epsilon$, which will be much smaller than our final advantage. Moreover, if Alice and Bob sample consistently in step 2 then the transcript they end up with after step 3 must be the same.


Let $S$ and $R$ be the sets defined in \Cref{eq:mext-def-S,eq:mext-R} for our choice of $K$. In addition to $S$ and $R$, we need the following sets to analyze the simulating protocol: 
\begin{align*}
    Q &= \Big \{xym \eta^A \eta^B: \eta^A \leq g_1(xm) \cdot 2^{-\lceil \log g_1(xm)\rceil}, \eta^B \leq g_2(ym)\cdot 2^{\lceil \log g_1(xm)\rceil - 3\mext/I}  \Big\},\\
    \mathcal{Z} &= \Big\{xymh:\exists \text{ unique integer $z$ with }|z + \log g_2(ym )|\leq \frac{3\mext}{I} \text{ and }  h(z)= h(\lceil \log g_1(xm) \rceil)\Big\}.
\end{align*}

Let $\mathcal{G}$ denote the event that the protocol reaches the final step without aborting and having communicated at most $(\mext\cdot 2^{15\mext/I}\cdot \log C)/\beta$ bits in step 3. Define $\mathcal{A}(xym) \in \{\pm 1\}$ by 
$$ \mathcal{A}(xym) = \mathsf{sign}\Big(\E_{q(xy \vert  m)}[(-1)^{f(xy)}] \Big) \cdot (-1)^{f(xy)}.$$
Our protocol computes $f(xy)$ correctly when $\mathcal{G}$ happens and $\mathcal{A}(xym)=1$. 
The advantage of the protocol is at least
\begin{align} 
&\Gamma(\mathcal{Z}QS \mathcal{G}) \cdot  \E_{\Gamma}[\mathcal{A}(xym)\vert\mathcal{Z}QS\mathcal{G}] - \Gamma(\mathcal{G} (\mathcal{Z}QS)^c)\label{eq:mext-finaladvantage}
\end{align}

We shall prove:
\begin{align}
    \E_{\Gamma}[\mathcal{A}(xym)\vert\mathcal{Z}QS\mathcal{G}]&\geq \Omega(2^{-\delta \mext /(12I)}) ,\label{eq:mext-expectadbound}\\
    \Gamma(\calZ QS\mathcal{G}) & \geq \Omega(2^{-3 \mext/I}), \label{eq:mext-qszgbound}\\
    \Gamma(\mathcal{G}(\mathcal{Z}QS)^c) & \leq O(2^{-4\mext /I}). \label{eq:mext-negative}
\end{align}
By \Cref{eq:mext-finaladvantage}, since $\delta \leq 1$, we can choose $K$ to be large enough to  prove the theorem, since  $\alpha + K \geq \mext /I \geq K$. 

We first prove \Cref{eq:mext-negative}.
By the union bound, we have:
\begin{align*}
 \Gamma(\mathcal{G} (\mathcal{Z}QS)^c)
  & \leq  \Gamma( \mathcal{Z}^c \mathcal{G})+\Gamma(S^c  \mathcal{GZ})+  \Gamma(Q^c \mathcal{GZ}S).
\end{align*}

The definition of the protocol ensures that $\Gamma(\mathcal{Z}^c \mathcal{G})=0$. Moreover, $\Gamma(Q^c \mathcal{GZ}S) = 0$, because if the event $\mathcal{Z}S$ happens and the parties do not abort, then:
\begin{align*}
\eta^A &\leq g_1(xm)\cdot 2^{\lceil \log g_1(xm) \rceil} \text{ and}\\ 
\eta^B &\leq g_2(ym)\cdot 2^{z - 3\mext/I } = g_2(ym)\cdot 2^{\lceil \log g_1(xm) \rceil  - 3\mext/I }.
\end{align*} 
Additionally, $\Gamma(S^c\mathcal{GZ}) \leq \Gamma(S^c\mathcal{Z}) \leq O(\eps\marg/I)$, since if $S^c\mathcal{Z}$ happens then there must have been a hash collision, which happens with probability at most $O(\eps\marg/I)$. 

In order to prove \Cref{eq:mext-expectadbound,eq:mext-qszgbound}, we need to first establish that $\Gamma(xym)$ is typically quite close to $p(xym)$. Indeed, consider a particular execution of step 3 in the protocol. At this point, some prefix $m_{\leq \ell}$ has been fixed. For $m$ consistent with this prefix $m_{\leq \ell}$, define the frontier
\begin{align*}
    r(xym) & =    \min\{r_{\ell+1}^A(xm),r_{\ell+1}^B(ym)\}. 
\end{align*}
When the parties finally accept a sample, it will be a string $m_{\leq r(xym)}$ on the frontier. By the definition of $r_{\ell+1}^A,r_{\ell+1}^B$, and by \Cref{eqn:smoothness}, we have that for all $m$,  $d_{r(xym)}(m) - d_{\ell}(m) \leq 45 \beta$. Setting $\tau=45 \beta$ and $\alpha=1/4$, we apply \Cref{divconcentrate} to conclude that if 
\begin{align*}
F^A&=\Big \{xym: \prod_{j=\ell+1 \text{ odd}}^{r(xym)}p(m_{j}\vert xy m_{\leq \ell}) \geq 2 \cdot  \prod_{j=\ell+1 \text{ odd}}^{r(xym)}p(m_{\leq r(xym)} \vert m_{\leq \ell})\Big \},\\
F^B&=\Big \{xym: \prod_{j=\ell+1 \text{ even}}^{r(xym)}p(m_{j}\vert xy m_{\leq \ell}) \geq 2 \cdot  \prod_{j=\ell+1 \text{ even}}^{r(xym)}p(m_{\leq r(xym)} \vert m_{\leq \ell})\Big \}.
\end{align*}
then 
\begin{align}
p(F^A \cup F^B \vert xy m_{\leq \ell}) \leq 4 \exp(-\Omega(1/\beta))\leq C^{-1}\cdot 2^{-5\mext/I}.\label{fafb}
\end{align}
Now, we perform a standard analysis of rejection sampling. Let $W$ denote the event that the first sample of $m_{r(xym)}$ is accepted in the protocol. Given $xym_{\leq \ell}$, the probability that $W$ occurs is
\begin{align}
    \Gamma(W\vert xym_{\leq \ell}) &\geq  \sum_{m'_{r(xym')}: xym'_{r(xym)} \notin F^A \cup F^B } p(m'_{\leq r(xym')}\vert xy m_{\leq \ell})/4 \notag\\
    &\geq 1/4-p(F^A \cup F^B \vert xy m_{\leq \ell})/4 \geq 1/4 - C^{-1}\cdot 2^{-5\mext/I} \geq 1/8, \label{eq:Wbound}
\end{align}
where here we abused notation to write $xym'_{r(xym)} \notin F^A \cup F^B$ to mean that the prefix is not consistent with any $m$ in $F^A \cup F^B$.

It is clear that the sampled point is independent of the event $\neg W$, so it is also independent of $W$. So, the probability that a particular prefix $m_{r(xym)}$ is sampled is the same as the probability that it is sampled conditioned on $W$. When $m_{r(xym)}$ is not consistent with $F^A \cup F^B$, the probability of such a point is 
\begin{align}
\frac{p(m_{\leq r(xym)}\vert xy m_{\leq \ell})/4}{1/4 - p(F^A \cup F^B \vert xym_{\leq \ell})/4} = p(m_{\leq r(xym)}\vert xy m_{\leq \ell})\cdot (1\pm O(C^{-1}\cdot 2^{-5\mext/I})). \label{acceptplusminus}
\end{align}
Let $B$ denote the event that the final sample $xym$ is such that at some point a prefix was sampled in $F^A \cup F^B$ during step 3. Whenever step 3 accepts a sample, the length of the transcript increases by at least $1$, so the number of times step 3 accepts a sample is at most $C$. Thus, by the union bound and \Cref{fafb}, 
\begin{align}
p(B) \leq O(2^{-5\mext/I}).\label{pb}
\end{align}
Moreover, by \Cref{acceptplusminus}, for $xym \notin B$, 
\begin{align}
    \Gamma(xym) = p(xym) \cdot (1\pm O(2^{-5\mext/I})).\label{multpm}
\end{align}    
     \Cref{pb,multpm} imply  
\begin{align}\label{eq:Gamma-B-upperbound}
\Gamma(B) = 1-\Gamma(B^c) \leq 1-p(B^c)\cdot (1 - O(2^{-5\mext/I})) \leq O(2^{-5\mext/I}).
\end{align}
Additionally, we have  
\begin{align}\label{qS-B}
    q(SB^c) = q(S) - q(BS) \geq q(S) - 2^{3\mext/I}\cdot p(B) \geq 1 - \Omega(2^{-\marg/I}),
\end{align}
by the definition of $S$, \Cref{claim:q(S)-lb,pb}.

Now we can begin to understand $\Gamma(\calZ QS \mathcal{G})$. For $xym\in S$,
\begin{align}\label{qgiven}
    \Gamma(Q\vert xym) = \frac{g_1(xm) \cdot 2^{-\lceil \log g_1(xm)\rceil}\cdot g_2(ym)\cdot 2^{\lceil \log g_1(xm)\rceil}}{2^{3\mext/I}} 
    &= \frac{q(xym)}{p(xym)} \cdot \frac{1}{2^{3\mext/I}}, 
\end{align}
where the first equality follows from the fact that $$g_2(ym) \cdot 2^{\lceil \log g_1(xm)\rceil} = 2^{\lceil \log g_1(xm)\rceil + \log g_2(ym)} \leq 2^{3\marg/I},$$ by the definition of $S$.

We can bound 
\begin{align}
    \Gamma(QSB^c) &= \sum_{xym\in S\cap  B^c}\Gamma(xym)\cdot \Gamma(Q\vert xym) \notag \\ &= \sum_{xym\in S\cap  B^c}p(xym)\cdot \Gamma(Q\vert xym)\cdot (1 \pm O(2^{-5\mext/I})). \tag{by \Cref{multpm}}\\
    &= 2^{-3\mext/I}\cdot  q(S  B^c) \cdot (1 \pm O(2^{-5\mext/I})) \label{qsb}\\
    &=  \Omega(2^{-3\mext/I}), \label{qsb1}
\end{align}
by \Cref{pb,qS-B}. We claim that for all $xym\in SB^c$,
\begin{align}
    \Gamma(\calZ\vert xymQSB^c) = \Gamma(\mathcal{Z}\vert xym) \geq 1 - O(\eps\mext/I). \label{zgiven}
\end{align}
The equality follows by observing that $xym$ determine $SB^c$, and given $xym$, $\calZ$ just depends on the choice of $h$, which is independent of $Q$. The inequality follows from the fact that for each $xym$ in $S$,  
the event $\mathcal{Z}^c$ can happen only if there exists an integer $z$ distinct from $\lceil \log g_1(xm) \rceil$ such that $h(\lceil \log g_1(xm) \rceil)=h(z)$ and $\abs{z + \log g_2(ym)} \leq 3\mext/I$. The probability that this happens is at most $O(\epsilon\cdot \mext /I )$. In particular, this implies 
\begin{align}
    \Gamma(\calZ\vert QS) \geq 1 - O(\eps\mext/I).\label{zqsbound}
\end{align}

For $xym \in S \cap B^c $, we have
\begin{align}
    \Gamma(xym \vert \mathcal{Z}QS B^c ) &= \frac{\Gamma(xym)\cdot \Gamma( \mathcal{Z}QS B^c \vert xym )}{\Gamma(\mathcal{Z}QS B^c )}\notag \\
    &= \frac{p(xym)\cdot   \Gamma(Q\vert xym) \cdot \Gamma( \mathcal{Z}\vert xym )}{\Gamma(\mathcal{Z}QS B^c)} \cdot (1\pm O(2^{-5\mext/I}))\tag{by \Cref{multpm}, and since $xym$ determine $S,B^c$} \\
    &= \frac{p(xym)}{\Gamma(QSB^c)} \cdot \frac{q(xym)}{p(xym)\cdot 2^{3\mext/I}} \cdot \frac{\Gamma(\mathcal{Z}\vert xym)}{\Gamma(\mathcal{Z}\vert QSB^c)}\cdot (1\pm O(2^{-5\mext/I})) \tag{By \Cref{qgiven}}\\
    &=  \frac{q(xym)}{q(SB^c)} \cdot (1\pm O(\epsilon \mext /I + 2^{-5\mext/I})). \tag{By \Cref{zgiven,qsb}} \\
    &= q(xym\vert SB^c) \cdot (1\pm O(\epsilon \mext /I + 2^{-5\mext/I})). \label{finalmultsample}
\end{align}

To argue that the protocol does not have too much communication, we show that typically the divergence costs of the accepted transcripts are small. Define the sets
\begin{align*}
    H &= \{xym: \log \frac{p(m\vert xy)}{p(m)} \leq \mext\cdot 2^{10\mext/I} \},\\
    F &= \{xym:  d_C(xym) > 2\mext\cdot 2^{10\mext/I} \}
\end{align*}
and the frontier 
\begin{align*}
    r(xym) & = \begin{cases}
        \min\{i: d_i(xym)> 2\mext\cdot 2^{10\mext/I}\} & \text{if such $i$ exists,}\\
        C & \text{otherwise.}
    \end{cases}
\end{align*}

We have $F \cap H \subseteq F_{r,\mext\cdot 2^{10\mext/I}}$, where $F_{r,\mext\cdot 2^{10\mext/I}}$ is the set from \Cref{badfrontiers}. 
By \Cref{eqn:smoothness}, and the the choice of $r$, $d_{r(xym)}(xym) \leq 2\mext\cdot 2^{10\mext/I} +5 \beta$, so we can apply \Cref{divconcentrate} to conclude that
\begin{align}
    p(FH) \leq p(F_{r,\mext\cdot 2^{10\mext/I}}) \leq 2 \exp(-\Omega(\mext\cdot 2^{10\mext/I})).\label{fhbound}
\end{align}

We have
\begin{align*}
    q(F\vert SB^c) & \leq \frac{q(FS)}{q(SB^c)}\notag \\
    &\leq \frac{q(FHS) + q(H^c)}{q(SB^c)} \\
    &\leq O(q(FHS) + q(H^c))  \tag{by \Cref{qS-B}}\\
&\leq O(q(FHS) + 2^{-10\mext/I})  \tag{by Markov's inequality and \Cref{eq:mext-log-ratio-bound}} \\
&\leq O(p(FHS) \cdot 2^{3 \mext/I} + 2^{-10\mext/I})\tag{using the definition of $S$} \\
&\leq O(\exp(-\Omega(\mext\cdot 2^{10\mext/I})) \cdot 2^{3 \mext/I} + 2^{-10\mext/I}),\notag  
\end{align*}
by \Cref{fhbound}. Putting this bound back into \Cref{finalmultsample}, we get
\begin{align}
    \Gamma(F \vert \mathcal{Z}QS B^c ) 
    &\leq  O(2^{-10\mext/I}) \label{gammafbound}
\end{align}

We note that every time step 3 accepts a sample, the divergence cost of the transcript increases by $20\beta$, and in expectation, the number of rounds of rejection sampling involved to accept a sample is at most $8$ by \Cref{eq:Wbound} and a standard calculation. Moreover, in each round, the players communicate at most $2 + 2\log C$ bits to exchange two indices in $\{1,\ldots,C\}$. Hence, given $xy$ and a transcript $m$ the expected communication to sample $m$ is at most $16\cdot (1 +\log C)\cdot d_C(xym)/(20\beta)$. 
%
%
Recall that $\calG$ occurs when the protocol reaches the final step having communicated at most $(\mext\cdot 2^{15\mext/I}\cdot \log C)/\beta$. Thus, Markov's inequality implies that 
\begin{align}
    \Gamma(\calG \vert \calZ QS B^c F^c) = 1- O(2^{-5\mext/I}). \label{ggivenbound}
\end{align}
So, we can conclude that 
\begin{align*}
    \Gamma(\calZ QS\calG) & \geq \Gamma(\calZ QSB^cF^c \cal G)\\
    & \geq \Gamma(QSB^c)\cdot  \Gamma(\calZ \vert QSB^c) \cdot \Gamma(F^c \vert \calZ QSB^c) \cdot \Gamma(\calG \vert \calZ QSB^cF^c)\\
    & \geq \Omega(2^{-3\mext/I}), \tag{by \Cref{qsb1,zgiven,ggivenbound,gammafbound}}
\end{align*}
proving \Cref{eq:mext-qszgbound}. Observe that by \Cref{eq:mext-qszgbound,eq:Gamma-B-upperbound,gammafbound},
\begin{align}
    \Gamma(B\vert \calZ QS \calG) & \leq \frac{\Gamma(B)}{\Gamma(\calZ QS \calG)}  \leq O(2^{-2\mext/I}),\label{fcond}\\
    \Gamma(F\vert \calZ QS \calG) & \leq \frac{\Gamma(F \calZ QS B^c)+ \Gamma(B)}{\Gamma(\calZ QS \calG)}  \leq O(2^{-2\mext/I}).\label{bcond}
\end{align}
Moreover, we have
\begin{align}
    \E_{\Gamma}[\mathcal{A}(xym)\vert \calZ QSB^c F^c] &\leq \Gamma(\calG \vert \calZ QSB^cF^c)\cdot \E_{\Gamma}[\mathcal{A}(xym)\vert \calZ QS \calG B^c F^c] + \Gamma(\calG^c\vert \calZ QS B^c F^c) \notag\\
    & \leq \E_{\Gamma}[\mathcal{A}(xym)\vert \calZ QS \calG B^c F^c] + O(2^{-5\mext/I}), \label{eq:adv-drop-G}\\
     \E_{\Gamma}[\mathcal{A}(xym)\vert \calZ QSB^c]   &\leq \Gamma(F^c\vert \calZ QSB^c )\cdot \E_{\Gamma}[\mathcal{A}(xym)\vert \calZ QSB^c F^c] + \Gamma(F\vert \calZ QSB^c)\notag\\
     & \leq \E_{\Gamma}[\mathcal{A}(xym)\vert \calZ QSB^c F^c] + O(2^{-10\mext/I}),\label{eq:adv-drop-F}
\end{align}
by \Cref{gammafbound}.

We are now ready to prove \Cref{eq:mext-expectadbound}. We have
\begin{align*}
    &\E_{\Gamma}[\mathcal{A}(xym)\vert \calZ QS\calG]\\ & \geq \Gamma(B^cF^c\vert \calZ QS\calG) \cdot \E_{\Gamma}[\mathcal{A}(xym)\vert \calZ QS\calG B^c F^c] - \Gamma(B\vert \calZ QS\calG  )- \Gamma(F\vert \calZ QS\calG  )\\
    & \geq (1-O(2^{-2\mext/I}))\cdot  \E_{\Gamma}[\mathcal{A}(xym)\vert \calZ QS\calG B^c F^c] - \Omega(2^{-2\mext/I}) \tag{by \Cref{bcond,fcond}}\\
    & \geq (1/2)\cdot \E_{\Gamma}[\mathcal{A}(xym)\vert \calZ QS B^c F^c] - \Omega(2^{-5\mext/I}) - \Omega(2^{-2\mext/I}) \tag{by \Cref{eq:adv-drop-G}} \\
    & \geq (1/2)\cdot \E_{\Gamma}[\mathcal{A}(xym)\vert \calZ QSB^c] - \Omega(2^{-10\mext/I}) - \Omega(2^{-2\mext/I}) \tag{by \Cref{eq:adv-drop-F}} \\
    & \geq (1/4)\cdot \E_{q}[\mathcal{A}(xym)\vert SB^c]  - \Omega(2^{-2\mext/I})\tag{by \Cref{finalmultsample}}\\
     & \geq  (1/4)\cdot  \E_{q(xym)}[\mathcal{A}(xym)]- \Omega(2^{-\mext/I}) \tag{by \Cref{qS-B}}\\
    & = (1/4)\cdot \E_{q(xym)}\bigg[\mathsf{sign}\bigg(\E_{q(x'y'\vert m)}\Big[(-1)^{f}\Big]\bigg) \cdot (-1)^{f(xy)}\bigg] - \Omega(2^{-\mext/I})\\
    &= (1/4)\cdot \E_{q(m)}\bigg[\bigg|\E_{q(xy\vert m)}\Big[(-1)^{f(xy)}\Big]\bigg|\bigg] - \Omega(2^{-\mext/I}) \\
    & \geq \Omega(2^{-\delta \mext/(12I)}). \tag{by \Cref{eq:mext-adbound}}
    \end{align*}
    This concludes the proof of the theorem.

%% file: bounded-round.tex
\section{Compressing  bounded-round protocols} \label{step3c}
We prove the \Cref{thm:bounded-round-sim}  in this section. Let $p(xym)$ be a protocol distribution such that $p(xy) = \mu(xy)$. We have $m=(m_0,\dotsc, m_r)$, which is the transcript consisting of $r$ messages along with the shared randomness. By assumption, $\marg_I(p,f) = \alpha I$, and $m_r \in \{0,1\}$. We assume without loss of generality that $r$ is even. Let $q(xym)$ be a rectangular distribution that realizes $\marg_I(p,f)$. For a large constant $K$, let $\marg = \marg_I(p,f)+ KI$. Since $\marg(p,f) \geq 0$, we have $\marg \geq KI$. Let $g_1,g_2$ be as in \Cref{eq:g1g2}. Let $\eps$ be a parameter such that $\eps = (2^{4\marg/I}\cdot(r+1))^{-1}$. 

We define a protocol $\Gamma$ whose communication complexity is bounded by $$O(r\cdot (\marg + \log(r/\eps))).$$
Since $\marg_I(p,f) = \alpha I$ we get that $\marg \leq (\alpha+K)\cdot I$, and it follows that the communication is bounded by $\Delta r(I+ \log r)$ for some $\Delta$ that only depends on $\alpha$.
Now, we describe $\Gamma$.
\begin{enumerate}
    \item Jointly sample $p(m_0)$. Alice sets $m^A_0 = m_0$ and Bob sets $m^B_0 = m_0$. Jointly sample $\eta^A_0,\eta^A_1,\eta^B \in [0,1]$ uniformly and independently. Jointly sample a uniformly random function $h : \mathbb{Z} \to \{1,\ldots,\lceil 1/\epsilon\rceil \}$.
    \item For each $i\in \{1,\ldots,r-1\}:$
    \begin{enumerate}
        \item If $i$ is odd, run the protocol $\psi$ from \Cref{lem:round-sim} with $u = p(m_i\vert m^A_{< i}x)$, $v = p(m_i\vert m^B_{< i}y)$, $L = 14\marg + 5\log (r+1)$ and error parameter $\epsilon$, to obtain functions $a_i,b_i$ and transcript $s$. Alice sets $m^A_i = a_i(us)$, Bob sets $m^B_i = b_i(vs)$. If $m^B_i = \bot$, Bob signals to abort in the next round and sends a random bit to Alice, which they both output.
        \item If $i$ is even, run the protocol $\psi$ from \Cref{lem:round-sim} with $u = p(m_i\vert m^B_{< i}y)$, $v = p(m_i\vert m^A_{< i}x)$, $L = 14\marg + 5\log (r+1)$ and error parameter $\epsilon$, to obtain functions $a_i,b_i$ and transcript $s$. Bob sets $m^B_i = a_i(us)$, Alice sets $m^A_i = b_i(vs)$. If $m^A_i = \bot$, Alice signals to abort in the next round and sends a random bit to Bob, which they both output.
    \end{enumerate}
    Let $\langle m^A\rangle, \langle m^B\rangle$ denote the values of $m^A$ and $m^B$ after the first $r-1$ rounds. 
    \item For each $b\in \{0,1\}$, Alice sends a message to Bob. If $\eta^A_b \leq \log g_1(x\langle m^A\rangle b)\cdot 2^{-\lceil \log g_1(x\langle m^A\rangle b)\rceil}$, Alice sends $h(\lceil \log g_1(x\langle m^A\rangle b)\rceil)$ to Bob, otherwise she sends $0$.
    \item Bob samples a bit $b$ according to $p(m_r\vert \langle m^B\rangle y)$. If there is a unique  integer $z$ such that
    \begin{align*}
        |z + \log g_2(y\langle m^B\rangle b)| &\leq 3\marg/I ,\\
        h(z)&=h(\lceil \log g_1(x\langle m^A\rangle b)\rceil),\\ \eta^B &\leq g_2(y\langle m^B\rangle b)\cdot2^{z - 3\marg/I}, 
    \end{align*}
    Bob sends  $\mathsf{sign} \Big( \E_{q(xy\vert \langle m^B\rangle b)}[(-1)^f]\Big) \in \{\pm 1\}$ to Alice. Otherwise, he sends $\bot$ to abort the protocol.
\end{enumerate}

We note that the above protocol involves at most $r$ rounds of communication, and in each of the first $r-1$ rounds, the communication from step 2 is at most 
$$14\marg + 5\log (r+1) + O(\log 1/\eps) \leq O(\marg + \log (r/\eps)).$$ 
In step 3, Alice additionally sends $O(\log 1/\eps)$ bits for the hashes. Hence, the total communication is at most $O(r\cdot (\marg + \log (r/\eps)))$.

We may assume that at the beginning of $\Gamma$, the players sample $r$ independent random tapes, where the $i$-th random tape is used for the $i$-th execution of the protocol $\psi$ from \Cref{lem:round-sim} in step 2 of $\Gamma$. 
Given this assumption, define $m$ as follows: $m_0 = m^A_0 = m^B_0$, and for all  $i\geq 1$, $m_i = a_i(p(m_i\vert m_{< i}xy)s)$, where $s$ is a transcript of the protocol $\psi$ from \Cref{lem:round-sim} that is determined given $x,y,m_{< i}$ and the $i$-th random tape, and $a_i$ is the function promised by the lemma. From item 1 of \Cref{lem:round-sim}, it is clear that $\Gamma(xym) = p(xym)$.

Let $S$ be the set defined in \Cref{eq:def-S} for our choice of $K$. In addition to $S$, we need the following sets to analyze the simulating protocol.
\begin{align*}
    Q &= \Big \{xym \eta^A_{m_r} \eta^B: \eta^A_{m_r} \leq g_1(xm) \cdot 2^{-\lceil \log g_1(xm)\rceil}, \eta^B \leq g_2(ym)\cdot 2^{\lceil \log g_1(xm)\rceil - 3\marg/I}  \Big\},\\
    \mathcal{E} &= \{ \langle  m^A \rangle\langle m^B\rangle m_{< r} : \langle m^A\rangle= \langle m^B\rangle = m_{< r}\},\\
    \mathcal{Z} &= \{xymh:\exists \text{ a unique integer $z$ with }|z + \log g_2(ym )|\leq 3\marg/I \text{ and }  h(z)= h(\lceil \log g_1(xm) \rceil)\}.
\end{align*}

Let $\mathcal{G}$ denote the event that the protocol reaches the final step without aborting, and define $\mathcal{A}(xym) \in \{\pm 1\}$ by 
$$ \mathcal{A}(xym) = \mathsf{sign}\Big(\E_{q(xy \vert  m)}[(-1)^{f(xy)}] \Big) \cdot (-1)^{f(xy)}.$$
Our protocol computes $f(xy)$ correctly when: $\mathcal{G}$ happens, $\mathcal{A}(xym)=1$ and $m_{< r} = \langle m^B \rangle$. Since $\mathcal{EZ}SQ \subseteq \mathcal{G}$, and $\mathcal{E}$ implies $m_{< r} = \langle m^B \rangle$, the advantage of our protocol is at least:
\begin{align} 
&\Gamma(\mathcal{EZ}SQ) \cdot  \E_{\Gamma(xym|\mathcal{EZ}SQ)}[\mathcal{A}(xym)] - \Gamma(\mathcal{G} (\mathcal{EZ}SQ)^c). \label{eq:bounded-round-adv}
\end{align}

We shall prove each of the following bounds:
\begin{align}
    \E_{\Gamma(xym\vert \mathcal{EZ}QS)}[\mathcal{A}(xym)]&\geq \Omega(2^{-\delta \marg /(12I)}) ,\label{eq:bounded-round-expectadbound}\\
    \Gamma(\mathcal{EZ}QS) & \geq \Omega(2^{-3 \marg/I}), \label{eq:bounded-round-qszebound}\\
    \Gamma(\mathcal{G}(\mathcal{EZ}SQ)^c) & \leq O(2^{-4\marg /I}). \label{eq:bounded-round-negative}
\end{align}
Because $\delta \leq 1$ and  $(\alpha+K) \geq \marg/I \geq K$, we can choose $K$ to be large enough to prove the theorem. 

We first upper bound $\Gamma(\mathcal{G} (\mathcal{EZ}SQ)^c)$.  By the union bound, we have:
\begin{align*}
 \Gamma(\mathcal{G} (\mathcal{EZ}SQ)^c)
  & \leq  
   \Gamma(\mathcal{G}\mathcal{E}^c) +  \Gamma( \mathcal{Z}^c\vert \mathcal{GE})+\Gamma(S^c  \mathcal{GEZ})+  \Gamma(Q^c \vert \mathcal{G EZ}S).
\end{align*}

The definition of the protocol ensures that $\Gamma(\mathcal{Z}^c\vert \mathcal{GE})=0$. Moreover, we claim that $\Gamma(Q^c \vert \mathcal{G EZ}S) = 0$, because if the event $\mathcal{EZ}S$ happens and the parties do not abort, then
\begin{align*}
\eta^A_{m_r} &\leq g_1(x\langle m^A\rangle m_r)\cdot 2^{\lceil \log g_1(x\langle m^A\rangle m_r) \rceil} = g_1(xm)\cdot 2^{\lceil \log g_1(xm) \rceil},\\ 
\eta^B &\leq g_2(y\langle m^B\rangle m_r)\cdot 2^{z - 3\marg/I } = g_2(ym)\cdot 2^{\lceil \log g_1(xm) \rceil  - 3\marg/I }.
\end{align*} 
The event $\mathcal{G}\mathcal{E}^c$ implies that $\psi$ made an error in one of the $r$ rounds, leaving Alice and Bob with strings that were not equal. The probability that this happens is at most $\epsilon \cdot r \leq 2^{-4\marg/I}$, by our choice of $\epsilon$. Finally, $\Gamma(S^c\mathcal{GEZ}) \leq \Gamma(S^c\mathcal{EZ}) \leq O(\eps\marg/I)$, since if $S^c\mathcal{EZ}$ happens then there must have been a hash collision, which happens with probability at most $O(\eps\marg/I)$. This implies \Cref{eq:bounded-round-negative}.

Now, we turn to proving \Cref{eq:bounded-round-qszebound}. Let us first estimate $\Gamma(QS)$. We have,
\begin{align*}
    \Gamma(QS) &=\sum_{xym\in S}\Gamma(xym)\cdot \Gamma(Q\vert xym) =\sum_{xym\in S}p(xym)\cdot \Gamma(Q\vert xym). 
\end{align*}
For $xym\in S$,
\begin{align}\label{eq:bounded-round-Q|xym}
    \Gamma(Q\vert xym) = \frac{g_1(xm) \cdot 2^{-\lceil \log g_1(xm)\rceil}\cdot g_2(ym)\cdot 2^{\lceil \log g_1(xm)\rceil}}{2^{3\marg/I}} 
    &= \frac{q(xym)}{p(xym)} \cdot \frac{1}{2^{3\marg/I}},
\end{align}
where the first equality follows from the fact that $$g_2(ym) \cdot 2^{\lceil \log g_1(xm)\rceil} = 2^{\lceil \log g_1(xm)\rceil + \log g_2(ym)} \leq 2^{3\marg/I},$$ by the definition of $S$. Therefore,
\begin{align}\label{eq:bounded-round-qs}
    \Gamma(QS) =\sum_{xym\in S}p(xym)\cdot \frac{q(xym)}{p(xym)} \cdot \frac{1}{2^{3\marg/I}} &= \frac{q(S)}{2^{3\marg/I}} \notag\\
    &\geq \frac{(1-5\cdot 2^{-\marg/I})}{2^{3\marg/I}}= \Omega(2^{-3\marg/I}),
\end{align}
where in the last line, we used \Cref{claim:q(S)-lb}.

We claim that for all $xym\in S$,
\begin{align}\label{eq:bounded-round-Z|xymQS}
    \Gamma(\calZ\vert xymQS) = \Gamma(\mathcal{Z}\vert xym) \geq 1 - O(\eps\marg/I).
\end{align}
The  equality follows by noting that $xym$ determine $S$ and given $xym$, $\calZ$ just depends on the choice of $h$, which is independent of $Q$. The event $\mathcal{Z}^c$ can happen only if there exists an integer $z$ distinct from $\lceil \log g_1(xm) \rceil$ such that $h(\lceil g_1(xm) \rceil)=h(z)$ and $\abs{z + \log g_2(ym)} \leq 3\marg/I$. The probability that this happens is at most $O(\epsilon\cdot \marg /I )$. Therefore, $\Gamma(\mathcal{Z}\vert xym) \geq 1- O(\epsilon \marg / I)\geq 1/2$, by our choice of $\epsilon$.
We conclude that
\begin{align}
 \Gamma(QS\mathcal{Z}) &= \Gamma(QS) \cdot \Gamma (\mathcal{Z}\vert QS) \geq \Omega(2^{-3 \marg/I}), \label{eq:bounded-round-zqsbound}
\end{align} 

For all $xym \in S$,
\begin{align}
    \Gamma(xym \vert QS\mathcal{Z}) &= \frac{\Gamma(xym)\cdot \Gamma(QS \mathcal{Z}\vert xym )}{\Gamma(Q S \mathcal{Z})}\notag \\
    &= \frac{p(xym)}{\Gamma(Q S) }  \cdot \Gamma(Q\vert xym) \cdot \frac{\Gamma(\mathcal{Z}\vert xymQS)}{\Gamma(\mathcal{Z}|QS)} \notag \\
    &= \frac{p(xym)}{\Gamma(Q S) } \cdot \frac{q(xym)}{p(xym)\cdot 2^{3\marg/I}} \cdot \frac{\Gamma(\mathcal{Z}\vert xymQS)}{\Gamma(\mathcal{Z}|QS)} \tag{By \Cref{eq:bounded-round-Q|xym}}\notag \\
    &=  \frac{q(xym)}{q(S)} \cdot \frac{\Gamma(\mathcal{Z}\vert xymQS)}{\Gamma(\mathcal{Z}|QS)} \tag{By \Cref{eq:bounded-round-qs}} \\
    &=  q(xym\vert S) \cdot (1\pm O(\epsilon \marg /I)) \label{eq:bounded-round-Gamma|QSZ=q|S},
\end{align}
where the last line follows from \Cref{eq:bounded-round-Z|xymQS}.

Given \Cref{eq:bounded-round-zqsbound}, to complete the proof of \Cref{eq:bounded-round-qszebound}, it will be enough to prove that $\Gamma(\calE \vert QS\mathcal{Z}) \geq 1/2$. Let $T$ be the set $T_K$ defined in \Cref{claim:R-dense-in-S} for our choice of $K$. We have 
\begin{align}
    \Gamma(\mathcal{E}^c\vert QS\mathcal{Z}) &\leq \Gamma(T^c\vert QSZ) + \Gamma(\mathcal{E}^c\vert QS\mathcal{Z} T) \notag \\
    &\leq q(T^c\vert S)\cdot (1+ O(\eps \marg/I)) + \Gamma(\mathcal{E}^c\vert QS\mathcal{Z} T) \tag{By \Cref{eq:bounded-round-Gamma|QSZ=q|S}}\\
    &\leq O(2^{-\marg/I}) + \Gamma(\mathcal{E}^c\vert QS\mathcal{Z} T) \tag{By \Cref{claim:R-dense-in-S}}\\
    &\leq O(2^{-\marg/I}) + 2\eps\cdot r  \leq O(2^{-\marg/I})\leq 1/2\label{eq:bounded-round-ecomplement-bound}
\end{align}
where in the last line, we used the fact that given $QS\calZ T$, item 2 and 3 of \Cref{lem:round-sim} guarantee that $\calE^c$ can only happen with probability at most $2\eps$ in each of the $r$ rounds. \Cref{eq:bounded-round-ecomplement-bound,eq:bounded-round-zqsbound} together prove \Cref{eq:bounded-round-qszebound}. 

Next, we prove \Cref{eq:bounded-round-expectadbound}. Since $|\mathcal{A}(xym)| \leq 1$, we have
\begin{align*}
     \E_{\Gamma(xym\vert QS\mathcal{Z})}[\mathcal{A}(xym)] \leq  \Gamma(\mathcal{E}\vert QS\mathcal{Z}) \cdot  \E_{\Gamma(xym\vert QS\mathcal{ZE})}[\mathcal{A}(xym)] + \Gamma(\mathcal{E}^c\vert QS\mathcal{Z}),
\end{align*}
and since $\Gamma(\mathcal{E}\vert QS\mathcal{Z}) \leq 1$, this gives
\begin{align*}
    \E_{\Gamma(xym\vert QS\mathcal{ZE})}[\mathcal{A}(xym)] & \geq  \E_{\Gamma(xym\vert QS\mathcal{Z})}[\mathcal{A}(xym)]- \Gamma(\mathcal{E}^c\vert QS\mathcal{Z})\\
    & \geq  \E_{q(xym\vert S)}[\mathcal{A}(xym)]- \Omega(\epsilon \marg/I + 2^{-\marg/I}) \tag{using \Cref{eq:bounded-round-ecomplement-bound,eq:bounded-round-Gamma|QSZ=q|S}}\\
    & \geq  \E_{q(xym)}[\mathcal{A}(xym)]- \Omega(2^{-\marg/I}) \tag{by \Cref{claim:q(S)-lb}}\\
    & = \E_{q(xym)}\bigg[\mathsf{sign}\bigg(\E_{q(x'y'\vert m)}\Big[(-1)^{f}\Big]\bigg) \cdot (-1)^{f(xy)}\bigg] - \Omega(2^{-\marg/I})\\
    &= \E_{q(m)}\bigg[\bigg|\E_{q(xy\vert m)}\Big[(-1)^{f(xy)}\Big]\bigg|\bigg] - \Omega(2^{-\marg/I}) \geq \Omega(2^{-\delta \marg/(12I)}),
\end{align*}
by \Cref{eqn:expectadbound}. This completes the proof of \Cref{eq:bounded-round-expectadbound}.

%% file: braverman-sim.tex
\section{Compression independent of communication}\label{sec:braverman-sim}
In this section, we prove \Cref{thm:braverman-sim}. Let $K$ be a sufficiently large constant to be determined later.
Let $p(xym)$ be a protocol distribution such that $p(xy) = \mu(xy)$ and $\marg_I(p,f) \leq \alpha I$. 
Let $q(xym)$ be a rectangular distribution that realizes $\marg_I(p,f)$. 

Define $\marg = \marg_I(p,f)+ KI$. Since $\marg(p,f) \geq 0$, we have $\marg \geq KI$. Let $g_1,g_2$ be as in \Cref{eq:g1g2}. Let $\eps$ be a parameter such that $\eps = 2^{-6\marg/I - 8\marg}$. We define a protocol $\Gamma$ whose communication complexity is bounded by $$2\log 1/\eps \leq O(6\marg/I + 8\marg) = \Delta I,$$
for some $\Delta$ that depends only on $\alpha$.

We describe the protocol $\Gamma$.
\begin{enumerate}
    \item Jointly sample $\eta^A,\eta^B \in [0,1]$ uniformly. Jointly sample two uniformly random functions $h,t : \mathbb{Z} \to \{1,\ldots,\lceil 1/\epsilon\rceil \}$.
    \item Jointly sample an infinite sequence of triples $(m^1,\rho_A^{1},\rho_B^{1}),(m^2,\rho_A^{2},\rho_B^{2}),\ldots,$ where $m^{i}$ is sampled uniformly at random from the set of all transcripts and $\rho_A^{i},\rho_B^{i}$ are sampled uniformly at random in $[0,1]$.
    \item Alice finds the first index $i_A$ such that 
    \begin{align*}
        \rho_A^{i_A} &\leq \prod_{j \text{ odd}}p(m^{i_A}_j\vert xm^{i_A}_{< j}) , \\
        \rho_B^{i_A} &\leq 2^{6\marg}\cdot \prod_{j \text{ even}}p(m^{i_A}_j\vert xm^{i_A}_{< j}).
    \end{align*}
    Alice checks if $\eta^A \leq g_1(xm^{i_A})\cdot 2^{\lceil \log g_1(xm^{i_A})\rceil}$, in which case she sends $t(i_A)$ and $h(\lceil \log g_1(x m^{i_A})\rceil)$ to Bob. Otherwise, she sends $\bot$ signaling to abort. 
    \item Bob finds the first index $i_B$ such that 
    \begin{align*}
        \rho_A^{i_B} &\leq 2^{6\marg}\cdot \prod_{j \text{ odd}}p(m^{i_B}_j\vert ym^{i_B}_{< j}), \\
        \rho_B^{i_B} &\leq \prod_{j \text{ even}}p(m^{i_B}_j\vert ym^{i_B}_{< j}).
    \end{align*}
    If $t(i_B) = t(i_A)$, he checks if there is a unique integer $z$ such that
    \begin{align*}
        |z + \log g_2(ym^{i_B})| &\leq 3\marg/I ,\\
        h(z)&=h(\lceil \log g_1(xm^{i_A})\rceil),\\ 
        \eta^B &\leq g_2(ym^{i_B})\cdot2^{z - 3\marg/I}, 
    \end{align*}
    If all these conditions are satisfied, he sends $\mathsf{sign} \Big( \E_{q(xy\vert  m^{i_B})}[(-1)^f]\Big) \in \{\pm 1\}$ to Alice. Otherwise, he sends $\bot$ to abort the protocol.
    \end{enumerate}
    The protocol has the feature that Alice sends at most $2\log 1/\eps$ bits to Bob.  Let $i_*$ be the smallest index such that 
    \[ \prod_{j \text{ odd}}p(m^{i_*}_j\vert xm^{i_*}_{< j}) \geq \rho^{i_*}_A \text{ and } \prod_{j \text{ even}}p(m^{i_*}_j\vert xm^{i_*}_{< j}) \geq \rho^{i_*}_B.\]
    Let $m = m^{i_*}$.
    We note that $\Gamma(xym) = p(xym)$. 

    Let $S$ and $T$ be the sets defined in \Cref{eq:def-S} and \Cref{eq:def-R-b} respectively for our choice of $K$. In addition to $S$, we need the following sets to analyze the simulating protocol: 
    \begin{align*}
    Q &= \Big \{xym \eta^A \eta^B: \eta^A \leq g_1(xm) \cdot 2^{-\lceil \log g_1(xm)\rceil}, \eta^B \leq g_2(ym)\cdot 2^{\lceil \log g_1(xm)\rceil - 3\marg/I}  \Big\},\\
    \mathcal{E} &= \Big  \{ i_A i_B i_* : i_A= i_B = i_*  \Big \},\\
    \mathcal{Z} &= \Big\{xymh:\exists \text{ a unique integer $z$ with }|z + \log g_2(ym )|\leq \frac{3\marg}{I} \text{ and }  h(z)= h(\lceil \log g_1(xm) \rceil)\Big\}.
    \end{align*}

    Let $\mathcal{G}$ denote the event that the protocol reaches the final step without aborting, and define $\mathcal{A}(xym) \in \{\pm 1\}$ by 
    $$ \mathcal{A}(xym) = \mathsf{sign}\Big(\E_{q(xy \vert  m)}[(-1)^{f(xy)}] \Big) \cdot (-1)^{f(xy)}.$$
    Our protocol computes $f(xy)$ correctly when: $\mathcal{G}$ happens, $\mathcal{A}(xym)=1$ and $m = m^{i_B}$. Since $\mathcal{EZ}SQ \subseteq \mathcal{G}$, and $\mathcal{E}$ implies $m=m^{i_B}$, the advantage of our protocol is at least:
    \begin{align} 
    &\Gamma(\mathcal{EZ}SQ) \cdot  \E_{\Gamma(xym|\mathcal{EZ}SQ)}[\mathcal{A}(xym)] - \Gamma(\mathcal{G} (\mathcal{EZ}SQ)^c). \label{finaladvantage-b}
    \end{align}

    We shall prove:
    \begin{align}
    \E_{\Gamma(xym\vert \mathcal{EZ}SQ)}[\mathcal{A}(xym)]&\geq \Omega(2^{-\delta \marg /(12I)}) ,\label{eqn:expectadbound-b}\\
    \Gamma(\mathcal{EZ}SQ) & \geq \Omega(2^{-6\marg/I - 6\marg}), \label{eqn:qszebound-b}\\
    \Gamma(\mathcal{G}(\mathcal{EZ}SQ)^c) & \leq O(2^{-6\marg/I - 7\marg}). \label{eqn:negative-b}
    \end{align}
    By \Cref{finaladvantage-b}, since $\delta \leq 1$, we can choose $K$ to be large enough to  prove the theorem, since  $\alpha+K \geq \marg/I \geq K$. 

    We first upper bound $\Gamma(\mathcal{G} (\mathcal{EZ}SQ)^c)$.  By the union bound, we have:
    \begin{align*}
    \Gamma(\mathcal{G} (\mathcal{EZ}SQ)^c)
    & \leq  
    \Gamma(\mathcal{G}\mathcal{E}^c) +  \Gamma(\mathcal{Z}^c\vert \mathcal{GE})+\Gamma(S^c  \mathcal{GEZ})+  \Gamma(Q^c \vert \mathcal{G EZ}S).
    \end{align*}

    The definition of the protocol ensures that $\Gamma(\mathcal{Z}^c\vert \mathcal{GE})=0$. Moreover, we claim that $\Gamma(Q^c \vert \mathcal{G EZ}S) = 0$, because if the event $\mathcal{EZ}S$ happens and the parties do not abort, then:
    \begin{align*}
    \eta^A &\leq g_1(xm^{i_A})\cdot 2^{\lceil \log g_1(xm^{i_A}) \rceil} = g_1(xm)\cdot 2^{\lceil \log g_1(xm) \rceil},\\ 
    \eta^B &\leq g_2(ym^{i_B})\cdot 2^{z - 3\marg/I } = g_2(ym)\cdot 2^{\lceil \log g_1(xm) \rceil  - 3\marg/I }.
    \end{align*} 
    The event $\mathcal{G}\mathcal{E}^c$ implies that there was a hash error for the triples accepted by Alice and Bob. The probability of this happening is at most $\eps$. Finally, $\Gamma(S^c\mathcal{GEZ}) \leq \Gamma(S^c\mathcal{EZ}) \leq O(\eps\marg/I)$, since if $S^c\mathcal{EZ}$ happens then there must have been a hash collision, which happens with also occurs with probability at most $2\eps$. By our choice of $\eps$, the total error is bounded by $2^{-6\marg/I - 8\marg}(2 + \marg/I) \leq 2^{-6\marg/I - 7\marg}$, for $K$ sufficiently large. This implies \Cref{eqn:negative-b}.

    Let us estimate $\Gamma(QS)$. We have,
    \begin{align*}
    \Gamma(QS) &=\sum_{xym\in S}\Gamma(xym)\cdot \Gamma(Q\vert xym) =\sum_{xym\in S}p(xym)\cdot \Gamma(Q\vert xym). 
    \end{align*}
    For $xym\in S$,
    \begin{align}\label{eq:final-reject}
    \Gamma(Q\vert xym) = \frac{g_1(xm) \cdot 2^{-\lceil \log g_1(xm)\rceil}\cdot g_2(ym)\cdot 2^{\lceil \log g_1(xm)\rceil}}{2^{3\marg/I}} 
    &= \frac{q(xym)}{p(xym)} \cdot \frac{1}{2^{3\marg/I}},
    \end{align}
    where the first equality follows from the fact that $$g_2(ym) \cdot 2^{\lceil \log g_1(xm)\rceil} = 2^{\lceil \log g_1(xm)\rceil + \log g_2(ym)} \leq 2^{3\marg/I},$$ by the definition of $S$. Therefore,
    \begin{align}\label{eq:qs-b}
    \Gamma(QS) =\sum_{xym\in S}p(xym)\cdot \frac{q(xym)}{p(xym)} \cdot \frac{1}{2^{3\marg/I}} &= \frac{q(S)}{2^{3\marg/I}} \notag\\
    &\geq \frac{(1-5\cdot 2^{-\marg/I})}{2^{3\marg/I}}= \Omega(2^{-3\marg/I}),
    \end{align}
    where in the last line, we used \Cref{claim:q(S)-lb}.

    We claim that for all $xym\in S$,
    \begin{align}\label{eq:Z|xymQS-braverman}
        \Gamma(\calZ\vert xymQS) = \Gamma(\mathcal{Z}\vert xym) \geq 1 - O(\eps\marg/I).
    \end{align}
    The equality follows by observing that $xym$ determine $S$ and given $xym$, $\calZ$ just depends on the choice of $h$, which is independent of $Q$. The event $\mathcal{Z}^c$ can happen only if there exists an integer $z$ distinct from $\lceil \log g_1(xm) \rceil$ such that $h(\lceil g_1(xm) \rceil)=h(z)$ and $\abs{z + \log g_2(ym)} \leq 3\marg/I$. The probability that this happens is at most $O(\epsilon\cdot \marg /I )$. Therefore, $\Gamma(\mathcal{Z}\vert xym) \geq 1- O(\epsilon \marg / I)\geq 1/2$, by our choice of $\epsilon$.
    We conclude that
    \begin{align}
        \Gamma(QS\mathcal{Z}) &= \Gamma(QS) \cdot \Gamma (\mathcal{Z}\vert QS) \geq \Omega(2^{-3 \marg/I}), \label{eqn:zqsbound-b}
    \end{align}
    
    Let $W$ be the event that $\min\{i_A,i_B,i_*\}= 1$ and let $T$ be the set defined in \Cref{eq:def-R-b} for our choice of $K$.
    We claim that $TW^c$ implies $i_*>1$, since if $xym\in T$ then $p(m\vert xy) \leq 2^{6\marg}\cdot \min \{p(m\vert x),p(m\vert y)\}$,
    which implies 
    \begin{align*}
    \prod_{j \text{ even}} p(m_j\vert ym_{< j}) &\leq 2^{6\marg}\cdot \prod_{j \text{ even}} p(m_j \vert xm_{< j}),\\
    \prod_{j \text{ odd}} p(m_j\vert xm_{< j}) &\leq 2^{6\marg}\cdot \prod_{j \text{ even}} p(m_j \vert ym_{< j}),
    \end{align*}
    and hence if $i_*=1$ then in fact $i_A = i_B = 1$.
    
    Now, we compute $\Gamma(\calE\vert QS\calZ)$.
    \begin{align}
        \Gamma(\calE\vert QS\calZ ) &= \Gamma(\calE \vert QS\calZ W) \notag \\
        &\geq \frac{\Gamma(i_A = i_B = i_* =1 \vert QS\calZ )}{\Gamma(i_A = 1\vert QS\calZ ) + \Gamma(i_B = 1\vert QS\calZ) + \Gamma(i_* = 1\vert QS\calZ)} \notag\\
        &\geq \frac{\Gamma(i_A = i_B = i_* = 1, QS\calZ )}{\Gamma(i_A = 1) + \Gamma(i_B = 1) + \Gamma(i_* = 1)}.\label{eq:firstindex}
    \end{align}

    Now, we estimate the numerator and denominator in the last expression. Let $\calM$ be the set of all transcripts in the support of $p$. We have,
    \begin{align*}
        &\Gamma(i_A = i_B= i_* =1, QS\calZ ) \\
        &\geq \sum_{xym\in S\cap T}\Gamma(i_A = i_B=i_* =1,xym, Q\calZ)\\
        &= \sum_{xym\in S\cap T}\Gamma(i_A = i_B=i_* =1,xym)\cdot \Gamma(Q\calZ\vert xym)\tag{given $xym$, $Q\calZ$ is independent of $i_A,i_B,i_*$}\\
        &= \sum_{xym\in S\cap T}p(xym)\cdot \frac{1}{|\calM|} \cdot \Gamma(Q\calZ\vert xym) \tag{by the definition $\Gamma$ and $T$}\\
        &= \sum_{xym\in S\cap T}p(xym)\cdot \frac{1}{|\calM|}\cdot \frac{q(xym)}{p(xym)2^{3\marg/I}}\cdot \Gamma(\calZ\vert xym) \tag{by \Cref{eq:final-reject}}\\
        &\geq q(ST)\cdot \frac{1}{|\calM|\cdot 2^{3\marg/I}}\cdot (1 - \Omega(\eps\marg/I)). \tag{by \Cref{eq:Z|xymQS-braverman}}
    \end{align*}
    
    Next,
    \begin{align*}
        \Gamma(i_A = 1) = \sum_{xm'}\Gamma(i_A =1,xm') 
        &\leq \sum_{xm'}p(x)\cdot\frac{1}{|\calM|}\cdot \prod_{j \text{ odd}}p(m'_j\vert xm'_{< j})\cdot 2^{6\marg}\cdot \prod_{j \text{ even}}p(m'_j\vert xm'_{< j})\\
        &\leq \sum_{xm'}p(xm')\cdot \frac{2^{6\marg}}{|\calM|}
        \leq \frac{2^{6\marg}}{|\calM|}.
    \end{align*}
    An identical calculation shows that $\Gamma(i_B = 1) \leq 2^{6\marg}/|\calM|$. Furthermore, 
    \begin{align*}
        \Gamma(i_* = 1) = \sum_{xym}\Gamma(i_* =1,xym)
        &= \sum_{xym}p(x)\cdot\frac{1}{|\calM|}\cdot \prod_{j \text{ odd}}p(m_j\vert xm_{< j})\cdot \prod_{j \text{ even}}p(m_j\vert ym_{< j})\\
        &\leq \sum_{xym}p(xym)\cdot \frac{1}{|\calM|}
        \leq \frac{1}{|\calM|}.
    \end{align*} 
    Plugging this into \Cref{eq:firstindex} we get
    \begin{align*}
    \Gamma(\calE\vert QS\calZ) &\geq \frac{q(ST)\cdot (1 - \Omega(\eps\marg/I))}{2^{6\marg + (3\marg/I) +2}} = \Omega(2^{-6\marg - (3\marg/I)}),
    \end{align*}
    by \Cref{claim:q(S)-lb} and \Cref{claim:R-bound-for-Braverman}.
    Using \Cref{eqn:zqsbound-b} we get that $\Gamma(QS\calZ\calE) = \Omega(2^{-6\marg - (6\marg/I)})$ as claimed in \Cref{eqn:qszebound-b}.

    For all $xym \in S$,
    \begin{align}
    \Gamma(xym \vert QS\mathcal{Z}) &= \frac{\Gamma(xym)\cdot \Gamma(QS \mathcal{Z}\vert xym )}{\Gamma(Q S \mathcal{Z})}\notag \\
    &= \frac{p(xym)}{\Gamma(Q S) }  \cdot \Gamma(Q\vert xym) \cdot \frac{\Gamma(\mathcal{Z}\vert xymQS)}{\Gamma(\mathcal{Z}|QS)} \notag \\
    &= \frac{p(xym)}{\Gamma(Q S) } \cdot \frac{q(xym)}{p(xym)\cdot 2^{3\marg/I}} \cdot \frac{\Gamma(\mathcal{Z}\vert xymQS)}{\Gamma(\mathcal{Z}|QS)} \tag{By \Cref{eq:final-reject}}\notag \\
    &=  \frac{q(xym)}{q(S)} \cdot \frac{\Gamma(\mathcal{Z}\vert xymQS)}{\Gamma(\mathcal{Z}|QS)}  \tag{By \Cref{eq:qs-b}} \\
    &=  q(xym\vert S) \cdot (1\pm O(\epsilon \marg /I)) ,\label{eq:Gamma|QSZ=q|S-b}
    \end{align}
    where the last line follows by \Cref{eq:Z|xymQS-braverman}.
    
    Next, we note that
    \begin{align}\label{eq:E-happens-if-T-happens|QSZi*=1}
        \Gamma(\calE\vert QS\calZ,i_*=1) \geq \Gamma(\calE, T\vert QS\calZ, i_*=1) = \Gamma(T\vert QS\calZ), 
    \end{align}
    where we used the fact that the event $T,i_*=1$ implies $\calE$ and that $xym$ is distributed independently of $i_*$. For any $xym\in S\cap T$
    \begin{align*}
        \Gamma(xym\vert QS\calZ\calE) &= \Gamma(xym\vert QS\calZ\calE W) \tag{$xym$ is independent of $W$ even conditioned on $QS\calZ\calE$}\\ 
        &= \Gamma(xym\vert QS\calZ\calE , i_*=1) \tag{the event $\calE W$ is the same as the event $\calE,i_*=1$}\\
        &= \frac{\Gamma(xym\calE\vert QS\calZ,i_*=1)}{\Gamma(\calE\vert QS\calZ i^*=1)}\\
        &= \Gamma(xym\vert QS\calZ)\cdot \frac{\Gamma(\calE\vert xym,i_*=1)}{\Gamma(\calE\vert QS\calZ i^*=1)} \\
        &= \frac{\Gamma(xym \vert QS\calZ)}{\Gamma(\calE\vert QS\calZ i^*=1)} \tag{because $xym\in S \cap T$}\\
        &= \Gamma(xym \vert QS\calZ)\cdot (1 \pm O(\Gamma(T^c\vert QS\calZ)))
    \end{align*}
    where the last inequality used the fact that $1 \geq \Gamma(\calE\vert QS\calZ i^*=1) \geq 1 - \Gamma(T^c\vert QS\calZ)$ by \Cref{eq:E-happens-if-T-happens|QSZi*=1}. Together with \Cref{eq:Gamma|QSZ=q|S-b} we get that for any $xym\in S\cap T$
    \begin{align}
        \Gamma(xym\vert QS\calZ\calE) &= q(xym\vert S)\cdot (1 \pm O(\Gamma(T^c\vert QS\calZ) + \eps\marg/I)) \notag\\
        &= q(xym\vert S)\cdot (1 \pm O(q(T^c\vert S) + \eps\marg/I))\notag\\
        &= q(xym\vert S)\cdot (1 \pm O(2^{-\marg/I} + \eps\marg/I)), \label{eq:Gamma|SQZET=q|ST}
    \end{align}
    where the last line follows by \Cref{claim:R-bound-for-Braverman}.
    Now, we complete the proof of \Cref{eqn:expectadbound-b}. We have
    \begin{align*}
     &\E_{\Gamma(xym\vert QS\mathcal{ZE})}[\mathcal{A}(xym)] \\
     &\geq   \sum_{xym\in S\cap T}\Gamma(xym\vert QS\mathcal{ZE})\cdot \mathcal{A}(xym) - \Gamma(T^c\vert QS\mathcal{ZE}) \\
     &\geq  \sum_{xym\in S\cap T}q(xym\vert S) \cdot \mathcal{A}(xym) - \Omega(2^{-\marg/I} + \eps\marg/I) - 1 + \Gamma(T\vert QS\mathcal{ZE}) \tag{by \Cref{eq:Gamma|SQZET=q|ST}}\\ 
     &\geq \E_{q(xym\vert S)}[\mathcal{A}(xym)] - q(T^c\vert S) - \Omega(2^{-\marg/I} + \eps\marg/I) - 1 + q(T\vert S)\cdot (1 - O(2^{-\marg/I} + \eps\marg/I)) \tag{by \Cref{eq:Gamma|SQZET=q|ST}}\\
     &\geq \E_{q(xym)}[\mathcal{A}(xym)] - q(S^c) - 2q(T^c\vert S) - \Omega(2^{-\marg/I} + \eps\marg/I) \\
     &= \E_{q(xym)}[\mathcal{A}(xym)] - \Omega(2^{-\marg/I} + \eps\marg/I) \tag{by \Cref{claim:q(S)-lb} and \Cref{claim:R-bound-for-Braverman}}\\
     &= \E_{q(xym)}\bigg[\mathsf{sign}\bigg(\E_{q(x'y'\vert m)}\Big[(-1)^{f}\Big]\bigg) \cdot (-1)^{f(xy)}\bigg] - \Omega(2^{-\marg/I})\\
     &= \E_{q(m)}\bigg[\bigg|\E_{q(xy\vert m)}\Big[(-1)^{f(xy)}\Big]\bigg|\bigg] - \Omega(2^{-\marg/I}) \geq \Omega(2^{-\delta \marg/(12I)}),
    \end{align*}
    by \Cref{eqn:expectadbound}. 